\documentclass[a4paper,UKenglish,cleveref,thm-restate]{lipics-v2021}
\hideLIPIcs  
\usepackage[linesnumbered,ruled,noline,noend,algonl]{algorithm2e}
\usepackage{booktabs}

\newtheorem{assumption}[theorem]{Assumption}
\newtheorem{fact}[theorem]{Fact}

\Crefname{equation}{Equation}{Equations}
\Crefname{assumption}{Assumption}{Assumptions}
\Crefname{fact}{Fact}{Facts}

\makeatletter
\@ifpackageloaded{algorithm2e}{%
\PackageInfo{cleveref}{`algorithm2e' support loaded}%
\crefalias{algocf}{algorithm}%
\crefalias{algocfline}{line}%
\crefalias{AlgoLine}{line}%
\let\cref@old@algocf@nl@sethref\algocf@nl@sethref%
\renewcommand{\algocf@nl@sethref}[1]{%
  \cref@old@algocf@nl@sethref{#1}%
  \cref@constructprefix{AlgoLine}{\cref@result}%
  \@ifundefined{cref@AlgoLine@alias}%
    {\def\@tempa{AlgoLine}}%
    {\def\@tempa{\csname cref@AlgoLine@alias\endcsname}}%
  \xdef\cref@currentlabel{%
    [\@tempa][\arabic{AlgoLine}][\cref@result]%
    \csname p@AlgoLine\endcsname\csname theAlgoLine\endcsname}}%
}{}
\makeatother

\newcommand{\Ot}{\widetilde{O}}
\newcommand{\pmax}{p_{\max}}
\newcommand{\wmax}{w_{\max}}
\newcommand{\OPT}{\textup{OPT}}
\newcommand{\vopt}{\widetilde \OPT}

\newcommand{\mc}[1]{\mathcal{#1}}

\newcommand{\maxconv}[2]{\textsc{MaxConv}(#1, #2)\xspace}
\newcommand{\minconv}[2]{\textsc{MinConv}(#1, #2)\xspace}

\newcommand\restr[2]{{
  \left.\kern-\nulldelimiterspace 
  #1 
  \littletaller 
  \right|_{#2} 
  }}

\newcommand{\littletaller}{\mathchoice{\vphantom{\big|}}{}{}{}}

\newcommand{\Cl}{C^{(\ell)}}
\newcommand{\Tl}{T^{(\ell)}}
\newcommand{\Al}{A^{(\ell)}}
\newcommand{\Bl}{B^{(\ell)}}
\newcommand{\CUPL}{C^{(\ell +1)}}
\newcommand{\TUPL}{T^{(\ell+1)}}
\newcommand{\AUPL}{A^{(\ell+1)}}
\newcommand{\BUPL}{B^{(\ell+1)}}

\renewcommand{\subset}{\subseteq}
\renewcommand{\le}{\leqslant}
\renewcommand{\leq}{\leqslant}
\renewcommand{\ge}{\geqslant}
\renewcommand{\geq}{\geqslant}

\renewcommand{\epsilon}{\ensuremath\varepsilon}

\renewcommand{\phi}{\ensuremath{\varphi}}

\title{Even Faster Knapsack via Rectangular Monotone Min-Plus Convolution and Balancing}
\titlerunning{Faster Knapsack via Rectangular Monotone MinPlusConvolution}

\author{Karl Bringmann}{Saarland University and Max Planck Institute for Informatics, Saarland Informatics Campus, Saarbrücken, Germany \and \url{https://people.mpi-inf.mpg.de/~kbringma/}}{bringmann@cs.uni-saarland.de}{https://orcid.org/0000-0003-1356-5177}{}

\author{Anita Dürr}{Saarland University and Max Planck Institute for Informatics, Saarland Informatics Campus, Saarbrücken, Germany \and \url{https://www.mpi-inf.mpg.de/departments/algorithms-complexity/people/current-members/anita-duerr}}{aduerr@mpi-inf.mpg.de}{https://orcid.org/0000-0003-0440-5008}{}

\author{Adam Polak}{Bocconi University, Milan, Italy \and \url{https://adampolak.github.io/}}{adam.polak@unibocconi.it}{https://orcid.org/0000-0003-4925-774X}{Part of this work was done when Adam Polak was at Max Planck Institute for Informatics.}

\authorrunning{K. Bringmann, A. Dürr, and A. Polak}

\Copyright{Karl Bringmann, Anita Dürr, and Adam Polak}

\funding{\emph{Karl Bringmann and Anita Dürr:} This work is part of the project TIPEA that has received funding from the European Research Council (ERC) under the European Unions Horizon 2020 research and innovation programme (grant agreement No. 850979).}

\acknowledgements{The authors thank Alejandro Cassis for many fruitful discussions.}



\nolinenumbers 

\begin{document}

\maketitle

\begin{abstract}
We present a pseudopolynomial-time algorithm for the Knapsack problem that has running time $\widetilde{O}(n + t\sqrt{p_{\max}})$, where $n$ is the number of items, $t$ is the knapsack capacity, and $p_{\max}$ is the maximum item profit. This improves over the $\widetilde{O}(n + t \, p_{\max})$-time algorithm based on the convolution and prediction technique by Bateni et al.~(STOC 2018). Moreover, we give some evidence, based on a strengthening of the Min-Plus Convolution Hypothesis, that our running time might be optimal.

Our algorithm uses two new technical tools, which might be of independent interest. First, we generalize the $\widetilde{O}(n^{1.5})$-time algorithm for bounded monotone min-plus convolution by Chi et al.~(STOC 2022) to the \emph{rectangular} case where the range of entries can be different from the sequence length. Second, we give a reduction from general knapsack instances to \emph{balanced} instances, where all items have nearly the same profit-to-weight ratio, up to a constant factor.

Using these techniques, we can also obtain algorithms that run in time
$\widetilde{O}(n + \textup{OPT}\sqrt{w_{\max}})$,
$\widetilde{O}(n + (nw_{\max}p_{\max})^{1/3}t^{2/3})$, and
$\widetilde{O}(n + (nw_{\max}p_{\max})^{1/3} \textup{OPT}^{2/3})$,
where $\textup{OPT}$ is the optimal total profit and $w_{\max}$ is the maximum item weight.
\end{abstract}

\section{Introduction}

In the Knapsack problem\footnote{Some related works refer to this problem as 0-1-Knapsack to distinguish it from its variants that allow picking an item multiple times in a solution, e.g., Multiple Knapsack or Unbounded Knapsack. In this paper we consider only the 0-1-Knapsack variant, hence we write Knapsack for short.} the input consists of a set of $n$ items, where item $i$ has weight $w_i \in \mathbb N$ and profit $p_i \in \mathbb N$, as well as a weight budget $t \in \mathbb N$ (also referred to as \emph{knapsack capacity}). The task is to compute the maximum total profit of any subset of items with total weight at most~$t$, i.e., we want to compute
$\OPT := \max\{ \sum_{i = 1}^n p_i x_i \mid x \in \{0, 1\}^n, \sum_{i = 1}^n w_i x_i \leq t\}$. 
Knapsack is one of the most fundamental problems in the intersection of computer science, mathematical optimization, and operations research.  
Since Knapsack is one of Karp's original 21 NP-complete problems \cite{Karp72}, we cannot hope for polynomial-time algorithms. However, when the input integers are small, we can consider pseudopolynomial-time algorithms where the running time depends polynomially on $n$ and the input integers. A well-known example is Bellman's dynamic programming algorithm that runs in time $O(n \cdot t)$, or alternatively in time $O(n \cdot \OPT)$~\cite{Bellman56}. 

Cygan et al.~\cite{CyganMWW19} and Künnemann et al.~\cite{KunnemannPS17} showed that under the Min-Plus Convolution Hypothesis there is no algorithm solving Knapsack in time $\Ot((n + t)^{2 - \epsilon})$ or $\Ot((n + \OPT)^{2 - \epsilon})$ for any constant $\epsilon > 0$. Hence in the regimes $t = \Theta(n)$ or $\OPT = \Theta(n)$ Bellman's dynamic programming algorithms are near-optimal. To overcome this barrier, recent works study the complexity of Knapsack in terms of two additional parameters: the maximum weight $\wmax$ and the maximum profit $\pmax$ of the given items. 
Note that we can assume without loss of generality that $\wmax \leq t$ and $\pmax \leq \OPT$. 
Clearly, by the same lower bounds as above there is no algorithm solving Knapsack in time $\Ot((n + \wmax)^{2 - \epsilon})$ or $\Ot((n + \pmax)^{2 - \epsilon})$ for any $\epsilon > 0$.
However, in certain regimes small polynomial dependencies on $\wmax$ and $\pmax$ can lead to faster algorithms compared to the standard dynamic programming algorithm. 
\cref{fig:table_knapsack} lists the results of prior work with this parameterization. To compare these running times, observe that we can assume without loss of generality that $t \leq n \cdot \wmax$ and $\OPT \leq n  \cdot \pmax$, since any feasible solution includes at most all $n$ items. We remark that most of the cited algorithms, including our contributions, are randomized. 

\begin{table}[h]
    \centering
    \caption{Pseudopolynomial-time algorithms for Knapsack.}
    \label{fig:table_knapsack}
    \footnotesize
    \begin{tabular}{ @{}ll@{} }
     \toprule
     \small{Reference} & \small{Running Time} \\
     \midrule
        Bellman \cite{Bellman56} & $O(n \cdot \min \{t, \OPT\})$ \\
        Pisinger \cite{Pisinger99} & $O(n \cdot \wmax \cdot \pmax)$ \\
        Kellerer and Pferschy \cite{KellererP04}, also \cite{BateniHSS18,AxiotisT19} & $\Ot(n + \min\{ t \cdot \wmax, \OPT \cdot \pmax \})$ \\
        Bateni, Hajiaghayi, Seddighin and Stein \cite{BateniHSS18} & $\Ot(n + t \cdot \pmax)$ \\
        Axiotis and Tzamos \cite{AxiotisT19} & $\Ot(n \cdot \min \{\wmax^2, \pmax^2 \})$ \\
        Polak, Rohwedder and W\k{e}grzycki \cite{PRW21} & $\Ot(n + \min \{\wmax^3, \pmax^3\})$ \\
        Bringmann and Cassis \cite{BringmannC22} & $\Ot(n + (t+\OPT)^{1.5})$ \\ 
        Bringmann and Cassis \cite{BringmannC23} & $\Ot(n \cdot \min \{\wmax \cdot \pmax^{2/3}, \pmax \cdot \wmax^{2/3}\})$ \\
        Jin \cite{Jin23} and He and Xu \cite{HeXu23} & $\Ot(n + \min\{\wmax^{5/2}, \pmax^{5/2}\}) $ \\
        Jin \cite{Jin23} & $\Ot(n \cdot \min\{\wmax^{3/2}, \pmax^{3/2}\}) $ \\
        Chen, Lian, Mao and Zhang \cite{CLMZ23} & $\Ot(n + \min \{ \wmax^{12/5}, \pmax^{12/5}\})$ \\
        Bringmann \cite{BringmannQuad23} and Jin \cite{JinQuad23} & $\Ot(n + \min\{\wmax^2, \pmax^2\})$ \\
        He and Xu \cite{HeXu23} & $\Ot(n^{3/2} \cdot \min\{\wmax, \pmax\})$ \\
        \cref{thm:knapsack_bmbm}, this work & $\Ot(n + t \sqrt{\pmax})$ \\
        \cref{thm:knapsack_bmbm_sym}, this work & $\Ot(n + \OPT \sqrt{\wmax} )$ \\
        \cref{thm:knapsack_bmdp}, this work & $\Ot(n + (n  \wmax \pmax)^{1/3} \cdot t^{2/3})$ \\
        \cref{thm:knapsack_bmdp_sym}, this work & $\Ot(n + (n  \wmax \pmax)^{1/3} \cdot \OPT^{2/3})$ \\
     \bottomrule
    \end{tabular}
\end{table}

\subsection{Our results}
Our main contribution is an $\Ot(n + t \sqrt{\pmax})$-time algorithm for Knapsack.

\begin{theorem}
\label{thm:knapsack_bmbm}
  There is a randomized algorithm for Knapsack that is correct with high probability and runs in time $\Ot(n + t \sqrt{\pmax})$.
\end{theorem}

Let us put this result in context. Bellman's algorithm and many other Knapsack algorithms in Table~\ref{fig:table_knapsack} run in pseudopolynomial time with respect to \emph{either} weights or profits. The first exception is Pisinger's $O(n\cdot\wmax\cdot\pmax)$-time algorithm~\cite{Pisinger99}, which offers an improvement in the regime where both weights and profits are small. Later, Bateni et al.~\cite{BateniHSS18} introduced the \emph{convolution and prediction} technique, which enabled them to improve over Pisinger's running time to $\Ot(n + t \pmax)$. Prior to our work, this was the best known pseudopolynomial upper bound in terms of $n$, $t$ and $\pmax$. In \cref{thm:knapsack_bmbm} we improve this running time by a factor $\sqrt{\pmax}$. We will discuss below that further improvements in terms of this parameterization seem difficult to obtain (see \cref{thm:lb}).

\subparagraph*{Further upper bounds.}
A long line of research~\cite{AxiotisT19,PRW21,Jin23,HeXu23,CLMZ23,BringmannQuad23,JinQuad23} recently culminated into an $\Ot(n+\wmax^2)$-time algorithm for Knapsack~\cite{BringmannQuad23, JinQuad23}, which matches the conditional lower bound ruling out time $\Ot((n + \wmax)^{2 - \epsilon})$ for any $\varepsilon > 0$~\cite{CyganMWW19,KunnemannPS17}. The biggest remaining open problem in this line of research is whether Knapsack can be solved in time $\Ot(n \cdot \wmax)$, which again would match the conditional lower bound and would be favourable if $n$ is smaller than $\wmax$. Our next result is a step in this direction: We design a Knapsack algorithm whose running time is the weighted geometric mean (with weights $1/3$ and $2/3$) of $\Ot(n \cdot \wmax)$ and the running time $\Ot(t \sqrt{\pmax})$ of~\cref{thm:knapsack_bmbm} (ignoring additive terms $\Ot(n)$).

\begin{theorem}
\label{thm:knapsack_bmdp}
    There is a randomized algorithm for Knapsack that is correct with high probability and runs in time $\Ot(n + (n \wmax \pmax)^{1/3} \cdot t^{2/3})$.
\end{theorem}

We also show that one can change our previous two algorithms to obtain symmetric running times where weight and profit parameters are exchanged.

\begin{theorem}\label{thm:knapsack_bmbm_sym}
    There is a randomized algorithm for Knapsack that is correct with high probability and runs in time $\Ot(n + \OPT \sqrt{\wmax})$.
\end{theorem}

\begin{theorem}\label{thm:knapsack_bmdp_sym}
    There is a randomized algorithm for Knapsack that is correct with high probability and runs in time $\Ot(n + (n \wmax  \pmax)^{1/3} \cdot \OPT^{2/3})$.
\end{theorem}

\subparagraph*{Lower bound?}

Finally, we give some argument why it might be difficult to improve upon
any of our running times
by a factor polynomial in any of the five parameters $n$, $\wmax$, $\pmax$, $t$ and $\OPT$. Specifically, we present a fine-grained reduction from the following variant of min-plus convolution.

\begin{restatable}[Bounded Min-Plus Convolution Verification Problem]{definition}{DefMPConvVer}
\label{def:mpconvver}
Given sequences $A[0 \ldots n-1], B[0 \ldots n-1]$, and $C[0 \ldots 2n-2]$ with entries in $\{0,1,\ldots,n\}$, determine whether for all $k$ we have $C[k] \le \min_{i+j=k} A[i] + B[j]$.
\end{restatable}

Min-plus convolution can be naively solved in time $O(n^2)$, and the Min-Plus Convolution Hypothesis postulates that this time cannot be improved to $O(n^{2-\varepsilon})$ for any $\varepsilon>0$ even for integer entries bounded by $M = \textup{poly}(n)$. For small $M$, min-plus convolution can be solved in time $\Ot(nM)$ using Fast Fourier Transform (FFT). Thus, $M=\Theta(n)$ is the smallest bound for which min-plus convolution conceivably might require quadratic time (although this is not asserted or implied by any standard hypothesis). This situation can be compared to the Strong 3SUM Hypothesis, which asserts hardness of the 3SUM problem with the smallest universe size that is not solved in subquadratic time by FFT.

Our reduction is not from the problem of computing the convolution, but only from the problem of verifying whether a given third sequence lower bounds the convolution element-wise. These two variants -- computation and verification -- are equivalent for the general unbounded min-plus convolution~\cite{CyganMWW19}, but no such equivalence is known for the bounded version (because the relevant reduction blows up the entries). 
We can show a reduction from the (potentially easier) verification problem.

\begin{restatable}{theorem}{ThmLB}
\label{thm:lb}
If Knapsack can be solved faster than the running time of any of \cref{thm:knapsack_bmbm,thm:knapsack_bmdp,thm:knapsack_bmbm_sym,thm:knapsack_bmdp_sym} by at least a factor polynomial in any of $n$, $\wmax$, $\pmax$, $t$, or $\OPT$, then Bounded Min-Plus Convolution Verification can be solved in time $O(n^{2-\varepsilon})$ for some $\varepsilon > 0$.
\end{restatable}

Specifically, we show that if Knapsack with parameters $\wmax, t = \Theta(n)$ and $\pmax, \OPT = \Theta(n^2)$ can be solved in time $O(n^{2-\varepsilon})$, then so can Bounded Min-Plus Convolution Verification. The same holds for Knapsack with parameters $\wmax, t = \Theta(n^2)$ and $\pmax, \OPT = \Theta(n)$.

This gives some evidence that our running times achieved in \cref{thm:knapsack_bmbm,thm:knapsack_bmdp,thm:knapsack_bmbm_sym,thm:knapsack_bmdp_sym} are near-optimal. While this lower bound is not assuming a standard hypothesis from fine-grained complexity, it still describes a barrier that needs to be overcome by any improved algorithm. 

\subsection{Technical overview}
The algorithms in \cref{thm:knapsack_bmbm,thm:knapsack_bmdp,thm:knapsack_bmbm_sym,thm:knapsack_bmdp_sym} follow the \emph{convolve and partition} paradigm used in many recent algorithms for Knapsack and Subset Sum (see, e.g., \cite{Bringmann17,BringmannC22,BringmannC23,BateniHSS18}). 
Our general setup follows \cite{BringmannC23}: We split the items at random into $2^q$ groups. In the base case, for each group and each target weight $j$ we compute the maximum profit attainable with weight at most~$j$ using items from that group. These groups are then combined in a tree-like fashion by computing max-plus convolutions. A key observation is that those sequences are monotone non-decreasing with non-negative entries, and one can bound the range of entries. 
We deviate from \cite{BringmannC23} in the algorithms for solving the base case and for combining subproblems by max-plus convolution:
For the base case, we use improved variants of the Knapsack algorithms of~\cite{BringmannC22} or~\cite{HeXu23} to obtain \cref{thm:knapsack_bmbm,thm:knapsack_bmbm_sym} or \cref{thm:knapsack_bmdp,thm:knapsack_bmdp_sym}, respectively. 
For the combination by max-plus convolution, we use the specialized max-plus convolution algorithm that we discuss next.

\paragraph*{Rectangular Monotone Max-Plus Convolution}
The max-plus convolution of two sequences $A, B \in \mathbb Z^n$ is defined as the sequence $C \in \mathbb Z^{2n-1}$ such that $C[k] = \max_{i + j = k} \{A[i] + B[j]\}$. This is well-known to be equivalent to min-plus convolution, and is more relevant for Knapsack applications, therefore from now on we only consider max-plus convolution.
Despite the quadratic time complexity of max-plus convolution on general instances, there are algorithms running in strongly subquadratic time if we assume some structure on the input sequences, see \cite{ChanL15,ChiDXZ22_stocs}. 
In fact, fast algorithms for structured max-plus convolution are exploited in multiple Knapsack algorithms: Kellerer and Pferschy~\cite{KellererP04}, Axiotis and Tzamos~\cite{AxiotisT19} and Polak, Rohwedder, and W\k{e}grzycki~\cite{PRW21} use the SMAWK algorithm~\cite{AggarwalKMSW87}, which can be used to compute in linear time the max-plus convolution of two sequences where one is concave. Bringmann and Cassis~\cite{BringmannC23} develop a subquadratic algorithm to compute the max-plus convolution between two near-concave sequences, and use this algorithm to solve Knapsack in time $\Ot(n \cdot \min \{\wmax \cdot \pmax^{2/3}, \pmax \cdot \wmax^{2/3}\})$. Finally, another Knapsack algorithm of Bringmann and Cassis in \cite{BringmannC22} uses the algorithm due to Chi, Duan, Xie and Zhang~\cite{ChiDXZ22_stocs} that computes the max-plus convolution between monotone sequences of non-negative entries bounded by $O(n)$ in time $\Ot(n^{1.5})$.

To obtain our \cref{thm:knapsack_bmbm,thm:knapsack_bmdp,thm:knapsack_bmbm_sym,thm:knapsack_bmdp_sym} we exploit a modification of the algorithm of Chi, Duan, Xie and Zhang~\cite{ChiDXZ22_stocs}. In particular, the following theorem generalizes the result of \cite{ChiDXZ22_stocs} to monotone sequences with non-negative entries bounded by an arbitrary parameter $M$. 

\begin{restatable}[Slight modification of \cite{ChiDXZ22_stocs}]{theorem}{MPConv}
    \label{thm:MPConv}
    The min-plus or max-plus convolution of two monotone (either both non-decreasing or both non-increasing) sequences of length at most $n$ with entries in $\{0,1,\ldots,M\}$ can be computed by a randomized algorithm that is correct with high probability and runs in time $\Ot(n\sqrt{M})$.
\end{restatable}

As a side result that might be of independent interest, we show that the assumption in \cref{thm:MPConv}, that both input sequences are monotone, can be replaced without loss of generality by the assumption that at least one input sequence is monotone, see \cref{thm:onetwomonotone}.

\begin{restatable}{theorem}{ThmOneTwoMinPlus}    
\label{thm:onetwomonotone}
    Suppose that there is an algorithm computing the max-plus convolution of two monotone non-decreasing sequences $A,B \in \{0, 1, \dots, M\}^n$ in time $T_2(n,M)$, and assume that $T_2(n,M)$ is monotone in $n$.
    Then there also is an algorithm computing the max-plus convolution of a monotone non-decreasing sequence $A \in \{0, 1, \dots, M\}^n$ and an arbitrary (i.e., not necessarily monotone) sequence $B \in \{0, 1, \dots, M\}^n$ in time $T_1(n,M)$ which satisfies the recurrence $T_1(n,M) \le 2\, T_1(n/2,M) + O(T_2(n,M))$.
    
    The same statement holds with ``non-decreasing'' replaced by ''non-increasing'', or with ``max-plus'' replaced by ``min-plus'', or both.
\end{restatable}

\paragraph*{Balancing}
In the above described Knapsack algorithm of \cref{thm:knapsack_bmbm,thm:knapsack_bmdp,thm:knapsack_bmbm_sym,thm:knapsack_bmdp_sym}, the sequences for which we want to compute the max-plus convolution are monotone non-decreasing and contain non-negative entries. To use \cref{thm:MPConv} for the computation of their max-plus convolution, we need to ensure that the entries also have bounded values. We will show that under the \emph{balancedness} assumption $t / \wmax = \Theta(\OPT / \pmax)$, and due to the random splitting, it suffices to consider entries in a small weight interval and in a small profit interval. In order to use this observation, the algorithms of \cref{thm:knapsack_bmbm,thm:knapsack_bmdp,thm:knapsack_bmbm_sym,thm:knapsack_bmdp_sym}  first reduce a Knapsack instance $(\mc I, t)$ to another instance $(\mc I', t')$ where the balancedness assumption is satisfied, and then solve Knapsack on this balanced instance $(\mc I', t')$. 

\begin{restatable}{lemma}{LemRedToBal}
\label{lem:reduction_balanced_knapsack}
Solving Knapsack can be reduced, in randomized time $\Ot(n + \wmax\sqrt{\pmax})$ (respectively $\Ot(n + \pmax \sqrt{\wmax})$), to solving a Knapsack instance for $O(\wmax)$ consecutive capacities (respectively $O(\pmax)$ consecutive profits), where the reduced instance satisfies $t / \wmax = \Theta(\OPT / \pmax)$ and consists of a subset of the items of the original instance; in particular, all relevant parameters $n$, $\wmax$, $\pmax$, $t$, and $\OPT$ of the reduced instance are no greater than those of the original instance.
\end{restatable}

\subsection{Outline}

After preliminaries in \cref{sec:preliminaries}, in \cref{sec:knapsack_algos} we present two Knapsack algorithms corresponding to \cref{thm:knapsack_bmbm,thm:knapsack_bmdp} that assume the balancedness assumption. 
We justify the balancedness assumption in \cref{sec:reduction} by proving \cref{lem:reduction_balanced_knapsack}.
In \cref{sec:lower_bound} we prove our conditional lower bound (\cref{thm:lb}). 
In \cref{sec:1to2MinConv} we prove the side result that in \cref{thm:MPConv} the assumption that both input sequences are monotone can be replaced without loss of generality by the assumption that at least one input sequence is monotone.

\cref{sec:sym_knapsack} contains variations of the Knapsack algorithms of \cref{sec:knapsack_algos} corresponding to \cref{thm:knapsack_bmbm_sym,thm:knapsack_bmdp_sym}. We explain how the result of \cite{ChiDXZ22_stocs} generalizes to \cref{thm:MPConv} in \cref{sec:adapt_chieetal_MPConv}. In \cref{sec:adapt_subsetsum_knapsack} we discuss how to derive from \cite{BringmannC22} the Knapsack algorithms (\cref{thm:Knapsack_SubsetSum,thm:Knapsack_SubsetSum_sym}) used in \cref{sec:knapsack_algos,sec:sym_knapsack}.

\section{Preliminaries}\label{sec:preliminaries}

We use the notation $\mathbb N = \{0, 1, 2, \dots\}$ and define $[n] := \{ 1, 2, \dots, n\}$ for $n \in \mathbb N$. 
Let $A[i_s \dots i_f]$ be an integer array of length $i_f -i_s + 1$ with start index $i_s$ and end index $i_f$.
We interpret out-of-bound entries as $-\infty$, and thus, when it is clear from context, simply denote the array $A[i_s \dots i_f]$ by $A$.
Then $-A$ is the entry-wise negation of $A$.
We call $A$ \emph{monotone non-decreasing} (respectively \emph{non-increasing}), if for every $i, j$ such that $i_s \leq i \leq j \leq i_f$ we have $A[i] \leq A[j]$ (resp. $A[i] \geq A[j]$). $A$ is \emph{monotone} if it is either monotone non-decreasing or monotone non-increasing.

\begin{definition}[Restriction to index and entry interval]
Suppose that $A$ is monotone and consider intervals $I \subset \mathbb N$ and $V \subset \mathbb Z$. 
We define the operation $D \gets A[I ; V]$ as follows. If there exist no index $i \in I$ with $A[i] \in V$, then set $D$ to the empty array. Otherwise, let $i_{\min} := \min\{i \in I : A[i] \in V\}$ and $i_{\max} := \max \{i \in I : A[i] \in V \}$, and set $D$ to the subarray $A[i_{\min} \dots i_{\max}]$. Note that since $A$ is monotone, for every $i \in \{i_{\min}, \dots, i_{\max}\}$ we have $A[i] \in V$. Thus $A[I ; V]$ returns the subarray of $A$ with indices in $I$ and values in $V$.

We sometimes abbreviate $A[\{0,\ldots,i\}; \{0,\ldots,v\}]$ by $A[0 \ldots i; 0 \ldots v]$.
\end{definition}

\subparagraph*{Max-plus convolution.}
Let $A[i_s \dots i_f]$ and $B[j_s \dots j_f]$ be two integer arrays of length $n := i_f - i_s + 1$ and $m := j_f -j_s + 1$, respectively. Assume without loss of generality that $n \geq m$.
The max-plus convolution problem on instance $(A,B)$ asks to compute the finite values of the array $C := \maxconv{A}{B}$, which is defined as $C[k] := \max_{i + j = k} \{A[i] + B[j]\}$ for every $k \in \mathbb N$; here $i,j$ range over all integers with $i+j=k$. Note that $C[k]$ is finite only for $k \in \{i_s + j_s, \dots, i_f + j_f \}$.

The min-plus convolution problem is defined analogously by replacing max by min. Note that the two operations are equivalent since $\minconv{A}{B} = - \maxconv{{-A}}{{-B}}$. In the context of min-plus convolution, we interpret out-of-bound entries as $\infty$ instead of $-\infty$.

When the sequences $A[i_s \dots i_f]$ and $B[j_s \dots j_f]$ are either both monotone non-decreasing or both monotone non-increasing, and with values contained in $\{0, 1, \dots, M\}$, for some integer $M$, then the problem of computing $\maxconv{A}{B}$ is called the \emph{bounded monotone} max-plus convolution problem. We call the general case with arbitrary $M$ \emph{rectangular}, as opposed to the \emph{square} bounded monotone max-plus convolution where $M = \Theta(n)$.
Chi et al.~\cite{ChiDXZ22_stocs} showed that square bounded monotone max-plus convolution can be solved in time $\Ot(n^{1.5})$. By slightly adapting their algorithm, we show in 
\cref{sec:adapt_chieetal_MPConv} that rectangular bounded monotone max-plus convolution can be solved in time $\Ot(n \sqrt{M})$.

\MPConv*

\subparagraph*{Knapsack.}
The Knapsack problem is defined as follows.
Let $\mc I = \{(w_1, p_1), (w_2, p_2), \allowbreak \dots, (w_n, p_n)\}$ be a (multi-)set of $n$ items, where item $i$ has weight $w_i$ and profit $p_i$. Let $t \in \mathbb N$ be a weight capacity. The goal is to compute $\OPT = \max \sum_{i = 1}^n p_i x_i$ subject to the constraints $x \in \{0,1\}^n$ and $\sum_{i= 1}^n w_i x_i \leq t$.
We denote by $\wmax = \max_i w_i$ the maximum weight and by $\pmax = \max_i p_i$ the maximum profit of the items in $\mc I$. Note that by removing items that have weight larger than $t$ we can assume without loss of generality that $\wmax \leq t$. Then every single item is a feasible solution, so $\pmax \leq \OPT$. 
If $t \ge n \wmax$ then all items can be picked in a solution and the result is $\OPT = \sum_i p_i$, so the instance is trivial; therefore we can assume without loss of generality that $t < n \wmax$.
Since any feasible solution contains at most all $n$ items we also have $\OPT \leq n \cdot \pmax$. We can also assume that $n \geq 10$, since for $n = O(1)$ a standard $O(2^n)$-time algorithm runs in time $O(1)$. 

We identify each item $(w_i, p_i) \in \mc I$ with its index $i \in [n]$ so that any subset of items $\mc J \subset \mc I$ can be identified with the set of indices $S \subseteq [n]$ such that $\mc J = \{ (w_i, p_i) : i \in S \}$. With slight abuse of notation we sometimes write $\mc J = S$. We define the partial weight and partial profit functions $w_{\mc J}(x) := \sum_{i \in \mc J} w_i x_i$ and $p_{\mc J}(x) := \sum_{i \in \mc J} p_i x_i$ for $\mc J \subset \mc I$. We also define the profit sequence $\mc P_{\mc J}[\cdot]$ such that
$$
\mc P_{\mc J}[k] = \max\{ p_{\mc J}(x) \ \vert\ x\in \{0, 1\}^n, w_{\mc J}(x) \leq k\}
$$
for any $k \in \mathbb N$. Note that the task is to compute $\OPT = \mc P_{\mc I}[t]$.

\subparagraph*{Computing $\boldsymbol{\mc P_{\mc I}}$.}
A standard way to compute (part of) the profit sequence $\mc P_{\mc I}$ is to use dynamic programming:

\begin{theorem}[Bellman~\cite{Bellman56}]\label{thm:Knapsack_DP}
    Given a Knapsack instance $(\mc I, t)$ and $k \in \mathbb N$, the sequence $\mc P_{\mc I}[0 \dots k]$ can be computed in time $O(|\mc I| \cdot k)$.
\end{theorem}

Bringmann and Cassis exploit in \cite{BringmannC22} the fact that $\mc P_{\mc I}$ is monotone non-decreasing. They show that one can compute $\mc P_{\mc I}[0 \dots j; 0 \ldots j]$ in roughly the same time as it takes to compute a square bounded monotone max-plus convolution of length $j$. In 
\cref{sec:adapt_subsetsum_knapsack} 
we slightly generalize their algorithm so that it computes the entries of $\mc P_{\mc I}[0 \dots j; 0 \dots v]$. The modified algorithm uses rectangular instead of square bounded monotone max-plus convolutions. Combining the result with \cref{thm:MPConv}, we prove the following theorem in 
\cref{sec:adapt_subsetsum_knapsack}.

\begin{theorem}[Slight modification of \cite{BringmannC22}]\label{thm:Knapsack_SubsetSum}
    Given a Knapsack instance $(\mc I, t)$ and $v \in \mathbb N$, the sequence $\mc P_{\mc I}[0 \ldots t; 0 \ldots v]$ can be computed by a randomized algorithm that is correct with high probability and runs in time $\Ot(n + t \sqrt{v})$. 
\end{theorem}

\subparagraph*{Approximating $\OPT$.}
We use the following variant of the greedy algorithm for Knapsack. 
\begin{lemma}[e.g.~{\cite[Theorem 2.5.4]{KPP04}}]\label{lem:approx_opt}
    Given a Knapsack instance $(\mc I, t)$, one can compute $\vopt \in \mathbb N$ such that $\OPT \leq \vopt \leq \OPT + \pmax$ and $\pmax \leq \vopt \leq n \cdot \pmax$ in $\Ot(n)$ time.
\end{lemma}
\begin{proof}
The greedy algorithm works as follows. Sort and relabel the elements in decreasing order of profit-to-weight ratio such that $p_1 / w_1 \geq  p_2 / w_2  \geq \cdots \geq p_n / w_n$.
Select the maximum prefix of items $\{1,2,\ldots,i^*\}$ such that $\sum_{i=1}^{i^*} w_i \le t$. 
We define $\vopt := \sum_{i=1}^{\min\{i^* + 1, n\}} p_i$.

The fractional solution $x^{\textup{LP}}$ which fully selects the items in $\{1,2,\ldots,i^*\}$ and selects a $(t - w_{\mc I}(x))/w_{i^* + 1}$-fraction of item $i^* + 1$ is the optimal solution to the linear programming relaxation of the Knapsack problem (see \cite[Theorem 2.2.1]{KPP04}). Thus, $\OPT \leq p_{\mc I}(x^{\textup{LP}}) \le \vopt$. Since we assumed without loss of generality that $\wmax \le t$, each single item fits into the knapsack, which implies $\pmax \le \OPT \le \vopt$.
Since the solution $\{1,2,\ldots,i^*\}$ is feasible and item $i^* + 1$ has profit at most $\pmax$, we have $\vopt \le \OPT + \pmax$. Finally, we have $\vopt \le \sum_{i=1}^n p_i \le n \cdot \pmax$. 
\end{proof}

\subparagraph*{Pareto optimum of $\boldsymbol{\mc P_{\mc I}}$.}
The sequence $\mc P_{\mc I}$ is monotone non-decreasing, so we can define the \emph{break points} of $\mc P_{\mc I}$ as the integers $k \in \mathbb N$ such that $\mc P_{\mc I}[k - 1] < \mc P_{\mc I}[k]$. In particular, $\mc P_{\mc I}$ is constant between two break points, and thus it is enough to focus on the values taken at break points of $\mc P_{\mc I}$. 
For every break point $k \in \mathbb N$, there exists $x \in \{0, 1\}^n$ with $w_{\mc I}(x) = k$ and $\mc P_{\mc I}[k] = p_{\mc I}(x)$. We call such a vector a \emph{Pareto optimum} of $\mc P_{\mc I}$. Indeed, by the definition of $\mc P_{\mc I}$, if a vector $y \in \{0, 1\}^n$ has higher profit $p_{\mc I}(y) > p_{\mc I}(x) = \mc P_{\mc I}[w_{\mc I}(x)]$ then it necessarily has higher weight $w_{\mc I}(y) > w_{\mc I}(x)$. We observe the following property of Pareto optima. 

\begin{lemma}\label{lem:pareto_maximum}
    Let $x \in \{0, 1\}^n$ be a Pareto optimum of $\mc P_{\mc I}$ and let $\mc J \subset \mc I$. Consider a vector $y \in \{0, 1\}^n$ such that $w_{\mc J}(y) \leq w_{\mc J}(x)$ and $p_{\mc J}(y) \geq p_{\mc J}(x)$. Then necessarily $p_{\mc J}(y) = p_{\mc J}(x)$.
\end{lemma}
\begin{proof}
    Suppose for the sake of contradiction that $p_{\mc J}(y) > p_{\mc J}(x)$. Consider the vector $y'$ that is equal to $y$ on $\mc J$ and equal to $x$ on $\mc I \setminus \mc J$. Then $p_{\mc I} (y') = p_{\mc J}(y) + p_{\mc I \setminus \mc J}(x) > p_{\mc J}(x) + p_{\mc I \setminus \mc J}(x) = p_{\mc I}(x)$. We also have $w_{\mc I}(y') = w_{\mc J}(y) + w_{\mc I \setminus \mc J}(x) \leq w_{\mc J}(x) + w_{\mc I \setminus \mc J}(x) = w_{\mc I}(x)$. This contradicts $x$ being a Pareto optimum.
\end{proof}

\medskip
We use $\Ot$-notation to hide poly-logarithmic factors in the input size $n$ and the largest input number $U$, i.e., $\Ot(T) := \bigcup_{c \geq 0} O(T \log^c(n \cdot U))$. In particular, for Knapsack we hide polylogarithmic factors in $n,\wmax,\pmax$.
Many subroutines that we use throughout the paper are randomized and compute the correct output with probability at least $1-1/n$. Standard boosting improves the success probability to $1-1/n^{10}$ at the cost of only a constant factor increase in running time. We can therefore assume that these subroutines have success probability $1-1/n^{10}$.

\section{Knapsack algorithm for balanced instances}\label{sec:knapsack_algos}

In this section we focus on balanced Knapsack instances, i.e., instances satisfying $t/\wmax =
\Theta(\OPT/\pmax)$. We call this the \emph{balancedness assumption}. In \cref{sec:reduction} 
we show that any Knapsack instance can be reduced to a balanced instance (see \cref{lem:reduction_balanced_knapsack}). 
Combined with the following \cref{lem:knapsack_bmbm,lem:knapsack_bmdp}, this proves \cref{thm:knapsack_bmbm,thm:knapsack_bmdp}.

\begin{restatable}{lemma}{knapsackBMBM}
    \label{lem:knapsack_bmbm}
        For any Knapsack instance $(\mc I, t)$ satisfying $t / \wmax  = \Theta (\OPT / \pmax)$ the sequence $\mc P_{\mc I}[T ; P]$ for $T := [t - \sqrt{t \cdot \wmax}, t + \sqrt{t \cdot \wmax}]$, $P := [\vopt - \sqrt{\vopt \cdot \pmax}, \allowbreak \vopt + \sqrt{\vopt \cdot \pmax}]$ and $\OPT \leq \vopt \leq \OPT + \pmax$ can be computed by a randomized algorithm in time $\Ot(n + t \sqrt{\pmax})$.
\end{restatable}

\begin{restatable}{lemma}{knapsackBMDP}
    \label{lem:knapsack_bmdp}
        For any Knapsack instance $(\mc I, t)$ satisfying $t / \wmax  = \Theta (\OPT / \pmax)$ the sequence $\mc P_{\mc I}[T ; P]$ for $T := [t - \sqrt{t \cdot \wmax}, t + \sqrt{t \cdot \wmax}]$, $P := [\vopt - \sqrt{\vopt \cdot \pmax}, \allowbreak \vopt + \sqrt{\vopt \cdot \pmax}]$ and $\OPT \leq \vopt \leq \OPT + \pmax$ can be computed by a randomized algorithm in time $\Ot(n + (n \wmax \pmax)^{1/3} \cdot t^{2/3})$.
    \end{restatable}

We prove \cref{lem:knapsack_bmbm,lem:knapsack_bmdp} in \cref{sec:bmbm,sec:bmdp} respectively. The presented algorithms compute the optimal profit, so in \cref{sec:reconstruction} we discuss how to reconstruct an optimal solution. 

\subsection{\texorpdfstring{$\boldsymbol{\Ot(n + t \sqrt{\pmax})}$}{O(n + t sqrt(pmax))}-time algorithm}\label{sec:bmbm}
We now prove \cref{lem:knapsack_bmbm}.
Observe that, with the notation of \cref{lem:knapsack_bmbm}, we have $\mc P_{\mc I}[t] = \OPT$, $t \in T$ and $\OPT \in P$, since $\pmax \leq \OPT$. Hence the algorithm in \cref{lem:knapsack_bmbm} computes in particular the value $\mc P_{\mc I}[t] = \OPT$.

\subparagraph*{Idea.}
The idea of the algorithm is to randomly split the items of $\mc I$ into $2^q$ groups $\mc I_1^{q}, \dots, \mc I_{2^q}^{q}$, for some parameter $q$ which we define later. Using the $\Ot(n + t \sqrt{v})$ time Knapsack algorithm (\cref{thm:Knapsack_SubsetSum}), we compute a subarray of $\mc P_{\mc I_q^j}$ for every $j \in [2^q]$. The arrays are then combined in a tree-like fashion by taking their max-plus convolution. A key observation is that, with high probability, it suffices to compute a subarray of $\mc P_{\mc I_q^j}$ for a small range of indices and a small range of values. The same will hold for the intermediate arrays resulting from the max-plus convolutions. Since the sequences are monotone non-decreasing, we can use the rectangular bounded monotone max-plus convolution algorithm of \cref{thm:MPConv} to accelerate the computation. We explain the algorithm in more details below before proving its correctness and analyzing its running time.
\newline

\begin{algorithm}[H]
    \caption{The $\Ot( n + t \sqrt{\pmax})$-time algorithm of \cref{lem:knapsack_bmbm}. The input $(\mc I, t)$ is a Knapsack instance such that $t / \wmax = \Theta(\OPT / \pmax)$.
    }\label{alg:knapsack_bmbm}
        $\wmax \gets \max_{i \in [n]} w_i$ \\
        $\pmax \gets \max_{i \in [n]} p_i$ \\
        Compute an approximation $\vopt$ of $\OPT$ using \cref{lem:approx_opt}.\\
        $q \gets $ largest integer such that $2^q \leq \min\{ t / \wmax,\vopt/\pmax\}$ \label{alg_line:bmbm_def_q} \\
        $\eta \gets 17 \log n$ \\
        $\Delta_w \gets t \cdot \wmax$ \\
        $\Delta_p \gets \vopt\cdot \pmax$\\
        $\mc I_1^q, \dots, \mc I_{2^q}^q \gets $ random partitioning of $\mc I$ into $2^q$ groups \\
        
        $W^q \gets \left[\frac{t}{2^{q}} - \sqrt{\frac{\Delta_w}{2^{q}}} \eta,\ \frac{t}{2^{q}} + \sqrt{\frac{\Delta_w}{2^{q}}} \eta\right]$ \\
        $P^q \gets \left[\frac{\vopt}{2^{q}} - \sqrt{\frac{\Delta_p}{2^{q}}} \eta,\ \frac{\vopt}{2^{q}} + \sqrt{\frac{\Delta_p}{2^{q}}} \eta\right]$ \\
        $W^* \gets \left[0, \frac{t}{2^q} + \sqrt{\frac{\Delta_w}{2^q}} \eta \right]$ \\
        $P^* \gets \left[0, \frac{\vopt}{2^q} + \sqrt{\frac{\Delta_p}{2^q}} \eta \right]$ \\
    
        \For{$j = 1, \dots, 2^q$}{\label{alg_line:bmbm_start_base}
            Compute $D_j^q \gets {\mc P_{\mc I_j^q}\left[W^* ; P^*\right]}$ using \cref{thm:Knapsack_SubsetSum} \label{alg_line:bmbm_leaf_D}
            \\
            $C_j^q \gets D_j^q[W^q ; P^q]$ \label{alg_line:bmbm_leaf_C} \\
        }
        \For{$\ell = q-1, \dots, 0 $}{\label{alg_line:bmbm_start_combi}
            $W^\ell \gets \left[\frac{t}{2^{\ell}} - \sqrt{\frac{\Delta_w}{2^{\ell}}} \eta,\ \frac{t}{2^{\ell}} + \sqrt{\frac{\Delta_w}{2^{\ell}}} \eta\right]$ \\
            $P^\ell \gets \left[\frac{\vopt}{2^{\ell}} - \sqrt{\frac{\Delta_p}{2^{\ell}}} \eta,\ \frac{\vopt}{2^{\ell}} + \sqrt{\frac{\Delta_p}{2^{\ell}}} \eta\right]$ \\
            
            \For{$j = 1, \dots, 2^\ell$}{
                $D_j^\ell \gets \maxconv{C_{2j -1}^{\ell + 1}}{C_{2j}^{\ell + 1}}$ using \cref{thm:MPConv} \label{alg_line:bmbm_combi_D} \\
                $C_j^\ell \gets D_j^\ell[W^\ell ; P^\ell]$ \label{alg_line:bmbm_combi_C} \\
    
                }
            }
        $T \gets [t - \sqrt{t \cdot \wmax}, t + \sqrt{t \cdot \wmax}]$ \\
        $P \gets [\vopt- \sqrt{\vopt\cdot \pmax}, \vopt+ \sqrt{\vopt\cdot \pmax}]$ \\
        \Return{$C_1^0[T ; P]$}
\end{algorithm}

\subparagraph*{Algorithm.}
The algorithm of \cref{lem:knapsack_bmbm} is presented in pseudocode in \cref{alg:knapsack_bmbm}. Let $\vopt$ be the approximation of $\OPT$ from \cref{lem:approx_opt}, i.e.\ $\vopt$ satisfies $\OPT \leq \vopt\leq \OPT + \pmax$ and $\pmax \leq \vopt\leq n \cdot \pmax$. Note that since $\pmax \leq \OPT$, we have $\vopt= \Theta(\OPT)$.
Set the parameters $\eta := 17 \log n$ and $q$ to be the largest integer such that $2^q \leq \min\{ t / \wmax,\vopt/\pmax\}$. We also define $\Delta_w := t \cdot \wmax$ and $\Delta_p := \vopt\cdot \pmax $, as well as the weight and profit intervals 
for $\ell \in \{0, \dots, q\}$
$$W^\ell := \left[\frac{t}{2^{\ell}} - \sqrt{\frac{\Delta_w}{2^{\ell}}} \eta,\ \frac{t}{2^{\ell}} + \sqrt{\frac{\Delta_w}{2^{\ell}}} \eta\right] \quad \text{ and } \quad P^\ell := \left[\frac{\vopt}{2^{\ell}} - \sqrt{\frac{\Delta_p}{2^{\ell}}} \eta,\ \frac{\vopt}{2^{\ell}} + \sqrt{\frac{\Delta_p}{2^{\ell}}} \eta\right].$$

\cref{alg:knapsack_bmbm} starts by splitting the items of $\mc I$ into $2^q$ groups $\mc I_1^{q}, \dots, \mc I_{2^q}^{q}$ uniformly at random. For each group $\mc I_j^{q}$ it computes the sequence $D_j^q := \mc P_{\mc I_j^q}[W^* ; P^* ]$ using \cref{thm:Knapsack_SubsetSum}, where $W^* := \left[0, \frac{t}{2^q} + \sqrt{\frac{\Delta_w}{2^q}} \eta \right]$ and $P^* := \left[0, \frac{\vopt}{2^q} + \sqrt{\frac{\Delta_p}{2^q}} \eta \right]$.
Then it extracts the entries corresponding to weights in $W^q$ and profits in $P^q$, i.e., $C_q^j := D_j^q[W^q ; P^q]$. 
Next, the algorithm iterates over the levels $\ell = q - 1, \dots, 0$. For every iteration $j \in[2^\ell]$, the set of items in group $j$ on level $\ell$ is $\mc I_j^\ell = \mc I_{2j-1}^{\ell + 1} \cup \mc I_{2j}^{\ell +1}$ and the algorithm computes the max-plus convolution $D_j^\ell$ of the arrays $C_{2j-1}^{\ell +1}$ and $C_{2j}^{\ell +1}$. It extracts the relevant entries of weights in $W^\ell$ and profits in $P^\ell$, i.e., $C_j^\ell := D_j^\ell[W^\ell ; P^\ell]$. Finally, observe that when $\ell = 0$ then $\mc I_1^0 = \mc I$. The algorithm returns the sequence $C_1^0[T ; P]$, for the intervals $T := [t - \sqrt{t \cdot \wmax}, t + \sqrt{t \cdot \wmax}]$ and $P := [\vopt - \sqrt{\vopt \cdot \pmax}, \vopt + \sqrt{\vopt \cdot \pmax}]$.

\subsubsection{Correctness of Algorithm \ref{alg:knapsack_bmbm}}

Let us analyze the correctness of the algorithm. For the rest of this section, fix a Knapsack instance $(\mc I, t)$ with $n := |\mc I|$ and such that $t /\wmax = \Theta(\OPT / \pmax)$.
First, recall that we defined $q$ to be the largest integer such that $2^q \leq \min \{t/\wmax,\vopt/ \pmax\}$. In particular, since $t \leq n \wmax$, we have $2^q \leq t / \wmax \leq n$. Moreover since $\wmax \leq t$ and $\pmax \leq\vopt$, we have $2^q \geq 1$. So $2^q$ is a valid choice for the number of groups in which we split the item set $\mc I$. Also note that $2^q = \Theta(t / \wmax) = \Theta(\vopt/\pmax)$.
Next, we argue that the subarray $C_j^\ell$ constructed in \cref{alg_line:bmbm_leaf_C,alg_line:bmbm_combi_C} is monotone non-decreasing. 

\begin{lemma}\label{lem:bmbm_monotonicity}
    For every level $\ell \in \{0, \dots, q\}$ and iteration $j \in[2^\ell]$,  
    the sequence $C_j^\ell$ is monotone non-decreasing.
\end{lemma}
\begin{proof}
    For $\ell = q$ and $j \in[2^q]$, $D_j^q$ is a subarray of $\mc P_{\mc I_j^q}$, which is monotone non-decreasing by definition.
    Hence $D_j^q$ is monotone non-decreasing, and since $W^q$ and $P^q$ are intervals, the array $C_j^q = D_j^q[W^q ; P^q]$ is also monotone non-decreasing. The statement follows from induction by noting that the max-plus convolution of two monotone non-decreasing sequences is a monotone non-decreasing sequence.
\end{proof}

The above lemma justifies the use of \cref{thm:MPConv} to compute the max-plus convolutions in \cref{alg_line:bmbm_combi_D}. We now explain why it is enough to restrict the entries of $D_j^\ell$ corresponding to indices in $W^\ell$ and values in $P^\ell$.
The following lemma shows that, for any fixed subset of items, the weight and profit of that subset restricted to $\mc I_j^\ell$ are concentrated around their expectations. 

\begin{lemma}\label{lem:bmbm_double_concentration}
   Let $x \in \{0, 1\}^n$. Fix $\ell \in \{0, \dots, q\}$ and $j \in[2^\ell]$. Then with probability at least $1 - 1/n^7$ the following holds:
    \[
    \left|w_{\mc I_j^\ell}(x) -  w_{\mc I}(x)/2^\ell \right| \leq \sqrt{\Delta_w / 2^\ell} \cdot 16 \log n
    \quad
    \text{ and }
    \quad
    \left|p_{\mc I_j^\ell}(x) -  p_{\mc I}(x)/2^\ell \right| \leq \sqrt{\Delta_p / 2^\ell} \cdot 16 \log n.
    \]
\end{lemma}

\begin{proof}
    By construction, $\mc I^\ell_j$ is a random subset of $\mc I$ where each item is included with probability $p := 1/2^\ell$. For each item $i \in [n]$, let $Z_i$ be a random variable taking value $w_i x_i$ with probability $p$, and 0 with probability $1-p$. Then $Z := \sum_{i = 1}^n Z_i$ has the same distribution as $w_{\mc I^\ell_j}(x)$ and $\mathbb{E}(Z)=w_{\mc I}(x) p$.

    Using Bernstein's inequality (see, e.g., \cite[Theorem 1.2]{DubhashiPanconesi}) we get that for any $\lambda > 0$:
    \begin{align*}
        \mathbb{P}(|Z - \mathbb{E}(Z)| \geq \lambda) & \leq 2 \exp\left(- \frac{\lambda^2}{2 \cdot \textup{Var}(Z) + \frac{2}{3} \lambda \cdot \wmax }\right)
        \\
        &\leq 2 \exp\left(- \min \left\{ \frac{\lambda^2}{4 \cdot \textup{Var}(Z)}, \frac{\lambda}{2\wmax} \right\} \right)
    \end{align*}
    Set $\lambda:=\sqrt{p \cdot \Delta_w} \cdot 16 \log n$. We can bound the variance of $Z$ as follows:
    \begin{align*}
        \textup{Var}(Z) = \sum_{i = 1}^n p (1-p) w_i^2 x_i^2 &\leq p \cdot \wmax \sum_{i=1}^n w_i x_i \\
        &= p \cdot \wmax \cdot w_{\mc I}(x) \leq p \cdot \wmax \cdot t = p \cdot \Delta_w.
    \end{align*}
    Hence $\lambda^2/(4 \cdot \textup{Var}(Z)) \geq 16 \log n$. To bound $\lambda / (2 \wmax)$, note that $2^q \leq t / \wmax$ so $p = \frac{1}{2^\ell} \geq \frac{1}{2^q} \geq \frac{\wmax}{t}$. Thus,
    $$
    \frac{\lambda}{2\wmax} = \frac{\sqrt{p \cdot \Delta_w} \cdot 16 \log n}{2 \wmax} \geq 8 \log n.
    $$
    Combining all the above we obtain that
    $$
    |w_{\mc I_j^\ell}(x) - w_{\mc I}(x) / 2^\ell| = |Z - \mathbb{E}(Z)| \leq \lambda = \sqrt{\Delta_w / 2^\ell} \cdot 16 \log n
    $$
    holds with probability at least $1 - 2/n^8$.
    
    We can apply a similar reasoning on $p_{\mc I_j^\ell}(x)$ and get the analogous result that 
    $
    |p_{\mc I_j^\ell}(x) - p_{\mc I}(x) / 2^\ell| \leq \sqrt{\Delta_p / 2^\ell} \cdot 16 \log n
    $
    holds with probability at least $1 - 2/n^8$.
    To this end, we define a random variable $Y$, analogous to $Z$, with respect to profits and set the constant in Bernstein's inequality to $\lambda = \sqrt{p \cdot \Delta_p} \cdot 16 \log n$. Then to bound $\textup{Var}(Y)$ we use $p_{\mc I}(x) \leq \OPT \leq\vopt$, and to bound $\lambda/(2\pmax)$ we use the fact that $2^q \leq \vopt/\pmax$ so that $p \geq \pmax /\vopt$. By a union bound, both events hold with probability at least $1 - 4/n^8 \geq 1 - 1/n^7$ (recall that we can assume $n \geq 10$).
\end{proof}

In the next lemma, we show that, as a consequence of \cref{lem:bmbm_double_concentration}, at level $\ell$ the weights and profits of solutions of interest restricted to $\mc I_j^q$ lie with sufficiently high probability in $W^\ell$ and $P^\ell$.

\begin{lemma}\label{lem:bmbm_total_concentration}
    Let $x \in \{0, 1\}^n$ such that $|w_{\mc I}(x) - t| \leq 2 \sqrt{\Delta_w}$ and $|p_{\mc I}(x) -\vopt| \leq 2\sqrt{\Delta_p}$. 
    Fix a level $\ell \in \{0, \dots, q\}$ and an iteration $j \in[2^\ell]$. Then with probability at least $1 - 1/n^7$ we have $w_{\mc I_j^\ell} (x) \in W^\ell$ and $p_{\mc I_j^\ell} (x) \in P^\ell$.
\end{lemma}

\begin{proof}
    By \cref{lem:bmbm_double_concentration} we have with probability at least $1 - 1/n^7$
    \[
    \left|w_{\mc I_j^\ell}(x) -  w_{\mc I}(x)/2^\ell \right| \leq \sqrt{\Delta_w / 2^\ell} 16 \log n
    \quad
    \text{ and }
    \quad
    \left|p_{\mc I_j^\ell}(x) -  p_{\mc I}(x)/2^\ell \right| \leq \sqrt{\Delta_p / 2^\ell} \cdot 16 \log n.
    \]
    We condition on that event.
    Since $|w_{\mc I}(x) - t| \leq 2\sqrt{\Delta_w}$, we have:
    \begin{align*}
        |w_{\mc I^\ell_j}(x) - t / 2^\ell| &\leq |w_{\mc I^\ell_j}(x) - w_{\mc I}(x) / 2^\ell| + \frac{1}{2^\ell}|w_{\mc I}(x) - t| \\
        & \leq \sqrt{\Delta_w / 2^\ell} \cdot 16 \log n + 2\sqrt{\Delta_w}/{2^\ell} \leq \sqrt{\Delta_w / 2^\ell} \cdot 17 \log n.
    \end{align*}
    Here the last step follows from $2^\ell \geq 1$ and $n \geq 10$. Since we set $\eta = 17 \log n$, the above implies that $w_{\mc I_j^\ell}(x) \in W^\ell$.
    Similarly, we can deduce from $|p_{\mc I}(x) -\vopt| \leq 2\sqrt{\Delta_p}$ that $p_{\mc I_j^\ell} (x) \in P^\ell$.
\end{proof}

Using \cref{lem:bmbm_total_concentration} we can argue that at level $\ell$ it suffices to compute the subarray of $D_j^\ell$ corresponding to indices in $W^\ell$ and values in $P^\ell$. We make this idea precise in \cref{lem:bmbm_properties}.

\begin{lemma}\label{lem:bmbm_properties}
    Let $x \in \{0, 1\}^n$ be a Pareto optimum of $\mc P_{\mc I}$ satisfying $|w_{\mc I}(x) - t| \leq 2 \sqrt{\Delta_w}$ and $|p_{\mc I}(x) -\vopt| \leq 2 \sqrt{\Delta_p} $.
    Then with probability at least $1 - 1/n^5$ we have for all $\ell \in \{0, \dots, q\}$ and all $j \in[2^\ell]$ that 
    $w_{\mc I_j^\ell}(x) \in W^\ell$, $p_{\mc I_j^\ell}(x) \in P^\ell$ and $C_j^\ell[w_{\mc I_j^\ell} (x)] = p_{\mc I_j^\ell} (x)$.
\end{lemma}
\begin{proof}
    By \cref{lem:bmbm_total_concentration}, for fixed $\ell \in \{0, \dots, q\}$ and $j \in[2^\ell]$ we have $w_{\mc I_j^\ell} (x) \in W^\ell$ and $p_{\mc I_j^\ell} (x) \in P^\ell$ with probability at least $1 - 1/n^7$.
    Since $2^q \leq n$ we can afford a union bound and deduce that $w_{\mc I_j^\ell} (x) \in W^\ell$ and $p_{\mc I_j^\ell} (x) \in P^\ell$ holds \emph{for all} $\ell \in \{0, \dots, q\}$ and \emph{for all} $j \in[2^\ell]$ with probability at least $1 - 1/n^5$.
    We condition on that event and prove by induction that $C_j^\ell[w_{\mc I_j^\ell} (x)] = p_{\mc I_j^\ell} (x)$ for all $\ell \in \{0, \dots, q\}$ and all $j \in [2^\ell]$.
    
    For the base case, fix $\ell = q$ and $j \in [2^\ell]$. 
    Recall that $\mc P_{\mc I_j^q}[k]$ is the maximum profit of a subset of items of $\mc I_j^q$ of weight at most $k$. 
    So if $y$ is such that $\mc P_{\mc I_j^q}[w_{\mc I_j^q}(x)] = p_{\mc I_j^q}(y)$ and $w_{\mc I_j^q}(y) \leq w_{\mc I_j^q}(x)$, then $p_{\mc I_j^q}(y) \geq p_{\mc I_j^q}(x)$. By \cref{lem:pareto_maximum}, since $x$ is a Pareto optimum of $\mc P_{\mc I}$, we deduce $p_{\mc I_j^q}(y) = p_{\mc I_j^q}(x)$. We have $w_{\mc I_j^q}(x) \in W^q$ and $\mc P_{\mc I_j^q}[w_{\mc I_j^q}(x)] = p_{\mc I_j^q}(x) \in P^q$, so by the construction in \cref{alg_line:bmbm_leaf_C} $C_j^q[w_{\mc I_j^q}(x)]  = p_{\mc I_j^q}(x)$.
    
    In the inductive step, fix $\ell < q$ and $j\in[2^\ell]$. 
    We want to prove that $D_j^\ell[w_{\mc I_j^\ell}] = p_{\mc I_j^\ell}(x)$. Indeed, since $w_{\mc I_j^\ell}(x) \in W^\ell$ and $p_{\mc I_j^\ell}(x) \in P^\ell$, this shows that $C_j^\ell[w_{\mc I_j^\ell}(x)] = D_j^\ell[w_{\mc I_j^\ell}] = p_{\mc I_j^\ell}(x)$. 
    By induction, $D_j^\ell[w_{\mc I_j^\ell}(x)]$ is the profit of some subset of items of $\mc I_j^\ell$ of weight at most $w_{\mc I_j^\ell}(x)$. So there exists $y \in \{0, 1\}^n$ such that $D_j^\ell[w_{\mc I_j^\ell}(x)] = p_{\mc I_j^\ell}(y)$ and $w_{\mc I_j^\ell}(y) \leq w_{\mc I_j^\ell}(x)$. Then
    \begin{align*}
        p_{\mc I_j^\ell}(y) = D_j^\ell[w_{\mc I_j^q}(x)] &= \max \left\{ C_{2j-1}^{\ell + 1}[k] + C_{2j}^{\ell + 1}[k'] \ : \ k + k' = w_{\mc I_j^\ell}(x) \right\}\\
        & \geq C_{2j-1}^{\ell + 1}[w_{\mc I_{2j-1}^{\ell+1}}(x)] + C_{2j}^{\ell + 1}[w_{\mc I_{2j}^{\ell+1}}(x)] \\
        & = p_{\mc I_{2j -1}^{\ell +1}}(x) + p_{\mc I_{2j}^{\ell +1}}(x) = p_{\mc I_{j}^{\ell}}(x)
    \end{align*}
    where we use the induction hypothesis and the fact that $\mc I_j^{\ell} = \mc I_{2j-1}^{\ell + 1} \cup \mc I_{2j}^{\ell + 1}$ is a partitioning. Recall that we interpret out-of-bound entries of arrays as $- \infty$.
    Since $x$ is a Pareto optimum of $\mc P_{\mc I}$, we obtain $p_{\mc I_j^\ell}(y) = p_{\mc I_j^\ell}(x)$ by \cref{lem:pareto_maximum}, and thus $D_j^\ell[w_{\mc I_j^\ell}(x)] = p_{\mc I_j^\ell}(x)$. This implies $C_j^\ell[w_{\mc I_j^q}(x)] = p_{\mc I_j^q}(x)$ as argued above.
\end{proof}

Finally, we can prove that \cref{alg:knapsack_bmbm} correctly computes $\mc P_{\mc I}[T ; P]$ as defined in \cref{lem:knapsack_bmbm} with high probability. 
Note that we can boost the success probability to any polynomial by repeating this algorithm and taking the entry-wise maximum of each computed sequence $C_1^0$.

\begin{lemma}[Correctness of \cref{alg:knapsack_bmbm}]\label{lem:bmbm_correctness}
    Let $T := [t - \sqrt{\Delta_w}, t + \sqrt{\Delta_w}]$ and $P := [\vopt- \sqrt{\Delta_p}, \vopt+ \sqrt{\Delta_p}]$. Then with probability at least $1 - 1/n$ we have $C_1^0[T ; P] = \mc P_{\mc I}[T ; P]$.
\end{lemma}

\begin{proof}
    First, observe that $T \subset W^0$ and $P \subset P^0$. 
    Let $K^0$ be the set of indices of $C_1^0$, i.e., $K^0 := \{k \ | \ k \in W^0, D_1^0[k] \in P^0\}$. Let $K$ be the interval such that $C_1^0[K] = C_1^0[T ; P]$, i.e., $K := \{k \ |\ k \in T, C_1^0[k] \in P \}$.
    We want to show that $C_1^0[k] = \mc P_{\mc I}[k]$ for every $k \in K$ with high probability. Since $C_1^0$ and $\mc P_{\mc I}$ are monotone non-decreasing (see \cref{lem:bmbm_monotonicity}), to compare $C_1^0$ and $\mc P_{\mc I}$ it is enough to focus on break points. Recall that $k \in \mathbb N$ is a break point of $\mc P_{\mc I}$ if $\mc P_{\mc I}[k - 1] < \mc P_{\mc I}[k]$, and that each break point $k$ corresponds to a Pareto optimum $x \in \{0, 1\}^n$ such that $w_{\mc I}(x) = k$ and $\mc P_{\mc I}[w_{\mc I}(x)] = p_{\mc I}(x)$. 
    We claim that for any $k \in K$ and $k' \leq k$ maximal such that $k'$ is a break point of $\mc P_{\mc I}$ we have $C_1^0[k'] = \mc P_{\mc I}[k']$. Together with monotonicity this proves that $C_1^0[k] = \mc P_{\mc I}[k]$ for all $k \in K$ as desired.
    
    To prove the claim, we first need to establish that every $k \in K$ has a break point $k' \leq k$ that is not too far, specifically $k' \geq t - 2 \sqrt{\Delta_w}$. We prove that $[t - 2 \sqrt{\Delta_w} , t - \sqrt{\Delta_w} ]$ contains a break point of $\mc P_{\mc I}$. Let $y \in \{0, 1\}^n$ be such that  $w_{\mc I}(y) \leq t - 2\sqrt{\Delta_w} $ and $\mc P_{\mc I}[t - 2\sqrt{\Delta_w}] = p_{\mc I}(y)$. Let $y'$ be $y$ with an additional item. This is always possible since we can assume without loss of generality that the total weight of all items in $\mc I$ exceeds $t$, i.e., any subset of items of weight at most $t$ leaves at least one item out. The additional item has weight at most $\wmax$ and profit at least $1$. So $w_{\mc I}(y') \leq w_{\mc I}(y) + \wmax \leq t - 2\sqrt{\Delta_w} + \wmax \leq t - \sqrt{\Delta_w}$ and $p_{\mc I}(y) < p_{\mc I}(y')$. In particular, we have $p_{\mc I}(y') \leq \mc P_{\mc I}[t - \sqrt{\Delta_w}]$. We obtain $\mc P_{\mc I}[t - 2\sqrt{\Delta_w}] = p_{\mc I}(y) < p_{\mc I}(y') \leq  \mc P_{\mc I}[t - \sqrt{\Delta_w}]$. Therefore, $[t - 2 \sqrt \Delta_w, t - \sqrt \Delta_w]$ contains a break point.
    
    Recall that our goal is to show that for any $k \in K$ and $k' \leq k$ maximal such that $k'$ is a break point of $\mc P_{\mc I}$ it holds that $C_1^0[k'] = \mc P_{\mc I}[k']$. Since we showed that $[t - 2 \sqrt{\Delta_w}, t - \sqrt{\Delta_w}]$ contains a break point, we define $T' := T \cup [t - 2 \sqrt{\Delta_w}, t - \sqrt{\Delta_w}] = [t - 2 \sqrt{\Delta_w}, t + \sqrt{\Delta_w}]$, and $K'$ such that $C_1^0[K'] = C_1^0[T' ; P]$, i.e., $K' := \{k \ |\ k \in T', C_1^0[k] \in P \}$. Then all it remains to show is that $C_1^0[k] = \mc P_{\mc I}[k]$ for every break point $k \in K'$. Fix a break point $k \in K'$ and let $x \in \{0, 1\}^n$ be the Pareto optimum such that $w_{\mc I}(x) = k$ and $\mc P_{\mc I}[k] = p_{\mc I}(x)$. Then in particular $w_{\mc I}(x) \in T'$ and $p_{\mc I}(x) \in P$, and thus $|w_{\mc I}(x) - t| \leq 2\sqrt \Delta_w$ and $|p_{\mc I}(x) -\vopt| \leq 2 \sqrt\Delta_p$.
    By \cref{lem:bmbm_properties}, this implies $C_1^0[k] = p_{\mc I}(x) = \mc P_{\mc I}[k]$ with probability at least $1 - 1/n^5$. 
    Since $|K'| \leq |P| \leq 2 \pmax \sqrt{n}$, by a union bound over all break points $k \in K'$, we obtain that $C_1^0[T ; P] = \mc P_{\mc I}[T; P]$ with probability at least $1 - 2 \pmax \sqrt{n} / n^5 \geq 1 - 2 \pmax / n^4$. 
    Note that if $n^3 \leq 2 \pmax$, then in particular $n^2 \leq 2 \pmax$ and we can use Bellman's dynamic program to compute the profit sequence in time $O(n \cdot t) = O(t \sqrt{\pmax})$ (see \cref{thm:Knapsack_DP}). Hence, we can assume that $n^3 \geq 2 \pmax$. Thus, with probability at least $1 - 2 \pmax /n^4 \geq 1 - 1/n$ we have $C_1^0[T ; P] = \mc P_{\mc I}[T; P]$. 
\end{proof}

\subsubsection{Running time of Algorithm \ref{alg:knapsack_bmbm}}

\begin{lemma}\label{lem:bmbm_running time_level}
    For a fixed level $\ell \in \{0, \dots, q - 1\}$ and iteration $j \in[2^\ell]$, the computation of $D_j^\ell$ in 
    \cref{alg_line:bmbm_combi_D} takes time $\Ot((t / 2^\ell)^{3/4}  \pmax^{1/2}  \wmax^{1/4})$.
\end{lemma}
\begin{proof}
    By \cref{lem:bmbm_monotonicity}, the sequences $C_{2j-1}^{\ell+1}$ and $C_{2j}^{\ell+1}$ are bounded monotone. Additionally, they have length at most $|W^{\ell+1}| = \Ot(\sqrt{\Delta_w / 2^\ell})$ and the values are in a range of length at most $|P^\ell| = \Ot(\sqrt{\Delta_p / 2^\ell})$.
    So their max-plus convolution can be computed using the algorithm of \cref{thm:MPConv} in time $\Ot( (\Delta_w / 2^{\ell})^{1/2}  (\Delta_p / 2^{\ell})^{1/4})$. We apply the definitions of $\Delta_w = t \wmax$ and $\Delta_p = \vopt \pmax$ and the balancedness assumption $t/ \wmax = \Theta(\vopt/ \pmax)$, which yields $\Delta_p = O(t\pmax^2 / \wmax)$, to bound the running time by $\Ot((t / 2^\ell)^{3/4}  \pmax^{1/2}  \wmax^{1/4})$.
\end{proof}

\begin{lemma}\label{lem:bmbm_runtime}
    \cref{alg:knapsack_bmbm} runs in time $\Ot(n + t \sqrt{\pmax})$.
\end{lemma}

\begin{proof}
    We first bound the running time of the base case, i.e., the computations of 
    \crefrange{alg_line:bmbm_start_base}{alg_line:bmbm_leaf_C}.
    For each $j \in[2^q]$, the array $D_j^q$ is obtained by computing the sequence $\mc P_{\mc I_j^q}\left[W^* ; P^*\right]$, where $W^* := \left[0, \frac{t}{2^q} + \sqrt{\frac{\Delta_w}{2^q}} \eta \right]$ and $P^* := \left[0, \frac{\vopt}{2^q} + \sqrt{\frac{\Delta_p}{2^q}} \eta \right]$. 
    Since $\Delta_w =t \wmax$, $\eta = O(\log n)$ and $2^q = \Theta(t / \wmax)$, we can bound $ \frac{t}{2^q} + \sqrt{\frac{\Delta_w}{2^q}} \eta = \Ot(\wmax)$, and analogously $\frac{\vopt}{2^q} + \sqrt{\frac{\Delta_p}{2^q}} \eta = \Ot(\pmax)$. Using \cref{thm:Knapsack_SubsetSum}, we can therefore compute $D_j^q$ in time $\Ot(|\mc I_j^q| + \wmax \sqrt{\pmax})$.
    Hence, the total running time of the base case is:
    \begin{align*}
        \sum_{j=1}^{2^q}  \Ot \left( |\mc I_j^q| + \wmax \cdot  \sqrt{\pmax} \right)  
        = \Ot\left( n + 2^q \cdot \wmax \cdot \sqrt{\pmax} \right) 
        = \Ot\left( n + t \cdot \sqrt{\pmax} \right)
    \end{align*}
    where we again used $2^q = \Theta(t / \wmax)$.
    
    Using \cref{lem:bmbm_running time_level}, we bound the running time of the combination step, i.e., the computations of \crefrange{alg_line:bmbm_start_combi}{alg_line:bmbm_combi_C}, as follows:
    \begin{align*}
    \sum_{\ell = 0}^{q - 1} \sum_{j = 1}^{2^\ell} \Ot\left( (t / 2^\ell)^{3/4}  \pmax^{1/2}  \wmax^{1/4}\right)
    &= \sum_{\ell = 0}^{q -1} \Ot\left( t^{3/4}  \pmax^{1/2}  (\wmax \cdot 2^\ell)^{1/4}\right)
    \end{align*}
    This is a geometric series, so it is bounded by $\Ot(t^{3/4}  \pmax^{1/2}  (\wmax \cdot 2^{q})^{1/4})$. Since $2^q \leq t / \wmax$ we obtain a running time of $\Ot(t \sqrt{\pmax})$. Hence, in total \cref{alg:knapsack_bmbm} takes time $\Ot(n + t \sqrt{\pmax})$.
\end{proof}

\subsection{\texorpdfstring{$\boldsymbol{\Ot(n + (n \wmax \pmax)^{1/3} t^{2/3})}$}{O(n + (n wmax pmax)\^1/3 t\^2/3)}-time algorithm}\label{sec:bmdp}

In this section we modify \cref{alg:knapsack_bmbm} to obtain an algorithm running in time $\Ot(n + (n \wmax \pmax)^{1/3} \cdot t^{2/3})$, thus proving \cref{lem:knapsack_bmdp}.

\knapsackBMDP*

We obtain \cref{lem:knapsack_bmdp} by replacing the algorithm used in the base case of \cref{alg:knapsack_bmbm}. Instead of using \cref{thm:Knapsack_SubsetSum}, which is derived from Bringmann and Cassis \cite{BringmannC22}, we use the algorithm of the following lemma.

\begin{lemma}\label{lem:knapsack_newleaf}
    For any Knapsack instance $(\mc I, t)$ and any $\ell \in \mathbb N, 2 \leq \ell \leq t$, the sequence $\mc P_{\mc I}[t - \ell \dots t + \ell]$ can be computed in time $\Ot(n \sqrt{t \cdot \wmax} + n \ell)$ by a randomized algorithm that is correct with probability at least $1-1/n$.
\end{lemma}

We obtain \cref{lem:knapsack_newleaf} by refining an idea of He and Xu~\cite{HeXu23} that they used to design a $\Ot(n^{3/2}\wmax)$-time Knapsack algorithm. Our refinement allows us to replace in this running time a factor $\sqrt{n \cdot \wmax}$ by a factor $\sqrt{t}$.
Actually, the algorithm proving \cref{lem:knapsack_newleaf} is very simple (see \cref{alg:knapsack_newleaf}). It first randomly permutes the items. Then it performs the Bellman's classic dynamic programming algorithm (\cref{thm:Knapsack_DP}), but computes only a portion of the DP table. More precisely, when processing the $i$-th item (in the random order), instead of computing the whole profit sequence $\mc P_{\{1,\ldots,i\}}[0 \dots t]$, it computes only the subarray of length $\Delta = \Ot(\ell + \sqrt{t \cdot \wmax})$ centered around $\frac{i}{n} \cdot t$, which is roughly the expected weight of an optimal solution restricted to the first $i$ items.

We note that our algorithm is almost identical to the algorithm of He and Xu, and the only difference is that they use $\Delta = \Ot(\sqrt{n} \cdot \wmax)$. We also follow their analysis -- the key difference is that our \cref{lem:sample} gives a stronger bound than an analogous bound of theirs.

\begin{algorithm}[h]
\caption{The $\Ot(n \sqrt{t \cdot \wmax} + n\ell)$-time algorithm of \cref{lem:knapsack_newleaf}. The input is a Knapsack instance $(\mc I, t)$ and a parameter $\ell \in \mathbb N, 2 \leq \ell \leq t$. Accessing a negative index or an uninitialized element in $C_{k-1}$ returns $-\infty$.}\label{alg:knapsack_newleaf}
$\sigma[1 \dots n] \gets $ random permutation of $\{1,\ldots,n\}$\\
$\Delta \gets \ell + \lceil 4 \sqrt{t \cdot \wmax \cdot \log (n\ell)} \rceil$ \\
$C_0[0] \gets 0$ \\
\For{$k = 1, \ldots, n$}{
  \For{$j = \frac{k}{n}t - \Delta, \ldots, \frac{k}{n}t + \Delta$}{
    $C_k[j] \gets \max(C_{k-1}[j], w_{\sigma[k]} + C_{k-1}[j - p_{\sigma[k]}])$ \\
  }
}
\For{$j = t - \ell + 1, \ldots, t + \ell$}{
  $C_n[j] \gets \max(C_n[j - 1], C_n[j])$ \\
}
\For{$j \in \{t - \ell, \ldots, t + \ell\} \cap [w_{\mc I}([n]), \infty)$}{
  $C_n[j] = p_{\mc I}([n])$ \label{alg_line:nl_edgecase} \tcp{edge case when all items fit into the knapsack}
}
\Return{$C_n[t - \ell \dots t + \ell]$}\\
\end{algorithm}

Before we argue about the correctness of the algorithm, let us recall the classic probabilistic inequality of Hoeffding.

\begin{lemma}[Hoeffding bound~\cite{Hoeffding63}]\label{lem:hoeffding}
Let $Y_1, Y_2, \ldots, Y_n$ be independent random variables, where $Y_i$ takes values from $[l_i, h_i]$, and let $S := Y_1 + Y_2 + \cdots + Y_n$ denote their sum. Then, for all $\lambda > 0$,
\[
  \mathbb{P}\bigl( | S - \mathbb{E}(S) | \geq \lambda \bigr)
  \leq
  2 \cdot \exp\left(\frac{-2t^2}{\sum_{i \in [n]}(h_i - l_i)^2}\right).
\]
\end{lemma}

Intuitively, we would like to apply this bound to the intersection of a prefix of a random permutation and a fixed optimal solution. However, the items in the prefix are not independent because the $i$-th prefix is formed by sampling $i$ indices from $[n]$ \emph{without replacement}. In the next lemma we circumvent this issue with a simple trick. We note that He and Xu~\cite{HeXu23} in their analysis use a weaker variant of Hoeffding's bound -- where all random variables share the same lower and upper bound -- which easily generalizes to the setting of samples without replacement, as already noticed by Hoeffding in his original paper~\cite{Hoeffding63}. Our lemma yields an analogous result for the stronger variant of the inequality -- with varying lower and upper bounds -- which is needed for achieving our improved running time.

\begin{lemma}\label{lem:sample}
Let $a_1, a_2, \ldots, a_n \in \mathbb{N}$, and let $X_1, X_2, \ldots, X_k$ be a $k$-element sample \emph{without replacement} from $\{a_1, a_2, \ldots, a_n\}$. Define $A := \sum_{i \in [n]} a_i$, and $a_{\max} := \max_{i \in [n]} a_i$. Fix $\delta \in (0, 1/4)$. Then, with probability at least $1 - \delta$ over the choice of the sample, it holds that
\[ \left| (X_1 + X_2 + \cdots + X_k) - \frac{k}{n} A \right| \leq \sqrt{A \cdot a_{\max} \log(n / \delta)}. \]
\end{lemma}

\begin{proof}
Let $B_1, B_2, \ldots, B_n$ be i.i.d.~Bernoulli random variables taking the value $1$ with probability $k/n$, and let $S := \sum_{i \in [n]} B_i \cdot a_i$. Observe that $\mathbb{E}(S) = \frac{k}{n}A$. We apply Hoeffding bound (\cref{lem:hoeffding}) to $Y_i = B_i a_i$ and $\lambda = \sqrt{Aa_{\max} \log(n / \delta)}$. Note that $l_i = 0$ and $h_i = a_i$, so we get
\[
  \mathbb{P}\bigl( |S - \mathbb{E}(S)| \geq \lambda \bigr)
  \leq
  2 \cdot \exp\left(\frac{-2 A a_{\max} \log(n / \delta)}{\sum_{i \in [n]} a_i^2}\right)
  \leq 
  2 \cdot \exp(-2 \log(n / \delta) )
  <
  2\frac{\delta^2}{n^2}.
\]
Let $N := \sum_{i \in [n]} B_i$. Observe that the random variable $S$ conditioned on the event $N = k$ has the same distribution as $X_1 + X_2 + \cdots + X_k$. On the other hand, $N$ is a binomial distribution with $n$ trials and success probability $k/n$; since $k$ is a mode of that distribution (see, e.g., \cite{KaasB80}), and the support size is $n+1$, it follows that $N=k$ holds with probability at least $1/(n+1)$. We get that 
\begin{align*}
\mathbb{P}\left(
  \left| (X_1 + X_2 + \cdots + X_k) - \frac{k}{n}A \right| \geq \lambda
\right)
& =
\mathbb{P}\bigl( |S - \mathbb{E}(S) | \geq \lambda \bigm| N=k \bigr) \\
& = \mathbb{P}\bigl( |S - \mathbb{E}(S) | \geq \lambda \land N=k \bigr) \bigm/ \mathbb{P}(N=k) \\
& \leq \mathbb{P}\bigl( |S - \mathbb{E}(S) | \geq \lambda \bigr) \ / \ \mathbb{P}(N=k) \\
& < 2\frac{\delta^2}{n^2} \cdot (n+1) < \delta,
\end{align*}
with the last inequality following from the assumption that $\delta < 1/4$.
\end{proof}

Now we are ready to argue about the desired properties of \cref{alg:knapsack_newleaf}.

\begin{proof}[Proof of \cref{lem:knapsack_newleaf}]
\cref{alg:knapsack_newleaf} clearly runs in time $O(n\Delta) = \Ot(n \sqrt{t \cdot \wmax} + n\ell)$. It remains to prove that with probability at least $1 - 1/n$ it returns a correct answer.

For every target $t' \in \{t - \ell, \dots, t + \ell\}$ let us fix an optimal solution $x^{(t')} \in \{0,1\}^n$ of profit $\mc P_{\mc I}[t']$ and weight at most $t'$.
When $t' \geq w_{\mc I}([n])$, all the items together fit into the knapsack of capacity $t'$, and thus $\mc P_{\mc I}[t'] = p_{\mc I}([n])$. This edge case is handled in \cref{alg_line:nl_edgecase}. From now on focus on the case where $t' < w_{\mc I}([n])$. In that case we have $w_{\mc I}(x^{(t')}) \in (t' - \wmax, t']$. For every such $t'$ and for every $k \in [n]$, we apply \cref{lem:sample} to $a_i = x^{(t')}_i \cdot w_i$ and $X_i = a_{\sigma[i]}$ with $\delta = 1 / (n^2 \cdot (2 \ell + 1))$, and we conclude that with probability at least $1-\delta$ it holds that
\begin{align*}
w_{\{1,\ldots,k\}}(x^{(t')}) & \in \left[\frac{k}{n}w_{\mc I}(x^{(t')})  \pm \sqrt{t \cdot \wmax \cdot \log(n/\delta)} \right] \\
& \subseteq \left[\frac{k}{n}w_{\mc I}(x^{(t')}) \pm 3 \sqrt{t \cdot \wmax \cdot \log(n\ell)} \right]  \qquad\qquad\ \ (\text{using } n/\delta \leq n^3 \cdot \ell^3) \\
& \subseteq \left[\frac{k}{n}t' \pm 4 \sqrt{t \cdot \wmax \cdot \log(n\ell)} \right] \qquad\qquad (\text{using } w_{\mc I}(x^{(t')}) > t' - \wmax) \\
& \subseteq \left[\frac{k}{n}t \pm \bigl(\ell + 4 \sqrt{t \cdot \wmax \cdot \log(n\ell)} \bigr) \right] = \left[\frac{k}{n}t \pm \Delta \right]. 
\end{align*}
By a union bound, with probability at least $1 - 1/n$ this holds simultaneously for all such $t'$'s ($2\ell + 1$ of them) and $k$'s ($n$ of them). Let us condition on this event. It follows, by induction on $k$, that
\[C_k[w_{\{1,\ldots,k\}}(x^{(t')})] = p_{\{1,\ldots,k\}}(x^{(t')})\]
for every $t'$ and $k$. In particular for $k=n$ we have $w_{\{1,\ldots,n\}}(x^{(t')}) = w_{\mc I}(x^{(t')}) \leq t'$ and $C_n[t'] \geq C_n[w_{\mc I}(x^{(t')})] = p_{\mc I}(x^{(t')}) = \mc P_{\mc I}[t']$. On the other hand, clearly $C_n[t'] \leqslant \mc P_{\mc I}[t']$, so they are equal, which finishes the proof.
\end{proof}

We can now explain how to modify \cref{alg:knapsack_bmbm} to obtain \cref{lem:knapsack_bmdp}.

\begin{proof}[Proof of \cref{lem:knapsack_bmdp}]
Let $c := \min\{1, \tfrac \vopt \pmax \cdot \tfrac \wmax t\}$, and note that $c = \Theta(1)$ by the balancedness assumption and $\vopt = \Theta(\OPT)$. 
If $n \geq c \cdot t \sqrt{\pmax} / \wmax$, then $\Ot(n + t \sqrt{\pmax}) \leq \Ot(n + (n \wmax \pmax)^{1/3} \cdot t^{2/3})$ and thus \cref{lem:knapsack_bmdp} follows from \cref{lem:knapsack_bmbm}. 
Additionally, if $n^3 \leq 2 \pmax$ then 
Bellman's dynamic programming algorithm computes the complete sequence 
$\mc P_{\mc I}$ in time $O(n \cdot t) \le O((n \wmax \pmax)^{1/3} t^{2/3})$. 
In the remainder we can thus assume $n \leq c \cdot t \sqrt{\pmax} / \wmax$ and $2 \pmax \leq n^3$. In this case, we modify the algorithm of \cref{alg:knapsack_bmbm} as follows.

Let $q$ be the largest integer such that 
$2^q \leq \max\{1, n^{4/3} \cdot (\wmax / t)^{1/3} \cdot \pmax^{-2/3} \}$.
Consider the modification of \cref{alg:knapsack_bmbm} using the new value of $q$ and replacing the computation of $D_j^q$ in \cref{alg_line:bmbm_leaf_D} by the computation of $D_j^q := \mc P_{\mc I_j^q}\left[W^q \right]$ using \cref{alg:knapsack_newleaf} of \cref{lem:knapsack_newleaf}. As a reminder, we defined $\eta := 17 \log n$ and
$W^q := \left[\frac{t}{2^{q}} - \sqrt{\frac{t \cdot \wmax}{2^{q}}} \eta,\ \frac{t}{2^{q}} + \sqrt{\frac{t \cdot \wmax}{2^{q}}} \eta\right]$. Hence, we call \cref{alg:knapsack_newleaf} with the Knapsack instance $(\mc I_j^q, t/ 2^q)$ and parameter $\ell = \sqrt{t \cdot \wmax / 2^q} \cdot \eta$.
So each computation of \cref{alg_line:bmbm_leaf_D} now takes time $\Ot(|\mc I_j^q| \sqrt{t \cdot \wmax / 2^q})$. In total, \cref{alg_line:bmbm_leaf_D} takes time
\begin{align*}
    \sum_{j = 1}^{2^q} \Ot\left(|\mc I_j^q|\sqrt{t \cdot \wmax / 2^q} \right) = \Ot\left(n \sqrt{t \cdot \wmax /2^q} \right) \le \Ot\left((n \wmax \pmax)^{1/3} t^{2/3}\right),
\end{align*}
where the last step follows from the inequality $2^q \ge n^{4/3} \cdot (\wmax / t)^{1/3} \cdot \pmax^{-2/3} / 2$, which holds by our choice of $q$.

If $2^q = 1$, then no combination steps are performed. Otherwise, we have $2^q \le n^{4/3} \cdot (\wmax / t)^{1/3} \cdot \pmax^{-2/3}$. 
In this case, for the combination levels the same analysis as in \cref{lem:bmbm_runtime} shows that the total running time of all combination steps is $\Ot(t^{3/4} \cdot \pmax^{1/2} \cdot (\wmax 2^q)^{1/4}) \le \Ot\left((n \wmax \pmax)^{1/3} t^{2/3}\right)$.

The correctness argument of \cref{alg:knapsack_bmbm} works verbatim because all used inequalities on $2^q$ still hold, specifically we have $1 \le 2^q \le n$, $2^q \le t/\wmax$, and $2^q \le \vopt / \pmax$. We verify these inequalities in the remainder of this proof. 
If $2^q = 1$ then these inequalities are trivially satisfied. Otherwise, we have $1 < 2^q \le n^{4/3} \cdot (\wmax / t)^{1/3} \cdot \pmax^{-2/3}$. 
Then obviously $2^q \ge 1$. 
Since $n \leq c \cdot t \sqrt{\pmax} / \wmax$, by rearranging we have $n^{4/3} \cdot (\wmax / t)^{1/3} \cdot \pmax^{-2/3} \le c^{4/3} t / \wmax$, and thus $2^q \le c^{4/3} t / \wmax \le c \cdot t / \wmax$. Since $c \le 1$, we obtain $2^q \le t / \wmax \le n$. Since $c \le \tfrac \vopt \pmax \cdot \tfrac \wmax t$, we obtain $2^q \le \vopt / \pmax$.
Finally, the correctness argument of \cref{alg:knapsack_bmbm} additionally uses the bound $n^3 \geq 2 \pmax$, which we can assume as discussed above. 
\end{proof}

\subsection{Reconstructing an optimal solution}\label{sec:reconstruction}

The above algorithms are returning the optimal profit $\OPT$ of a given Knapsack instance. From that output we can reconstruct a solution $x \in \{0, 1\}^n$ such that $p_{\mc I}(x) = \OPT$ and $w_{\mc I}(x) \leq t$. 
Indeed, after running \cref{alg:knapsack_bmbm} we obtain the sequences $C_j^\ell$ for $\ell \in \{0, \dots, q\}$ and $j \in [2^\ell]$. For the output $C_1^0[t]$ we can find a witness $i \in W^0$ such that $C_1^1[i] + C_2^1[t - i] = C_1^0[t]$. This can be done in time $O(|W^0|)$ by simply trying all possibilities. We continue to search witnesses for $C_1^1[i]$ and $C_2^1[t - i]$ recursively. 
In the end, we reach one entry for each array $C_j^q$ for $j \in [2^q]$. In total this takes time $\sum_{\ell = 0}^q 2^\ell \cdot O(| W^0|) = \Ot(\sqrt{\Delta_w 2^q}) = \Ot(t)$, where we use $\Delta_w = t \cdot \wmax$ and $2^q \leq t / \wmax$.  Finally, each array $C_j^q$ is computed via the algorithm from \cref{thm:Knapsack_SubsetSum}, for which the solution reconstruction is described in \cite{BringmannC22} and does not add any extra overhead on the total running time of \cref{thm:Knapsack_SubsetSum}.
The same method can be used for the modified \cref{alg:knapsack_bmbm} presented in \cref{sec:bmdp}.

\section{Balancing instances}\label{sec:reduction}

In this section we show how a general Knapsack instance can be reduced, in randomized $\Ot(n + \min\{\wmax\sqrt{\pmax},\pmax\sqrt{\wmax}\})$ time, to an instance where all items have the same profit-to-weight ratio $\frac{p_i}{w_i}$ up to a constant multiplicative factor. In particular, this implies that $t / \wmax  = \Theta (\OPT / \pmax)$, i.e., the reduced instance satisfies the \emph{balancedness assumption}. This reduction combined with the algorithms for balanced instances (\cref{sec:knapsack_algos,sec:sym_knapsack}) gives us algorithms for the general case.

For notational convenience, in this section we denote solutions to the Knapsack problem by subsets of items (and not by indicator vectors, as in the rest of the paper). For a subset of items $J \subseteq [n]$, let $w(J)$ denote their total weight, i.e., $w(J) = \sum_{i \in J} w_i$, and $p(J)$ their total profit, i.e., $p(J) = \sum_{i \in J} p_i$.

Assume that the items are sorted in non-increasing order of profit-to-weight ratios $\frac{p_i}{w_i}$, and consider the \emph{maximum prefix solution} (as defined in~\cite{PRW21}), i.e., the solution consisting of items $P = \{1, \ldots, j\}$ for maximum $j$ such that $w_1 + \cdots + w_j \leq t$. Let $\rho = \frac{p_j}{w_j}$ be the profit-to-weight ratio of the last (i.e., the least profitable) item in this solution. Partition the items into three groups based on their profit-to-weight ratios:
\begin{itemize}
\item \emph{good} items, with profit-to-weight ratio above $2\rho$, $G = \{i \in [n] \mid \frac{p_i}{w_i} > 2\rho\}$;
\item \emph{bad} items, with the ratio below $\rho/2$, $B = \{i \in [n] \mid \frac{p_i}{w_i} < \rho/2\}$;
\item \emph{medium} items, the remaining ones, $M = \{i \in [n] \mid \frac{p_i}{w_i} \in [\rho/2, 2\rho]\}$.
\end{itemize}
Fix any optimal solution $Z \subseteq [n]$. We claim that the symmetric difference with the maximum prefix solution $(P \setminus Z) \cup (Z \setminus P)$ restricted to good and bad items is \emph{small}, both in terms of the total weight and the total profit. More precisely, let $\Delta = \big((P \setminus Z) \cup (Z \setminus P)\big) \cap (G \cup B)$; we claim that:

\begin{claim}
\label{cla:smalldelta}
$w(\Delta) \leq 10 \wmax$ and $p(\Delta) \leq 10 \pmax$.
\end{claim}

\begin{proof}
We focus on the claim regarding weights; the proof for the profits part is symmetric.

We prove the claim by contradiction. Assume that $w(\Delta) > 10 \wmax$. Note that all good items are included in the maximum prefix solution, and there are no bad items there, i.e., $G \subseteq P$ and $B \cap P = \emptyset$. Therefore, 
\[\Delta
= \big((P \setminus Z) \cup (Z \setminus P)\big) \cap (G \cup B)
= (G \setminus Z) \cup (B \cap Z),
\]
i.e., the difference $\Delta$ consists of good items that are not in the optimal solution, and of bad items that are in the optimal solution.
Consider partitioning $\Delta$ into
$\Delta_G = G \setminus Z$ and $\Delta_B = B \cap Z$.
There are two cases: either $w(\Delta_G) > 5 \wmax$, or $w(\Delta_B) > 5 \wmax$. Focus on the first case for now.

Fix a subset of items $\Delta_+ \subseteq \Delta_G$ such that their total weight satisfies $w(\Delta_+) \in [2\wmax, 3\wmax)$. Such a subset must exist because we can keep picking items from $\Delta_G$ one by one until their total weight reaches at least $2\wmax$. We cannot run out of items in this process because $w(\Delta_G) > 5\wmax > 2\wmax$. Since the last picked item contributes at most $\wmax$ to the total weight, the picked subset cannot overshoot $3\wmax$.

In the special case when all items fit into the knapsack (i.e., $w([n]) \leq t$) we have $P = Z = [n]$ and the claim trivially holds. When this is not the case, the total weights of both $P$ and $Z$ are within $\wmax$ from the knapsack capacity $t$, i.e., $w(P), w(Z) \in (t - \wmax, t]$, because otherwise one could add an item and improve one of these two solutions. In particular, $|w(P) - w(Z)| < \wmax$, and it follows that $w(Z \setminus P) > w(P \setminus Z) - \wmax \geq w(\Delta_G) - \wmax > 4 \wmax$, where the second inequality holds because $\Delta_G \subseteq P \setminus Z$. Knowing that $w(Z \setminus P) > 4 \wmax$, we can fix a subset of items $\Delta_- \subseteq Z \setminus P$ with total weight $w(\Delta_-) \in [3\wmax, 4\wmax)$.

Consider the solution $Z' = Z \cup \Delta_+ \setminus \Delta_-$. It is feasible because $w(\Delta_-) \geq 3\wmax > w(\Delta_+)$ and thus $w(Z') < w(Z) \leq t$. Recall that all items in $\Delta_+$ have profit-to-weight ratio above $2\rho$. Therefore, $p(\Delta_+) > 2\rho \cdot w(\Delta_+) \ge 2\rho \cdot 2\wmax$. On the other hand, all items in $\Delta_-$ have profit-to-weight ratio at most $\rho$, and so the total profit of $\Delta_-$ is at most $\rho \cdot w(\Delta_-) < 4\rho \cdot \wmax$. We conclude that $p(\Delta_+) > p(\Delta_-)$, and hence $p(Z') = p(Z) + p(\Delta_+) - p(\Delta_-) > p(Z)$, which contradicts the optimality of $Z$ and ends the proof for the case where $w(\Delta_G) > 5 \wmax$.

Let us now consider the remaining case, where $w(\Delta_B) > 5 \wmax$. The argument is symmetric. We fix a subset of items $\Delta_- \subseteq \Delta_B$ with total weight $w(\Delta_-) \in [3\wmax, 4\wmax)$. Then, since $|w(P) - w(Z)| < \wmax$, we have that 
$w(P \setminus Z) > w(Z \setminus P) - \wmax \geq w(\Delta_B) - \wmax > 4 \wmax$, so we can select a subset $\Delta_+ \subseteq P \setminus Z$ with total weight $w(\Delta_+) \in [2\wmax, 3\wmax)$. Consider the solution $Z' = Z \cup \Delta_+ \setminus \Delta_-$. As in previous case, $Z'$ is feasible because $w(\Delta_-) > w(\Delta_+)$. Now all items in $\Delta_+$ have profit-to-weight ratio at least $\rho$, and hence $p(\Delta_+) \geq \rho \cdot w(\Delta_+) \geq 2\rho \cdot \wmax$, while all items in $\Delta_-$ have profit-to-weight ratio at most $\rho / 2$, and so $p(\Delta_-)$ is at most $(\rho / 2) \cdot w(\Delta_-) < (\rho / 2) \cdot 4 \wmax = 2 \rho \cdot \wmax$. Again $p(Z') > p(Z)$, contradicting the optimality of $Z$.
\end{proof}

\subparagraph*{Algorithmic application of \cref{cla:smalldelta}.}

Now we use \cref{cla:smalldelta} to reduce any Knapsack instance to a balanced instance, in randomized time $\Ot(n + \min\{\wmax\sqrt{\pmax},\pmax\sqrt{\wmax}\})$, as desired. First, construct a minimum prefix solution $P$, and split the items into sets $G$, $M$, $B$ based on their profit-to-weight ratios. This takes $O(n)$ time~\cite[Section 2.1]{PRW21}.

Any optimal solution $Z$ to the original instance can be expressed as
\[Z = (G \setminus \Delta_G) \cup Z_M \cup \Delta_B,\]
where $\Delta_G \subseteq G$, $Z_M \subseteq M$, $\Delta_B \subseteq B$, and (thanks to \cref{cla:smalldelta}) it holds that
\begin{align*}
w(\Delta_G) &\leq 10\wmax, & p(\Delta_G) &\leq 10\pmax, \\
w(\Delta_B) &\leq 10\wmax, & p(\Delta_B) &\leq 10\pmax, \\
|w(Z_M) - w(P \cap M)| &\leq 11\wmax, & |p(Z_M) - p(P \cap M)| &\leq 11\pmax.
\end{align*}

The bounds for $Z_M$ follow from the observation that $P$ and $Z$ can differ in total weight and profit by at most $\wmax$ and $\pmax$, respectively, and $\Delta_G$ and $\Delta_B$ add at most $10\wmax$ and $10\pmax$ to that difference. Hence, knowing only $P$, we can already estimate $w(Z_M)$ (and $p(Z_M)$) up to an $O(\wmax)$ (and $O(\pmax)$, respectively) additive term.

Profit-to-weight ratios of medium items $M$ differ by at most the multiplicative factor of~$4$. The set of items $M$ together with the capacity $t' = w(P\cap M)$ constitutes the new reduced balanced instance. It has to be solved for all integer knapsack capacities in the range
\[[w(P \cap M) - 11\wmax, w(P \cap M) + 11\wmax] \cap [0, t]\]
or all integer profits in the range
\[[p(P \cap M) - 11\pmax, p(P \cap M) + 11\pmax] \cap [0, +\infty).\]

We remark at this point (and formally argue in the proof of \cref{thm:knapsack_bmbm,thm:knapsack_bmbm_sym} below) that algorithms in Section~\ref{sec:knapsack_algos} can output optimal solutions in these ranges without increasing their running times. Also, note that all parameters ($n$, $\wmax$, $\pmax$, $t$, $\OPT$) of the new instance $(M,t')$ are no larger than those of the original instance.

Once instance $(M,t')$ is solved for all these capacities, one can infer a solution to the original instance using the following approach. First, solve the Knapsack instance consisting of bad items $B$ for all targets in $[0, 10\wmax]$ and with profits capped to $10\pmax$. This can be done, using \cref{thm:Knapsack_SubsetSum,thm:Knapsack_SubsetSum_sym}, in time $\Ot(n + \min\{\wmax\sqrt{\pmax},\pmax\sqrt{\wmax}\})$, as desired.
Then, for every integer value $t'' \in [0, 10\wmax]$, we find the least profitable subset of $G$ with total weight at least $t''$. This computation can be reduced to a Knapsack instance with profits and weights swapped (see~\cite[Section 4]{PRW21}), so again it can be done in time $\Ot(n + \min\{\wmax\sqrt{\pmax},\pmax\sqrt{\wmax}\})$.

What remains to be done is finding an optimal way to combine the solutions of these three instances. This boils down to two max-plus convolutions of monotone sequences of length $O(\wmax)$ with values in a range of size $O(\pmax)$. Using \cref{thm:MPConv} it takes time $\Ot(\min\{\wmax\sqrt{\pmax},\pmax\sqrt{\wmax}\})$. Summarizing, we proved the following lemma.

\LemRedToBal*

Now we are ready to prove our theorems giving algorithms for the general (not necessarily balanced) case.

\begin{proof}[Proof of \cref{thm:knapsack_bmbm,thm:knapsack_bmbm_sym}]
Combining the reduction of \cref{lem:reduction_balanced_knapsack} with algorithms for balanced instances of \cref{lem:knapsack_bmbm,lem:knapsack_bmbm_sym} immediately gives \cref{thm:knapsack_bmbm,thm:knapsack_bmbm_sym}, i.e., algorithms for general instances with running times $\Ot(n + t \sqrt{\pmax})$ and $\Ot(n + \OPT \sqrt{\wmax})$, respectively. This is straightforward because, w.l.o.g., $\wmax \leq t$ and $\pmax \leq \OPT$, and hence the running time of the reduction is not larger than the running time of the algorithms.

We note that the algorithms of \cref{lem:knapsack_bmbm,lem:knapsack_bmbm_sym} output only solutions in intervals of lengths $2\sqrt{t \cdot \wmax} \geq 2\wmax$ and $2\sqrt{\OPT \cdot \pmax} \geq 2\pmax$, while \cref{lem:reduction_balanced_knapsack} requires intervals of lengths $c \cdot \wmax$ and $c \cdot \pmax$ for a constant $c > 2$ hidden in the asymptotic notation. This is however not an issue since we can always add two dummy items, one with weight $c\cdot\wmax$ and profit $0$ and another with weight $0$ and profit $c\cdot\pmax$, before calling any of the algorithms of \cref{lem:knapsack_bmbm,lem:knapsack_bmbm_sym} and then remove the latter item from the returned solution.
\end{proof}

\begin{proof}[Proof of \cref{thm:knapsack_bmdp,thm:knapsack_bmdp_sym}]
Similarly to the previous proof, we can combine \cref{lem:reduction_balanced_knapsack} with \cref{lem:knapsack_bmdp,lem:knapsack_bmdp_sym} to get \cref{thm:knapsack_bmdp,thm:knapsack_bmdp_sym}. It may seem that in this case the running time of the reduction could dominate the $\Ot(n + (n\wmax\pmax)^{1/3} t^{2/3})$ running time of the algorithm; this is however not the case if we fall back to the $O(nt)$ time Bellman's algorithm in the parameter regime where it is faster. Indeed,
\begin{align*}
\min\{\wmax \sqrt{\pmax}, nt\} & \leq
(\wmax\sqrt{\pmax})^{2/3} \cdot (nt)^{1/3} \\ & =
(n\wmax\pmax)^{1/3} \wmax^{1/3} t^{1/3} \leq 
(n\wmax\pmax)^{1/3} t^{2/3}.
\end{align*}
\end{proof}

\section{Reduction from bounded min-plus convolution verification}\label{sec:lower_bound}

In this section we show that all our Knapsack algorithms are optimal (up to subpolynomial factors) under the assumption that verifying min-plus convolution with entries bounded by $n$ requires quadratic time. 

\ThmLB*

\begin{proof}
    Recall (\cref{def:mpconvver}) that in the verification problem we are given three integer sequences 
    $a[0 \ldots n-1]$, $b[0 \ldots n-1]$, $c[0 \ldots 2n-2]$ with non-negative entries bounded by $n$ and we want to check if
    \[\forall_{k} \ c[k] \leq \min_{\mathrel{i+j=k}} \{a[i] + b[j]\},\]
    which is equivalent to
    \[\forall_{i+j=k} \ c[k] \leq a[i] + b[j].\tag{$\star$}\]

    First, create three sequences $A[0 \ldots n-1]$, $B[0 \ldots n-1]$, $C[0 \ldots 2n-2]$ such that $A[i] = a[i] + in$, $B[j] = b[j]+jn$, $C[k]=c[k]+kn$. The entries are now bounded by $2n^2$ but the sequences are non-decreasing. Crucially, the condition $(\star)$, which we want to verify, holds for $A$, $B$, $C$, if and only if it holds for $a$, $b$, $c$. In other words, this is a reduction from verifying min-plus convolution with entries bounded by $n$ to verifying monotone min-plus convolution with entries bounded by $O(n^2)$. Intuitively, as we shall see later in the proof, the monotonicity allows us to work with the condition ``$i+j \geq k$'' instead of ``$i + j = k$'', which is more useful for reducing to Knapsack.

    Now, create a Knapsack instance with $4n-1$ items consisting of:
    \begin{itemize}
        \item an item with weight $5n - i$ and profit $2n^2 - A[i]$, for each $i = 0, \ldots, n-1$;
        \item an item with weight $10n - j$ and profit $10n^2 - B[j]$, for each $j = 0, \ldots, n-1$;
        \item an item with weight $20n + k$ and profit $100n^2 + C[k] $, for each $k = 0, \ldots, 2n-2$.
    \end{itemize}
    We will call these items $A$-items, $B$-items, and $C$-items, respectively. Set the knapsack capacity to $t = 35n$. We show that:
    
    \begin{claim}
    \label{cla:lb}
      Condition $(\star)$ holds if and only if $\OPT \leq 112n^2$.
    \end{claim}

    \begin{claimproof}
    The ``if'' direction is straightforward. We prove the contraposition. If $(\star)$ is not true, there must exist $i$ and $j$ such that $c[i+j] > a[i] + b[j]$, and hence also $C[i+j] > A[i] + B[j]$. Consider the knapsack solution consisting of the $A$-item corresponding to $i$, the $B$-item corresponding to $j$, and the $C$-item corresponding to $k=i+j$. The total weight of this solution is exactly $35n$, and the total profit equals $112n^2 + C[i+j] - (A[i] + B[j]) > 112n^2$, hence $\OPT > 112n^2$.

    For the rest of the proof we focus on the ``only if'' direction. The proof for this direction is again for the contraposition: we show that $\OPT > 112n^2$ implies that $(\star)$ does not hold. To this end, we first claim that any optimal solution with total profit exceeding $112n^2$ consists of exactly one item of each of the three groups.
    
    Indeed, there is at most one $C$-item, because two would not fit in the capacity. In order to see that there is at least one $C$-item, observe that every $A$-item has profit-to-weight ratio at most $\frac{1}{2}n$, and every $B$-item has profit-to-weight ratio at most $\frac{10}{9}n$. Hence, a solution without $C$-items would have total profit at most $\frac{10}{9}n \cdot t < 40n^2$, which is less than the profit of a single $C$-item. Moreover, each single $C$-item fits the knapsack. Thus, there must be at least one (hence, exactly one) $C$-item in any optimal solution.

    The remaining capacity for $A$-items and $B$-items is between $13n$ and $15n$. Hence, at most one $B$-item fits. If there are no $B$-items,
    there can be at most three $A$-items, contributing at most $3 \cdot 2n^2$ to the total profit,
    which is less than the profit of any $B$-item,
    and each single $B$-item fits in the remaining capacity. Hence, there is at least one (and thus exactly one) $B$-item in any optimal solution.

    The remaining capacity for $A$-items is between $3n$ and $6n$, so at most one $A$-item fits. However, a single $B$-item and a single $C$-item can have the total profit at most $112n^2$, so there must be at least one (and thus exactly one) $A$-item.

    We conclude that there is indeed exactly one item of each type in the solution. Let us fix $i$, $j$, $k$ to the corresponding indices of these items. The knapsack capacity $t$ guarantees that $k \leq i + j$, and the total profit guarantees that $C[k] > A[i] + B[j]$. Since $C$ is non-decreasing, we have that $C[i + j] \geq C[k]$ and thus $C[i+j] > A[i]+B[j]$. This implies that condition $(\star)$ does not hold, which finishes ends the proof of \cref{cla:lb}.
    \end{claimproof}

    The Knapsack instance that we consider has $\wmax = O(n)$, $\pmax = O(n^2)$, $t = O(n)$, $\OPT = O(n^2)$. Any algorithm polynomially faster than $O(t\sqrt{\pmax})$ or $O((n\wmax\pmax)^{1/3}t^{2/3})$ would run on such instances in time truly subquadratic in $n$, implying a truly subquadratic time for the Bounded Monotone Min-Plus Convolution Verification problem.

    To prove the optimality of the symmetric running times, we consider an analogous instance with the role of indices and values swapped. Namely, the instance consists of the following $4n-1$ items:
    \begin{itemize}
        \item an item with weight $5n^2 + A[i]$ and profit $2n + i$, for each $i = 0, \ldots, n-1$
        \item an item with weight $10n^2 + B[j]$ and profit $10n + j$, for each $j = 0, \ldots, n-1$;
        \item an item with weight $20n^2 - C[k]$ and profit $100n - k$, for each $k = 0, \ldots, 2n-2$.
    \end{itemize}
    We set the knapsack capacity to $t = 35n^2 - 1$, and we claim that condition $(\star)$ holds if and only if $\OPT < 112n$. This claim can be proved by following the proof of \cref{cla:lb}. This instance has $\wmax = O(n^2)$, $\pmax = O(n)$, $t = O(n^2)$, and $\OPT = O(n)$. Hence, any algorithm polynomially faster than $O(\OPT \sqrt{\wmax})$ or $O((n\wmax\pmax)^{1/3}\OPT^{2/3})$ would run on such instances in $O(n^{2-\varepsilon})$ time.
\end{proof}

\section{From one monotone to two monotone sequences}\label{sec:1to2MinConv}

The algorithm of \cref{thm:MPConv} for max-plus convolution requires that both input sequences are monotone. In this section we give a black-box reduction that shows that it is enough to require that one sequence is monotone. We note that this is just a side result, which we believe might be of independent interest, while all our Knapsack algorithms already produce convolution instances with both sequences being monotone.

Recall that the max-plus convolution of sequences $A=A[0 \ldots n-1], B=B[0 \ldots n-1] \in \mathbb{Z}^n$ is defined as the sequence $C=C[0 \ldots 2n-2] \in \mathbb{Z}^{2n-1}$ with $C[k] = \max_{i + j = k} \{A[i] + B[j]\}$, where the maximum ranges over all values of $0 \le i,j < n$ with $i+j=k$, and out-of-bounds entries of $A$ and $B$ are interpreted as $-\infty$.
To have a succinct notation, in this section we denote this sequence $C$ by $A \star B$.

We study the case that both sequences $A,B$ have values in a bounded range~$\{0, 1, \dots, M\}$, and that at least one of the sequences $A,B$ is monotone. The following theorem shows that the special case of the max-plus convolution problem where one of the sequences is monotone is essentially equivalent to the special case where both sequences are monotone.

\ThmOneTwoMinPlus*

  

Note that if $T_2(n,M)$ is of the form $\Theta(n^\alpha f(M))$, then for $\alpha > 1$ the recurrence solves to $T_1(n,M) = O(T_2(n,M))$, and for $\alpha = 1$ the recurrence solves to $T_1(n,M) = O(T_2(n,M) \log n)$.

The remainder of this section is devoted to the proof of this theorem. So suppose we have an algorithm $\mathcal{A}_2$ that computes the max-plus convolution of two monotone non-decreasing sequences in time $T_2(n,M)$.

\subparagraph*{Warmup: first half of the sequence.} 
As a warmup we start with a slightly simplified definition of the max-plus convolution problem where we only want to compute the first half of the max-plus convolution of $A$ and $B$, more precisely given $A,B \in \{0, 1, \dots, M\}^n$ we want to compute the sequence $C' := (A \star B)[0 \dots n-1]$. This case illustrates the main reduction step, and we later generalize this to the complete sequence $A \star B$.

Denote by $\textup{pm}(B)$ the prefix-maxima sequence of $B$. More precisely we define $\textup{pm}(B)[i] := \max\{ B[j] : 0 \le j \le i \}$ for any $0 \le i < n$. 
We claim that the first halves of the sequences $A \star \textup{pm}(B)$ and $A \star B$ coincide:
\begin{claim} \label{cla:correctnessmonotonereduciton}
  We have $(A \star \textup{pm}(B))[k] = (A \star B)[k]$ for any $0 \le k < n$.
\end{claim}
This claim immediately gives our algorithm for computing $C'$:
Since $\textup{pm}(B)$ is monotone non-decreasing, we can use algorithm $\mathcal{A}_2$ to compute $A \star \textup{pm}(B)$. By the claim, this yields the desired sequence $C' = (A \star B)[0 \dots n-1] = (A \star \textup{pm}(B))[0 \dots n-1]$. 
Since the sequence $\textup{pm}(B)$ can be computed in time $O(n)$, and the call to $\mathcal{A}_2$ takes time $T_2(n,M) \ge \Omega(n)$, the total running time is bounded by $T_1(n,M) = O(T_2(n,M))$ (which trivially satisfies the desired recurrence).

It remains to prove the claim.

\begin{proof}[Proof of Claim~\ref{cla:correctnessmonotonereduciton}]
  Since $\textup{pm}(B)[i] \ge B[i]$ for all $i$, we clearly have $(A \star \textup{pm}(B))[k] \ge (A \star B)[k]$ for all $k$.
  For the opposite direction, fix an index $0 \le k < n$ and consider a right-most witness $(i,j)$ of $(A \star \textup{pm}(B))[k]$, i.e., among all pairs $(i,j)$ with $i + j = k$ and $A[i] + \textup{pm}(B)[j] = (A \star \textup{pm}(B))[k]$ pick the pair with maximum $i$, or equivalently minimum $j$.

  First consider the case $\textup{pm}(B)[j] = B[j]$. Then we obtain $(A \star B)[k] \ge A[i] + B[j] = A[i] + \textup{pm}(B)[j] = (A \star \textup{pm}(B))[k]$. Since we argued the opposite inequality before, we obtain $(A \star B)[k] = (A \star \textup{pm}(B))[k]$.

  Now consider the remaining case $\textup{pm}(B)[j] > B[j]$; we will see that this leads to a contradiction. Note that the prefix maximum $\textup{pm}(B)[j]$ corresponds to some entry $j' < j$ with $B[j'] = \textup{pm}(B)[j'] = \textup{pm}(B)[j]$. Consider the pair $(i',j')$ for $i' := i + (j-j')$. By construction we have $i'+j' = i+j = k$. Since $A$ is monotone non-decreasing and $i' \ge i$ we have $A[i'] \ge A[i]$, and by assumption we have $\textup{pm}(B)[j'] = \textup{pm}(B)[j]$. But then $A[i'] + \textup{pm}(B)[j'] \ge A[i] + \textup{pm}(B)[j] = (A \star \textup{pm}(B))[k]$, contradicting the choice of $(i,j)$ as a right-most witness. (Indeed, if the inequality is strict then $(i,j)$ is no witness, and if the inequality is an equality then $(i,j)$ is not right-most.)
  Observe that here we assumed that $i'$ is a valid index in $A$, i.e., we assumed $i' < n$. This is guaranteed for any $k < n$, since $i+j=k$ implies $i' = i + j - j' = k - j' \le k < n$.  (For $k \ge n$ the constructed index $i'$ would not necessarily exist in $A$, and thus the proof would break.)
\end{proof}

\subparagraph*{Complete sequence.}
Now we want to compute the complete sequence $A \star B$, not just its first half.
We assume for simplicity that the sequence length $n$ is a power of 2 (we discuss how to relax this assumption in the paragraph following the proof of Claim~\ref{cla:onetwomonotoneclaim}).

We split sequence $A$ into two halves denoted by $A_1[0 \dots n/2 - 1] := A[0 \dots n/2 - 1]$ and $A_2[0 \dots n/2 - 1] := A[n/2 \dots n-1]$, and similarly we split $B$ into two halves $B_1[0 \dots n/2 - 1] := B[0 \dots n/2 - 1]$ and $B_2[0 \dots n/2 - 1] := B[n/2 \dots n-1]$. 
Define the sequence $C \in \mathbb{N}^{2n-1}$ by setting for any $0 \le k < 2n-1$:
$$ C[k] := \max\{ (A_1 \star \textup{pm}(B_1))[k], (A_1 \star \textup{pm}(B_2))[k-n/2], (A_2 \star B_1)[k-n/2], (A_2 \star B_2)[k-n] \}, $$
where out-of-bounds entries are ignored (i.e., they are interpreted as $- \infty$). 
We claim that $C = A \star B$.
\begin{claim} \label{cla:onetwomonotoneclaim}
  We have $C[k] = (A \star B)[k]$ for any $0 \le k < 2n-1$.
\end{claim}
This identity immediately yields our algorithm for computing $A \star B$: 
Given $A,B \in \{0, 1, \dots, M\}^n$ such that $A$ is monotone non-decreasing, we split $A,B$ into $A_1,A_2,B_1,B_2$ as above. We compute $\textup{pm}(B_1)$ and $\textup{pm}(B_2)$ in time $O(n)$. Since $A_1$, $\textup{pm}(B_1)$, and $\textup{pm}(B_2)$ are monotone non-decreasing, we can use algorithm $\mathcal{A}_2$ to compute $A_1 \star \textup{pm}(B_1)$ and to compute $A_1 \star \textup{pm}(B_2)$. We recursively compute $A_2 \star B_1$ and we recursively compute $A_2 \star B_2$. Finally, we combine these sequences as in the definition of $C$ above to obtain the desired sequence $C = A \star B$.

Correctness is immediate from Claim~\ref{cla:onetwomonotoneclaim}. 
Denote the running time of this algorithm by $T_1(n,M)$ (not to be confused with the running time $T_2(n,M)$ of algorithm $\mathcal{A}_2$). Note that the two calls to algorithm $\mathcal{A}_2$ take time $O(T_2(n,M))$ (where we used monotonicity of $T$), the two recursive calls take time $2\, T_1(n/2,M)$, and the remaining steps take a negligible time of $O(n)$. We thus obtain the desired recurrence $T_1(n,M) \le 2\, T_1(n/2,M) + O(T_2(n,M))$.

It remains to prove the claim.

\begin{proof}[Proof of Claim~\ref{cla:onetwomonotoneclaim}]
  From the definition of the max-plus convolution we observe that
  \begin{align} 
    (A \star B)[k] = \max\{ (A_1 \star B_1)[k], (A_1 \star B_2)[k-n/2], (A_2 \star B_1)[k-n/2], (A_2 \star B_2)[k-n] \}.  \label{eq:AstarBdecomp}
  \end{align}
  Since $\textup{pm}(B_i)[j] \ge B_i[j]$ for all $i,j$, it follows that $C[k] \ge (A \star B)[k]$ for all $k$. 
  It remains to prove that $C[k] \le (A \star B)[k]$ for each $0 \le k < 2n-1$. We consider the four cases in the definition of $C[k]$.
  
  \proofsubparagraph{Case 1: $C[k] = (A_2 \star B_1)[k-n/2]$.} Then we have $C[k] = (A_2 \star B_1)[k-n/2] \le (A \star B)[k]$ by equation (\ref{eq:AstarBdecomp}), showing the desired inequality. 
  
  \proofsubparagraph{Case 2: $C[k] = (A_2 \star B_2)[k-n]$.} Analogous to the previous case.
  
  \proofsubparagraph{Case 3: $C[k] = (A_1 \star \textup{pm}(B_1))[k]$.} Consider a right-most witness $(i,j)$ of $(A_1 \star \textup{pm}(B_1))[k]$, i.e., among all pairs $(i,j)$ with $i+j=k$ and $A_1[i] + \textup{pm}(B_1)[j] = (A_1 \star \textup{pm}(B_1))[k]$ pick the pair with maximum $i$, or equivalently minimum $j$. 
  If $\textup{pm}(B_1)[j] = B_1[j]$ then we have $C[k] = (A_1 \star \textup{pm}(B_1))[k] = A_1[i] + \textup{pm}(B_1)[j] = A_1[i] + B_1[j] \le (A_1 \star B_1)[k] \le (A \star B)[k]$, which shows the desired inequality.

  Otherwise, we have $\textup{pm}(B_1)[j] > B_1[j]$. Then $\textup{pm}(B_1)[j]$ corresponds to some entry $j' < j$ with $B[j'] = \textup{pm}(B_1)[j'] = \textup{pm}(B_1)[j]$. Now consider the pair $(i',j')$ for $i' := i + (j - j')$. By construction we have $i' + j' = i + j = k$. Since $A$ is monotone non-decreasing and $i' \ge i$ we have $A[i'] \ge A[i]$. Hence, we have
  \begin{multline} \label{eq:someeq}
      A[i'] + B_1[j'] = A[i'] + \textup{pm}(B_1)[j'] \\ \ge A[i] + \textup{pm}(B_1)[j] = A_1[i] + \textup{pm}(B_1)[j] = (A_1 \star \textup{pm}(B_1))[k] = C[k].
  \end{multline}
  Note that $i' = i + j - j' = k - j' \le k < n$, where the last step uses the assumption of Case 3 and that the range of indices for $A_1 \star \textup{pm}(B_1)$ is from 0 to $n-2$. 
  
  If $i' < n/2$, then $A[i'] = A_1[i']$.
  Then we have
  $$ (A_1 \star \textup{pm}(B_1))[k] \ge A_1[i'] + \textup{pm}(B_1)[j'] \ge C[k] = (A_1 \star \textup{pm}(B_1))[k], $$
  where the first step follows from the definition of the max-plus convolution and $i'+j'=k$, the second step follows from (\ref{eq:someeq}) and $A[i'] = A_1[i']$, and the last step from the assumption of Case 3. 
  It follows that $(i',j')$ is a witness for $(A_1 \star \textup{pm}(B_1))[k]$ with $j' < j$, contradicting that $(i,j)$ is a right-most witness. 
  
  Otherwise, if $n/2 \le i' < n$, then $A[i'] = A_2[i'-n/2]$. 
  Then we have
  $$ (A_2 \star B_1)[k-n/2] \ge A_2[i'-n/2] + B_1[j'] \ge C[k] \ge (A_2 \star B_1)[k-n/2], $$
  where the first step follows from the definition of the max-plus convolution and $i'+j'=k$, the second step follows from (\ref{eq:someeq}) and $A[i'] = A_2[i'-n/2]$, and the last step from the definition of $C$.
  It follows that $C[k] = (A_2 \star B_1)[k-n/2]$, meaning we are in Case 1, which we already handled. 

  \proofsubparagraph{Case 4: $C[k] = (A_1 \star \textup{pm}(B_2))[k-n/2]$.} Analogous to the previous case.
\end{proof}

\subparagraph*{When the length is no power of two.}
Above we assumed the length~$n$ to be a power of 2. 
Now consider general $n$. If $n$ is even, then we can use the same recursion on two halves verbatim as above. If $n$ is odd, then we first split $A$ into $A' := A[0 \dots n-2]$ and the single entry $A[n-1]$, and similarly we split $B$ into $B' := B[0 \dots n-2]$ and $B[n-1]$. Since $A'$ and $B'$ have even length, we can compute $A' \star B'$ by the same recursion on two halves verbatim as above. We then use that for any $0 \le k < 2n-1$ we have 
$$(A \star B)[k] = \max\{ (A' \star B')[k], A[n-1] + B[k-(n-1)], A[k-(n-1)] + B[n-1]\}, $$
where out-of-bounds entries are ignored (i.e., interpreted as $-\infty$). 
This allows to compute $A \star B$ from $A' \star B'$ in time $O(n)$, which is negligible. Therefore, the same algorithm and analysis goes through. Since we assume the time bound $T_2(n,M)$ to be monotone in $n$, we can replace terms of the form $\lfloor n/2 \rfloor$ by $n/2$ in the analysis, so we obtain the same recurrence as before.

\subparagraph*{Max-plus vs min-plus.}
The max-plus convolution and min-plus convolution problems are equivalent via the following simple reduction: Given sequences $A,B \in \{0, 1, \dots, M\}^n$, replace every entry $A[i]$ by $A'[i] := M - A[i]$ and $B[j]$ by $B'[j] := M - B[j]$. Then, denoting by $C$ the max-plus convolution of $A$ and $B$ and by $C'$ the min-plus convolution of $A'$ and $B'$, we have $C[k] = 2M + 2 - C'[k]$ for all $k$. The opposite direction is analogous. 

To prove Theorem~\ref{thm:onetwomonotone} for min-plus convolution, we first apply this equivalence of the min-plus convolution and max-plus convolution problems to turn an algorithm for min-plus convolution on two monotone sequences into an algorithm for max-plus convolution on two monotone sequences, then we apply Theorem~\ref{thm:onetwomonotone} (for max-plus convolution) to obtain an algorithm for max-plus convolution on one monotone sequence, and then we apply the equivalence again to obtain an algorithm for min-plus convolution on one monotone sequence.

\subparagraph*{Non-decreasing vs non-increasing.}
A simple reduction shows that non-decreasing and non-increasing are equivalent properties for the max-plus  problem: Given sequences $A, B \in \{0, 1, \dots, M\}^n$, replace every entry $A[i]$ by $A'[i] := A[n - 1 - i]$ and $B[j]$ by $B'[j] := B[n - 1 - j]$. Then we have $(A \star B)[k] = (A' \star B')[2n - 2 - k]$ for any $k$. Moreover, if $A$ is monotone non-decreasing then $A'$ is monotone non-increasing and vice versa, and similarly for $B$. 

To prove Theorem~\ref{thm:onetwomonotone} for non-increasing sequences, we first apply this equivalence of non-decreasing and non-increasing to turn an algorithm for max-plus convolution on two non-increasing sequences into an algorithm for max-plus convolution on two non-decreasing sequences, then we apply Theorem~\ref{thm:onetwomonotone} (for non-decreasing) to obtain an algorithm for max-plus convolution on one non-decreasing sequence, and then we apply the equivalence of non-decreasing and non-increasing again to obtain an algorithm for max-plus convolution on one non-increasing sequence.

\vskip 1em
\noindent
Combining these two reductions yields Theorem~\ref{thm:onetwomonotone} for min-plus convolution on non-increasing sequences.

\bibliography{ref}

\appendix

\section{Additional Knapsack algorithms}\label{sec:sym_knapsack}

In this section we explain how to adapt the algorithms of \cref{sec:knapsack_algos} to obtain symmetric running times. To this end define the weight sequence $\mc W_{\mc J}$ as
$$
\mc W_{\mc J}[k] := \min \{w_{\mc J}(x) \ | \ x \in \{0, 1\}^n, p_{\mc J}(x) \geq k\}
$$
for any subset $\mc J \subset \mc I$ and index $k \in \mathbb N$. A standard method to compute $\mc W_{\mc I}$ is to use dynamic programming.
\begin{fact}[{Analog to \cref{thm:Knapsack_DP}}]\label{thm:Knapsack_DP_sym}
    Let $(\mc I, t)$ be a Knapsack instance. For any $k \in \mathbb N$, the sequence $\mc W_{\mc I}[0 \dots k]$ can be computed in time $O(|\mc I| \cdot k)$.
\end{fact}
\begin{proof}
    We have $\mc W_{\emptyset}[0] = 0$ and $\mc W_{\{1,\ldots,i\}} = \min \{ \mc W_{\{1,\ldots,i-1\}}[k], w_i + \mc W_{\{1,\ldots,i-1\}}[k - p_i] \}$.
    So we can compute $\mc W_{\mc I}[0 \dots k]$ using dynamic programming in time $O(|\mc I| \cdot k)$.
\end{proof}

Notice that $\mc W_{\mc I}$ is monotone non-decreasing. In \cref{sec:adapt_subsetsum_knapsack} we show how to adapt \cite[Theorem 19]{BringmannC22} to compute a subarray of $\mc W_{\mc I}$ using rectangular bounded monotone min-plus convolutions. This yields the following \cref{thm:Knapsack_SubsetSum_sym}.

\begin{theorem}[{Analog to \cref{thm:Knapsack_SubsetSum}}]\label{thm:Knapsack_SubsetSum_sym}
    Let $(\mc I, t)$ be a Knapsack instance and fix $v \in \mathbb N$. Consider the intervals $T := [0 \dots t]$ and $V := [0 \dots v]$. Then we can compute $\mc W_{\mc I}[V ; T]$ with high probability in time $\Ot(n + v \sqrt{t})$. 
\end{theorem}

\subparagraph*{Pareto optimum of $\boldsymbol{\mc W_{\mc I}}$.}
We define Pareto optima of $\mc W_{\mc I}$ analogously to Pareto optima of $\mc P_{\mc I}$ defined in \cref{sec:preliminaries}.
The sequence $\mc W_{\mc I}$ is monotone non-decreasing, so we can define the \emph{break points} of $\mc W_{\mc I}$ as the integers $k \in \mathbb N$ such that $\mc W_{\mc I}[k] < \mc W_{\mc I}[k + 1]$.\footnote{Please note the difference with a break point of $\mc P_{\mc I}$, which is defined as $k \in \mathbb N$ such that $\mc P_{\mc I}[k -1] < \mc P_{\mc I}[k]$.}
For every break point $k \in \mathbb N$ of $\mc W_{\mc I}$, there exists $x \in \{0, 1\}^n$ with $p_{\mc I}(x) = k$ and $\mc W_{\mc I}[p_{\mc I}(x)] = w_{\mc I}(x)$. We call such a vector a \emph{Pareto optimum} of $\mc W_{\mc I}$. 
Indeed, by the definition of $\mc W_{\mc I}$, if a vector $y \in \{0, 1\}^n$ has higher profit $p_{\mc I}(y) > p_{\mc I}(x)$ then it necessarily has higher weight $w_{\mc I}(y) > \mc W_{\mc I}[p_{\mc I}(x)] = w_{\mc I}(x)$. We observe the following property of Pareto optima of $\mc W_{\mc I}$. 

\begin{lemma}\label{lem:pareto_maximum_Wextended}
    Let $x \in \{0, 1\}^n$ be a Pareto optimum of $\mc W_{\mc I}$.
    Let $\mc J \subset \mc I$ and $y \in \{0, 1\}^n$ such that $w_{\mc J}(y) \leq w_{\mc J}(x)$ and $p_{\mc J}(y) \geq p_{\mc J}(x)$. Then $w_{\mc J}(x) = w_{\mc I}(y)$. 
\end{lemma}
\begin{proof}    
    Observe that since $x$ is a Pareto optimum of $\mc W_{\mc I}$, we have $w_{\mc I}(x) = \mc W_{\mc I}[p_{\mc I}(x)]$. In particular $w_{\mc I}(x) \leq w_{\mc I}(y)$ for any $y \in \{0, 1\}^n$ such that $p_{\mc I}(x) \leq p_{\mc I}(y)$. 
    Suppose for the sake of contradiction that $w_{\mc J}(y) < w_{\mc J}(x)$. Consider the vector $y'$ that is equal to $y$ on $\mc J$ and equal to $x$ on $\mc I \setminus \mc J$. Then $w_{\mc I}(y') = w_{\mc J}(y) + w_{\mc I \setminus \mc J}(x) < w_{\mc J}(x) + w_{\mc I \setminus \mc J}(x) = w_{\mc I}(x)$. We also have $p_{\mc I} (y') = p_{\mc J}(y) + p_{\mc I \setminus \mc J}(x) \geq p_{\mc J}(x) + p_{\mc I \setminus \mc J}(x) = p_{\mc I}(x)$. This contradicts $x$ being a Pareto optimum.
\end{proof}

\medskip
In the remaining of this section, we prove the following \cref{lem:knapsack_bmbm_sym,lem:knapsack_bmdp_sym} for balanced Knapsack instances, i.e., instances satisfying $t / \wmax = \Theta(\OPT / \pmax)$. By \cref{lem:reduction_balanced_knapsack} proven in \cref{sec:reduction}, any Knapsack instance can be reduced to a balanced instance. Hence, combining \cref{lem:reduction_balanced_knapsack} with \cref{lem:knapsack_bmbm_sym,lem:knapsack_bmdp_sym}, yields \cref{thm:knapsack_bmbm_sym,thm:knapsack_bmdp_sym}.

\begin{restatable}{lemma}{knapsackBMBMsym}
    \label{lem:knapsack_bmbm_sym}
        For any Knapsack instance $(\mc I, t)$ satisfying $t / \wmax  = \Theta (\OPT / \pmax)$ the sequence $\mc W_{\mc I}[P ; T]$ for $T := [t - \sqrt{t \cdot \wmax}, t + \sqrt{t \cdot \wmax}]$, $P := [\vopt - \sqrt{\vopt \cdot \pmax}, \allowbreak \vopt + \sqrt{\vopt \cdot \pmax}]$ and $\OPT \leq \vopt \leq \OPT + \pmax$ can be computed by a randomized algorithm in time $\Ot(n + OPT \sqrt{\wmax})$.
\end{restatable}

\begin{restatable}{lemma}{knapsackBMDPsym}
    \label{lem:knapsack_bmdp_sym}
        For any Knapsack instance $(\mc I, t)$ satisfying $t / \wmax  = \Theta (\OPT / \pmax)$ the sequence $\mc W_{\mc I}[P ; T]$ for $T := [t - \sqrt{t \cdot \wmax}, t + \sqrt{t \cdot \wmax}]$, $P := [\vopt - \sqrt{\vopt \cdot \pmax}, \allowbreak \vopt + \sqrt{\vopt \cdot \pmax}]$ and $\OPT \leq \vopt \leq \OPT + \pmax$ can be computed by a randomized algorithm in time $\Ot(n + (n \wmax \pmax)^{1/3} \cdot OPT^{2/3})$.
\end{restatable} 

Observe that, with the notations of \cref{lem:knapsack_bmbm_sym,lem:knapsack_bmdp_sym}, $\mc W_{\mc I}[\OPT] = t$, $\OPT \in P$ and $t \in T$. So $\OPT$ is the maximum index $k \in P$ such that $\mc W_{\mc I}[k] = t$, which we can find with a binary search.

\subsection{\texorpdfstring{$\boldsymbol{\Ot(n + \OPT \sqrt{\wmax})}$}{O(n + OPT sqrt(wmax))}-time Algorithm}\label{sec:bmbm_sym}
We prove \cref{lem:knapsack_bmbm_sym} by presenting \cref{alg:knapsack_bmbm_sym}, which has a similar structure to \cref{alg:knapsack_bmbm}. 
It uses the same values as \cref{alg:knapsack_bmbm} for the parameters $q$, $\eta$, $\Delta_w$, $\Delta_p$, $W^*$, $P^*$, and $W^\ell$ and $P^\ell$ for $\ell \in \{0, \dots, q\}$. The algorithm starts by randomly splitting the items of $\mc I$ into $2^q$ groups $\mc I_1^q, \dots, \mc I_{2^q}^q$. Using \cref{thm:Knapsack_SubsetSum_sym}, a subarray of $\mc W_{\mc I_{j}^q}$ is computed for every $j \in [2^q]$. Finally, the arrays are combined in a tree-like fashion by computing their min-plus convolution. Similar properties of $\mc P_{\mc I_j^q}$ hold for $\mc W_{\mc I_{j}^q}$, as we show below. Notice that the only difference between \cref{alg:knapsack_bmbm} and \cref{alg:knapsack_bmbm_sym} is the computation of $\mc W_{\mc I}$ instead of $\mc P_{\mc I}$ and accordingly the use of min-plus convolutions instead of max-plus convolutions to aggregate the sequences. 
Since we compute min-plus convolutions, we interpret out-of-bound entries of arrays as $+ \infty$, instead of $- \infty$ as was the case when computing max-plus convolutions.

\begin{algorithm}
\caption{The $\Ot( n + \OPT \sqrt{\wmax})$-time algorithm of \cref{lem:knapsack_bmbm_sym}. The input $(\mc I, t)$ is a Knapsack instance such that $t / \wmax = \Theta(\OPT / \pmax)$. The algorithm is analogous to \cref{alg:knapsack_bmbm} by replacing $\mc P_{\mc J}$ by $\mc W_{\mc J}$ and replacing the max-plus convolutions by min-plus convolutions.
}\label{alg:knapsack_bmbm_sym}
    $\wmax \gets \max_{i \in [n]} w_i$ \\
    $\pmax \gets \max_{i \in [n]} p_i$ \\
    Compute an approximation $\vopt$ of $\textup{\OPT}$ using \cite[Theorem 2.5.4]{KPP04}. \\
    $q \gets $ largest integer such that $2^q \leq \min\{ t / \wmax, \vopt/\pmax\}$ \label{alg_line:bmbm_sym_def_q} \\
    $\eta \gets 17 \log n$ \\
    $\Delta_w \gets t \cdot \wmax$ \\
    $\Delta_p \gets \vopt \cdot \pmax$\\
    $\mc I_1^q, \dots, \mc I_{2^q}^q \gets $ random partitioning of $\mc I$ into $2^q$ groups \\

    $W^q \gets \left[\frac{t}{2^{q}} - \sqrt{\frac{\Delta_w}{2^{q}}} \eta,\ \frac{t}{2^{q}} + \sqrt{\frac{\Delta_w}{2^{q}}} \eta\right]$ \\
    $P^q \gets \left[\frac{\vopt}{2^{q}} - \sqrt{\frac{\Delta_p}{2^{q}}} \eta,\ \frac{\vopt}{2^{q}} + \sqrt{\frac{\Delta_p}{2^{q}}} \eta\right]$ \\
    $W^* \gets \left[0, \frac{t}{2^q} + \sqrt{\frac{\Delta_w}{2^q}} \eta \right]$ \\
    $P^* \gets \left[0, \frac{\vopt}{2^q} + \sqrt{\frac{\Delta_p}{2^q}} \eta \right]$ \\
    \For{$j = 1, \dots, 2^q$}{\label{alg_line:bmbm_sym_start_base}
        Compute $D_j^q \gets {\mc W_{\mc I_j^q}\left[P^* ; W^*\right]}$ using \cref{thm:Knapsack_SubsetSum_sym} \label{alg_line:bmbm_sym_leaf_D}
        \\
        $C_j^q \gets D_j^q[P^q ; W^q]$ \label{alg_line:bmbm_sym_leaf_C} \\
    }
    \For{$\ell = q-1, \dots, 0 $}{\label{alg_line:bmbm_sym_start_combi}
        $W^\ell \gets \left[\frac{t}{2^{\ell}} - \sqrt{\frac{\Delta_w}{2^{\ell}}} \eta,\ \frac{t}{2^{\ell}} + \sqrt{\frac{\Delta_w}{2^{\ell}}} \eta\right]$ \\
        $P^\ell \gets \left[\frac{\vopt}{2^{\ell}} - \sqrt{\frac{\Delta_p}{2^{\ell}}} \eta,\ \frac{\vopt}{2^{\ell}} + \sqrt{\frac{\Delta_p}{2^{\ell}}} \eta\right]$ \\
        \For{$j = 1, \dots, 2^\ell$}{
            $D_j^\ell \gets \minconv{C_{2j -1}^{\ell + 1}}{C_{2j}^{\ell + 1}}$ using \cref{thm:MPConv} \label{alg_line:bmbm_sym_combi_D} \\
            $C_j^\ell \gets D_j^\ell[P^\ell ; W^\ell]$ \label{alg_line:bmbm_sym_combi_C} \\
            }
        }
    $T \gets [t - \sqrt{t \cdot \wmax}, t + \sqrt{t \cdot \wmax}]$ \\
    $P \gets [\vopt- \sqrt{\vopt\cdot \pmax}, \vopt+ \sqrt{\vopt\cdot \pmax}]$ \\
    \Return{$C_1^0[P ; T]$}
\end{algorithm}

\subsubsection{Correctness of Algorithm \ref{alg:knapsack_bmbm_sym}}
For the rest of this section, fix a Knapsack instance $(\mc I, t)$ with $n := |\mc I|$ and such that $t / \wmax = \Theta(\OPT / \pmax)$.
First, observe that we use the same definition for $q$ as in \cref{alg:knapsack_bmbm}, so $1 \leq 2^q \leq n$ and thus $2^q$ is a valid choice for the number of groups in which we split the item set $\mc I$. Next, we claim that the subarray $D_j^\ell[P^\ell ; W^\ell]$ constructed in \cref{alg_line:bmbm_sym_leaf_D,alg_line:bmbm_sym_combi_D} is monotone non-decreasing. 

\begin{lemma}[Analog to \cref{lem:bmbm_monotonicity}]\label{lem:bmbm_sym_monotonicity}
    For every level $\ell \in \{0, \dots, q\}$ and iteration $j \in[2^\ell]$,  
    the sequence $C_j^\ell$ is monotone non-decreasing.
\end{lemma}
\begin{proof}
    For $\ell = q$ and $j \in[2^q]$, $D_j^q$ is a subarray of $\mc W_{\mc I_j^q}$, which is monotone non-decreasing by definition.    
    Hence $D_j^q$ is monotone non-decreasing, and since $W^q$ and $P^q$ are intervals, the array $C_j^q = D_j^q[P^q ; W^q]$ is also monotone non-decreasing. The statement follows from induction by noting that the min-plus convolution of two monotone non-decreasing sequences is a monotone non-decreasing sequence.
\end{proof}

\cref{lem:bmbm_sym_monotonicity} justifies the use of \cref{thm:MPConv} to compute the min-plus convolution in \cref{alg_line:bmbm_sym_combi_D}. 
We explain why it is enough to only compute the entries of the sequence $D_j^\ell$ corresponding to indices in $P^\ell$ and values in $W^\ell$.
Note that \cref{lem:bmbm_double_concentration,lem:bmbm_total_concentration} hold for any random partition of $\mc I$ into $2^q$ groups $\mc I_1^q, \dots, \mc I_{2^q}^q$ such that $\mc I_j^\ell = \mc I_{2j + 1}^{\ell +1} \cup \mc I_{2j}^{\ell +1}$ for any $\ell \in \{0, \dots, q - 1\}$ and $j \in [2^\ell]$. In particular, the proofs of \cref{lem:bmbm_double_concentration,lem:bmbm_total_concentration} for \cref{alg:knapsack_bmbm} hold verbatim for \cref{alg:knapsack_bmbm_sym}. 
We can thus use \cref{lem:bmbm_total_concentration} to prove the following \cref{lem:bmbm_sym_properties}.

\begin{lemma}[Analog to \cref{lem:bmbm_properties}]\label{lem:bmbm_sym_properties}
    Let $x \in \{0, 1\}^n$ be a Pareto optimum of $\mc W_{\mc I}$ satisfying $|w_{\mc I}(x) - t| \leq 2\sqrt{\Delta_w}$ and $|p_{\mc I}(x) - v| \leq 2\sqrt{\Delta_p}$.
    Then with probability at least $1 - 1/n^5$ we have for all $\ell \in \{0, \dots, q\}$ and all $j \in[2^\ell]$ that $w_{\mc I_j^\ell}(x) \in W^\ell$, $p_{\mc I_j^\ell}(x) \in P^\ell$ and $C_j^\ell[p_{\mc I_j^\ell} (x)] = w_{\mc I_j^\ell} (x)$.
\end{lemma}
\begin{proof}
     By \cref{lem:bmbm_total_concentration}, for fixed $\ell \in \{0, \dots, q\}$ and $j \in[2^\ell]$ we have $w_{\mc I_j^\ell} (x) \in W^\ell$ and $p_{\mc I_j^\ell} (x) \in P^\ell$ with probability at least $1 - 1/n^7$.
    Since $2^q \leq n$ we can afford a union bound and deduce that $w_{\mc I_j^\ell} (x) \in W^\ell$ and $p_{\mc I_j^\ell} (x) \in P^\ell$ holds \emph{for all} $\ell \in \{0, \dots, q\}$ and \emph{for all} $j \in[2^\ell]$ with probability at least $1 - 1/n^5$.
    We condition on that event and prove by induction that $C_j^\ell[w_{\mc I_j^\ell} (x)] = p_{\mc I_j^\ell} (x)$ for all $\ell \in \{0, \dots, q\}$ and all $j \in [2^\ell]$.
    
    For the base case, fix $\ell = q$ and $j \in [2^\ell]$. 
    Recall that $\mc W_{\mc I_j^q}[k]$ is the minimum weight of a subset of items of $\mc I_j^q$ of profit at least $k$. 
    Let $y$ be such that $\mc W_{\mc I_j^q}[p_{\mc I_j^q}(x)] = w_{\mc I_j^q}(y)$ and $p_{\mc I_j^q}(y) \leq p_{\mc I_j^q}(x)$, then $w_{\mc I_j^q}(y) \geq w_{\mc I_j^q}(x)$. 
    By \cref{lem:pareto_maximum_Wextended}, since $x$ is a Pareto optimum of $\mc W_{\mc I}$, we deduce $w_{\mc I_j^q}(y) = w_{\mc I_j^q}(x)$. We have $p_{\mc I_j^q}(x) \in P^q$ and $\mc W_{\mc I_j^q}[p_{\mc I_j^q}(x)] = w_{\mc I_j^q}(x) \in W^q$, so by the construction in \cref{alg_line:bmbm_sym_leaf_C} $C_j^q[p_{\mc I_j^q}(x)]  = w_{\mc I_j^q}(x)$.
    
    In the inductive step, fix $\ell < q$ and $j\in[2^\ell]$. 
    We want to prove that $D_j^\ell[p_{\mc I_j^\ell}(x)] = w_{\mc I_j^\ell}(x)$. Indeed, since $w_{\mc I_j^\ell}(x) \in W^\ell$ and $p_{\mc I_j^\ell}(x) \in P^\ell$, this shows that $C_j^\ell[p_{\mc I_j^\ell}(x)] = D_j^\ell[p_{\mc I_j^\ell}(x)] = w_{\mc I_j^\ell}(x)$. 
    By induction, $D_j^\ell[p_{\mc I_j^\ell}(x)]$ is the profit of some subset of items of $\mc I_j^\ell$ of profit at least $p_{\mc I_j^\ell}(x)$. So there exists $y \in \{0, 1\}^n$ such that $D_j^\ell[p_{\mc I_j^\ell}(x)] = w_{\mc I_j^\ell}(y)$ and $p_{\mc I_j^\ell}(y) \geq p_{\mc I_j^\ell}(x)$. Then
    \begin{align*}
        w_{\mc I_j^\ell}(y) = D_j^\ell[p_{\mc I_j^q}(x)] &= \min \left\{ C_{2j-1}^{\ell + 1}[k] + C_{2j}^{\ell + 1}[k'] \ : \ k + k' = p_{\mc I_j^\ell}(x) \right\}\\
        & \leq C_{2j-1}^{\ell + 1}[p_{\mc I_{2j-1}^{\ell+1}}(x)] + C_{2j}^{\ell + 1}[p_{\mc I_{2j}^{\ell+1}}(x)] \\
        & = w_{\mc I_{2j -1}^{\ell +1}}(x) + w_{\mc I_{2j}^{\ell +1}}(x) = w_{\mc I_{j}^{\ell}}(x)
    \end{align*}
    where we use the induction hypothesis and the fact that $\mc I_j^{\ell} = \mc I_{2j-1}^{\ell + 1} \cup \mc I_{2j}^{\ell + 1}$ is a partitioning. Recall that we interpret out-of-bound entries of arrays as $+ \infty$.
    Since $x$ is a Pareto optimum of $\mc W_{\mc I}$, we obtain $w_{\mc I_j^\ell}(y) = w_{\mc I_j^\ell}(x)$ by \cref{lem:pareto_maximum_Wextended}, and then $D_j^\ell[p_{\mc I_j^\ell}(x)] = w_{\mc I_j^\ell}(x)$, which implies $C_j^\ell[p_{\mc I_j^q}(x)] = w_{\mc I_j^q}(x)$ as argued above.
\end{proof}

Let us state the following observation about break points of $\mc W_{\mc I}$.

\begin{lemma}\label{lem:break_points_W}
    Let $k$ be a break point of $\mc W_{\mc I}$. Then $[k - \pmax, k)$ contains a break point of $\mc W_{\mc I}$.
\end{lemma}
\begin{proof}
    Let $y \in \{0, 1\}^n$ a Pareto optimum of $\mc W_{\mc I}$ associated to $k$, i.e., $p_{\mc I}(y) = k$ and $\mc W_{\mc I}[p_{\mc I}(y)] = w_{\mc I}(y)$. Consider the vector $y'$ equal to $y$ where we removed one item. 
    The removed item has profit at most $\pmax$ and weight at least $1$. So $p_{\mc I}(y') \geq p_{\mc I}(y) - \pmax = k - \pmax $ and $w_{\mc I}(y') < w_{\mc I}(y)$. In particular, we have $w_{\mc I}(y') \geq \mc W_{\mc I}[k - \pmax]$ by definition of $\mc W_{\mc I}$. Hence $\mc W_{\mc I}[k - \pmax] < w_{\mc I}(y)  = \mc W_{\mc I}[k]$, thus $[k - \pmax, k)$ contains at least one break point of $\mc W_{\mc I}$.
\end{proof}

Finally, we can state the correctness of \cref{alg:knapsack_bmbm_sym} in \cref{lem:bmbm_sym_correctness}. 
Indeed, note that since $\vopt- \pmax \leq \OPT \leq \vopt$, we have $P \subset V$ for 
$P$ defined as in \cref{lem:knapsack_bmbm_sym} and $V$ defined as in \cref{lem:bmbm_sym_correctness}. The success probability can be boosted to any polynomial by repeating \cref{alg:knapsack_bmbm_sym} and taking the entry-wise maximum of computed arrays.

\begin{lemma}[Analog to \cref{lem:bmbm_correctness}]\label{lem:bmbm_sym_correctness}
    Let $T := [t - \sqrt{\Delta_w}, t + \sqrt{\Delta_w}]$ and $P := [\vopt- \sqrt{\Delta_p}, \vopt+ \sqrt{\Delta_p}]$. Then with probability at least $1 - 1/n$ we have $C_1^0[P ; T] = \mc W_{\mc I}[P ; T]$. 
\end{lemma}

\begin{proof}    
    First, observe that $T \subset W^0$ and $P \subset P^0$.
    Let $K^0$ be the set of indices of $C^0_1$, i.e., $K^0 := \{k \ | \ k \in P^0, D_1^0[k] \in W^0\}$. Let $K$ be the interval such that $C_1^0[K] = C_1^0[P ; T]$, i.e., $K := \{k \ |\ k \in P, C_1^0[k] \in T \}$. 
    The entry $C_1^0[k]$ corresponds to the profit of some subset of items of $\mc I$ of profit at least $k$, so clearly $C_1^0[k] \geq \mc W_{\mc I}[k]$ for every $k \in K^0$. We want to show that $C_1^0[k] \leq \mc W_{\mc I}[k]$ for every $k \in K$ with high probability. Since $C_1^0$ and $\mc W_{\mc I}$ are monotone non-decreasing (see \cref{lem:bmbm_sym_monotonicity}), to compare the two sequences it is enough to focus on break points. Recall that $k \in \mathbb N$ is a break point of $\mc W_{\mc I}$ if $\mc W_{\mc I}[k] < \mc W_{\mc I}[k + 1]$, and that for each break point $k$ there exists a Pareto optimum $x \in \{0, 1\}^n$ such that $p_{\mc I}(x)= k$ and $\mc W_{\mc I}[p_{\mc I}(x)] = w_{\mc I}(x)$. 

    To prove the claim, we first need to establish that every $k \in K$ has a break point $k' \geq k$ that is not too far, specifically $k' \leq \vopt + 2\sqrt{\Delta_p}$. By \cref{lem:break_points_W}, $[\vopt + \sqrt{\Delta_p}, \vopt + \sqrt{\Delta_p} + \pmax]$ contains a break point of $\mc W_{\mc I}$. Since $\pmax \leq \sqrt{\Delta_p}$, we deduce that every break point $k \in K$ admits a break point $k'\geq k$ such that $k' \in P' := [\vopt - \sqrt{\Delta_p}, \vopt + 2\sqrt{\Delta_p}]$. Let $K'$ be the interval such that $C_1^0[K'] = C_1^0[P' ; T]$, i.e., $K' := \{k \ |\ k \in P', C_1^0[k] \in T \}$. It remains to prove that $C_1^0[k] \leq \mc W_{\mc I}[k]$ for every break point $k \in K'$. 
    
    Fix a break point $k \in K'$ and let $x \in \{0, 1\}^n$ be the Pareto optimum such that $p_{\mc I}(x) = k$ and $\mc W_{\mc I}[p_{\mc I}(x)] = w_{\mc I}(x)$. Then in particular $w_{\mc I}(x) \in T$ and $p_{\mc I}(x) \in P'$, and thus $|w_{\mc I}(x) - t| \leq 2 \sqrt \Delta_w$ and $|p_{\mc I}(x) - \vopt| \leq \sqrt 2\Delta_p$.
    By \cref{lem:bmbm_sym_properties}, this implies that $w_{\mc I}(x) \in W^0$, $p_{\mc I}(x) \in P^0$ and $C_1^0[p_{\mc I}(x)] = w_{\mc I}(x)$ with probability at least $1 - 1/n^5$. Since $|K'| \leq |T| \leq 2 \wmax \sqrt{n}$, by a union bound over all break points $k \in K'$, we obtain that $C_1^0[P ; T] = \mc P_{\mc I}[P; T]$ with probability at least $1 - 2\wmax \sqrt{n}/ n^5 \geq 1 - 2 \wmax /n^4$. 
    Note that if $n \leq \sqrt{2 \wmax}$ then we can use dynamic programming to compute the weight sequence in time $O(n \cdot \OPT) = O(\OPT \sqrt{\wmax})$ (see \cref{thm:Knapsack_DP_sym}). Hence, we can assume that $2 \wmax \leq n^3$. Thus, with probability at least $1 - 2 \wmax /n^4 \geq 1- 1/n$, we have $C_1^0[P ; T] = \mc P_{\mc I}[P; T]$. 
\end{proof}

\subsubsection{Running time of Algorithm \ref{alg:knapsack_bmbm_sym}}

\begin{lemma}\label{lem:bmbm_sym_running_time_level}
    For a fixed level $\ell \in \{0, \dots, q - 1\}$ and iteration $j \in[2^\ell]$, the computation of $D_j^\ell$ in 
    \cref{alg_line:bmbm_sym_combi_D} takes time $\Ot((\OPT / 2^\ell)^{3/4} \wmax^{1/2}  \pmax^{1/4})$.
\end{lemma}
\begin{proof}
    By \cref{lem:bmbm_sym_monotonicity}, the sequences $C_{2j-1}^{\ell+1}$ and $C_{2j}^{\ell+1}$ are bounded monotone. Additionally, they have length at most $|P^{\ell+1}| = \Ot(\sqrt{\Delta_p / 2^\ell})$ and the values are in a range of length at most $|W^\ell| = \Ot(\sqrt{\Delta_w / 2^\ell})$.
    So their min-plus convolution can be computed using the algorithm of \cref{thm:MPConv} in time $\Ot( (\Delta_p / 2^{\ell})^{1/2}  (\Delta_w / 2^{\ell})^{1/4})$.
    We apply the definitions of $\Delta_w = t \wmax$ and $\Delta_p = \vopt \pmax$, the balancedness assumption $t/ \wmax = \Theta(\OPT/ \pmax)$ and the bound $\vopt = \Theta(\OPT)$, which yields $\Delta_w = O(\OPT \cdot \wmax^2 / \pmax)$, to bound the running time by $\Ot((\OPT / 2^\ell)^{3/4}  \pmax^{1/2}  \wmax^{1/4})$.
\end{proof}

\begin{lemma}\label{lem:bmbm_sym_runtime}
    \cref{alg:knapsack_bmbm_sym} runs in time $\Ot(n + \OPT \sqrt{\wmax})$.
\end{lemma}

\begin{proof}
    We first bound the running time of the base case, i.e., the computations of \crefrange{alg_line:bmbm_sym_start_base}{alg_line:bmbm_sym_leaf_D}.
    For each $j \in[2^q]$, the array $D_j^q$ is obtained by computing the sequence $\mc W_{\mc I_j^q}\left[P^* ; W^*\right]$ where $W^* := \left[0, \frac{t}{2^q} + \sqrt{\frac{\Delta_w}{2^q}} \eta \right]$ and  $P^* := \left[0, \frac{\vopt}{2^q} + \sqrt{\frac{\Delta_p}{2^q}} \eta \right]$.
    Since $\Delta_w =t \wmax$, $\eta = O(\log n)$ and $2^q = \Theta(t / \wmax)$, we can bound $ \frac{t}{2^q} + \sqrt{\frac{\Delta_w}{2^q}} \eta = \Ot(\wmax)$, and analogously $\frac{\vopt}{2^q} + \sqrt{\frac{\Delta_p}{2^q}} \eta = \Ot(\pmax)$. Using \cref{thm:Knapsack_SubsetSum_sym}, we can therefore compute $D_j^q$ in time $\Ot(|\mc I_j^q| + \pmax \sqrt{\wmax})$.
    Hence, the total running time of the base case is:
    \begin{align*}
        \sum_{j=1}^{2^q}  \Ot \left( |\mc I_j^q| + \pmax \cdot  \sqrt{\wmax} \right)  
        & = \Ot\left( n + 2^q \cdot \pmax \cdot \sqrt{\wmax} \right) \\
        & = \Ot\left( n + \vopt\cdot \sqrt{\wmax} \right) = \Ot\left( n + \OPT \cdot \sqrt{\wmax} \right)
    \end{align*}
    where we use $2^q = \Theta(\vopt/ \pmax)$ and $\vopt= \Theta(\OPT)$.
    
    Using \cref{lem:bmbm_sym_running_time_level}, we bound the running time of the combination step, i.e., the computations of \crefrange{alg_line:bmbm_sym_start_combi}{alg_line:bmbm_sym_combi_D}, as follows:
    \begin{align*}
    \sum_{\ell = 0}^{q - 1} \sum_{j = 1}^{2^\ell} \Ot\left( (\OPT / 2^\ell)^{3/4}  \wmax^{1/2}  \pmax^{1/4}\right)
    &= \sum_{\ell = 0}^{q -1} \Ot\left( \OPT^{3/4}  \wmax^{1/2}  (\pmax \cdot 2^\ell)^{1/4}\right)
    \end{align*}
    This is a geometric series, so it is bounded by $\Ot(\OPT^{3/4}  \wmax^{1/2}  (\pmax \cdot 2^{q})^{1/4})$. Since $2^q =\Theta( \vopt/ \pmax) = \Theta(\OPT / \pmax)$ we obtain a running time of $\Ot(\OPT \sqrt{\wmax})$. Hence, in total \cref{alg:knapsack_bmbm_sym} takes time $\Ot(n + \OPT \sqrt{\wmax})$.
\end{proof}

\subsection{\texorpdfstring{$\boldsymbol{\Ot(n  + (n \wmax \pmax)^{1/3} \OPT^{2/3})}$}{O(n + (n wmax pmax)\^1/3 OPT\^2/3}-time Algorithm}\label{sec:bmdp_sym}
Similarly to what is done in \cref{sec:bmdp}, we can modify \cref{alg:knapsack_bmbm_sym} to obtain an algorithm running in time $\Ot(n  + (n \wmax \pmax)^{1/3} \cdot \OPT^{2/3})$, thus proving \cref{lem:knapsack_bmdp_sym}. Together with the reduction of \cref{lem:reduction_balanced_knapsack} this shows \cref{thm:knapsack_bmdp_sym}. 

\knapsackBMDPsym*   

The algorithm of \cref{lem:knapsack_bmdp_sym} uses an analog result to \cref{lem:knapsack_newleaf}: We derive \cref{lem:knapsack_newleaf_sym}  from the idea of He and Xu~\cite{HeXu23} used to obtain a $\Ot(n^{3/2}\pmax)$ time Knapsack algorithm.

\begin{lemma}[Analog of \cref{lem:knapsack_newleaf}]
    \label{lem:knapsack_newleaf_sym}
    For any Knapsack instance $(\mc I, t)$, a nonnegative integer $v \in \mathbb N$, and any $\ell \in \mathbb N, 2 \leq \ell \leq v$, the sequence $\mc W_{\mc I}[v - \ell \dots v + \ell]$ can be computed in time $\Ot(n \sqrt{v \cdot \pmax} + n \ell)$ by a randomized algorithm that is correct with probability at least $1-1/n$.
\end{lemma}

\begin{algorithm}[H]
\SetKw{KwDownTo}{down to}
\caption{The $\Ot(n \sqrt{v \cdot \pmax} + n\ell)$-time algorithm of \cref{lem:knapsack_newleaf}. The input is a Knapsack instance $(\mc I, t)$, a nonnegative integer $v \in \mathbb N$, and a parameter $\ell \in \mathbb N, \ell \leq v$. Accessing a negative index or an uninitialized element in $C_{k-1}$ returns $+\infty$.}\label{alg:knapsack_newleaf_sym}
$\sigma[1 \dots n] \gets $ random permutation of $\{1,\ldots,n\}$\\
$\Delta \gets \ell + \lceil 4 \sqrt{v \cdot \pmax \cdot \log (n\ell)} \rceil$ \\
$C_0[0] \gets 0$ \\
\For{$k = 1, \ldots, n$}{
  \For{$j = \frac{k}{n}v - \Delta, \ldots, \frac{k}{n}v + \Delta$}{
    $C_k[j] \gets \min(C_{k-1}[j], p_{\sigma[k]} + C_{k-1}[j - w_{\sigma[k]}])$ \\
  }
}
\For{$j = v + \ell - 1$ \KwDownTo $v - \ell$}{
  $C_n[j] \gets \min(C_n[j], C_n[j + 1])$ \\
}
\Return{$C_n[(v - \ell \dots v + \ell)]$}\\
\end{algorithm}

\begin{proof}[Proof of \cref{lem:knapsack_newleaf_sym}]
\cref{alg:knapsack_newleaf_sym} clearly runs in $O(n\Delta) = \Ot(n \sqrt{v \cdot \pmax} + n\ell)$ time. Let us focus on proving that with probability at least $1 - 1/n$ it returns a correct answer.

For every target profit value $v' \in \{v - \ell \cdot \eta, \dots, v + \ell\}$ let us fix an optimal solution $x^{(v')} \in \{0,1\}^n$ of weight $\mc W_{\mc I}[v']$ and profit at least $v'$. Note that we have $p_{\mc I}(x^{(v')}) \in [v', v'+ \pmax)$. For every such $v'$ and for every $k \in [n]$, we apply \cref{lem:sample} to $a_i = x^{(v')}_i \cdot p_i$ and $X_i = a_{\sigma[i]}$ with $\delta = 1 / (n^2 \cdot (2 \ell + 1))$, and we conclude that with probability at least $1-\delta$ it holds that
\begin{align*}
p_{\{1,\ldots,k\}}(x^{(v')}) & \in \left[\frac{k}{n}p_{\mc I}(x^{(v')})  \pm \sqrt{v \cdot \pmax \cdot \log(n/\delta)} \right] \\
& \subseteq \left[\frac{k}{n}p_{\mc I}(x^{(v')}) \pm 3 \sqrt{v \cdot \pmax \cdot \log(n\ell)} \right]  \qquad\qquad\ \ (\text{using } n/\delta \leq n^3 \cdot \ell^3) \\
& \subseteq \left[\frac{k}{n}v' \pm 4 \sqrt{v \cdot \pmax \cdot \log(n\ell)} \right] \qquad\qquad (\text{using } p_{\mc I}(x^{(v')}) < v' + \pmax) \\
& \subseteq \left[\frac{k}{n}v \pm \bigl(\ell + 4 \sqrt{v \cdot \pmax \cdot \log(n\ell)} \bigr) \right] = \left[\frac{k}{n}v \pm \Delta \right]. 
\end{align*}
By a union bound, with probability at least $1 - 1/n$ this holds for all such $v'$ and $k$ simultaneously. Let us condition on this event. It follows, by induction on $k$, that
\[C_k[p_{\{1,\ldots,k\}}(x^{(v')})] = w_{\{1,\ldots,k\}}(x^{(v')})\]
for every $v'$ and $k$. In particular for $k=n$ we have $p_{\{1,\ldots,n\}}(x^{(v')}) = p_{\mc I}(x^{(v')}) \geq v'$ and $C_n[v'] \leq C_n[p_{\mc I}(x^{(v')})] = p_{\mc I}(x^{(v')}) = \mc W_{\mc I}[v']$. On the other hand, clearly $C_n[v'] \geqslant \mc W_{\mc I}[v']$, so they are equal, which finishes the proof.
\end{proof}

\begin{proof}[Proof of \cref{lem:knapsack_bmdp_sym}]
Let $c := \min\{1, \tfrac t \wmax \cdot \tfrac \pmax \vopt\}$, and note that $c = \Theta(1)$ by the balancedness assumption and $\vopt = \Theta(\OPT)$. 
If $n \geq c \cdot \vopt \sqrt{\wmax} / \pmax$, then since $\vopt = \Theta(\OPT)$ we have $\Ot(n + \OPT \sqrt{\wmax}) \leq \Ot(n + (n \wmax \pmax)^{1/3} \cdot \OPT^{2/3})$ and thus \cref{lem:knapsack_bmdp_sym} follows from \cref{lem:knapsack_bmbm_sym}. 
Additionally, if $n^3 \leq 2 \wmax$ then 
Bellman's dynamic programming algorithm computes the complete sequence 
$\mc W_{\mc I}$ in time $O(n \cdot \OPT) \le O((n \wmax \pmax)^{1/3} \OPT^{2/3})$. 
In the remainder we can thus assume $n \leq c \cdot \vopt \sqrt{\wmax} / \pmax$ and $2 \wmax \leq n^3$. In this case, we modify the algorithm of \cref{alg:knapsack_bmbm_sym} as follows. 

Let $q$ be the largest integer such that 
$2^q \leq \max\{1, n^{4/3} \cdot (\pmax / \vopt)^{1/3} \cdot \wmax^{-2/3} \}$.
Consider the modification of \cref{alg:knapsack_bmbm_sym} using the new value of $q$ and replacing the computation of $D_j^q$ in \cref{alg_line:bmbm_sym_leaf_D} by the computation of $D_j^q := \mc W_{\mc I_j^q}\left[P^q \right]$ using \cref{alg:knapsack_newleaf_sym} of \cref{lem:knapsack_newleaf_sym}. As a reminder, we defined $\eta := 17 \log n$ and
$P^q := \left[\frac{\vopt}{2^{q}} - \sqrt{\frac{\vopt \cdot \pmax}{2^{q}}} \eta,\ \frac{\vopt}{2^{q}} + \sqrt{\frac{\vopt \cdot \pmax}{2^{q}}} \eta\right]$. Hence, we call \cref{alg:knapsack_newleaf_sym} with the Knapsack instance $(\mc I_j^q, \vopt/ 2^q)$ and parameter $\ell = \sqrt{\vopt \cdot \pmax / 2^q} \cdot \eta$.
So each computation of \cref{alg_line:bmbm_sym_leaf_D} now takes time $\Ot(|\mc I_j^q| \sqrt{\vopt \cdot \pmax / 2^q}) = \Ot(|\mc I_j^q| \sqrt{\OPT \cdot \pmax / 2^q})$. In total, \cref{alg_line:bmbm_sym_leaf_D} takes time
\begin{align*}
    \sum_{j = 1}^{2^q} \Ot\left(|\mc I_j^q|\sqrt{\OPT \cdot \pmax / 2^q} \right) = \Ot\left(n \sqrt{\OPT \cdot \pmax /2^q} \right) \le \Ot\left((n \wmax \pmax)^{1/3} \OPT^{2/3}\right),
\end{align*}
where the last step follows from $\vopt = \Theta(\OPT)$ and the inequality $2^q \ge n^{4/3} \cdot (\pmax / \vopt)^{1/3} \cdot \wmax^{-2/3} / 2$, which holds by our choice of $q$.

If $2^q = 1$, then no combination steps are performed. Otherwise, we have $2^q \le n^{4/3} \cdot (\pmax / \vopt)^{1/3} \cdot \wmax^{-2/3} \leq O(n^{4/3} \cdot (\pmax / \OPT)^{1/3} \cdot \wmax^{-2/3})$, because $\vopt = \Theta(\OPT)$.
In this case, for the combination levels the same analysis as in \cref{lem:bmbm_sym_runtime} shows that the total running time of all combination steps is $\Ot(\OPT^{3/4} \cdot \wmax^{1/2} \cdot (\pmax 2^q)^{1/4}) \le \Ot\left((n \wmax \pmax)^{1/3} \OPT^{2/3}\right)$.

The correctness argument of \cref{alg:knapsack_bmbm_sym} works verbatim because all used inequalities on $2^q$ still hold, specifically we have $1 \le 2^q \le n$, $2^q \le t/\wmax$, and $2^q \le \vopt / \pmax$. We verify these inequalities in the remainder of this proof. 
If $2^q = 1$ then these inequalities are trivially satisfied. Otherwise, we have $1 < 2^q \le n^{4/3} \cdot (\pmax / \vopt)^{1/3} \cdot \wmax^{-2/3}$. 
Then obviously $2^q \ge 1$. 
Since $n \leq c \cdot \vopt \sqrt{\wmax} / \pmax$, by rearranging we have $n^{4/3} \cdot (\pmax / \vopt)^{1/3} \cdot \wmax^{-2/3} \le c^{4/3} \vopt / \pmax$, and thus $2^q \le c^{4/3} \vopt / \pmax \le c \cdot \vopt / \pmax$. Since $c \le 1$, we obtain $2^q \le \vopt / \pmax$. Since $c \le \tfrac t \wmax \cdot \tfrac \pmax \vopt$, we obtain $2^q \le t / \wmax \le n$. 
Finally, the correctness argument of \cref{alg:knapsack_bmbm_sym} additionally uses the bound $n^3 \geq 2 \wmax$, which we can assume as discussed above.
\end{proof}

\section{Rectangular Bounded Monotone Min-Plus Convolution}\label{sec:adapt_chieetal_MPConv}

In this section we prove the following theorem.

\MPConv*

More directly, we distill the following result from the work of Chi, Duan, Xie, and Zhang~\cite{ChiDXZ22_stocs}.

\begin{restatable}[Slight modification of \cite{ChiDXZ22_stocs}]{theorem}{MinConvGen}
    \label{thm:chieetal}
    Min-plus convolution of two monotone non-decreasing or non-decreasing sequences of length $n$ with entries in $\{0,\ldots,M\}$ can be solved in expected time $\Ot(n\sqrt{M} + M)$. 
\end{restatable}

Let us first argue that \cref{thm:chieetal} implies \cref{thm:MPConv}.

\begin{proof}[Proof of \cref{thm:MPConv}]
By the equivalence of min-plus convolution and max-plus convolution, it suffices to design an algorithm for min-plus convolution.

Running the better of the naive $O(n^2)$-time algorithm and Theorem~\ref{thm:chieetal} solves min-plus convolution in expected time $\Ot(\min\{n^2, n \sqrt{M} + M\}) \le \Ot(n \sqrt{M})$.

Then we use the standard conversion of a Las Vegas algorithm to a Monte Carlo algorithm. That is, we run the algorithm with a fixed time budget $C \cdot \log^C (nM) \cdot n \sqrt{M}$ and abort if the algorithm did not finish within this time budget. By Markov's inequality, for a sufficiently large constant $C>0$ the algorithm finishes with probability at least $0.9$. Repeating the resulting algorithm $\log n$ times, at least one run succeeds with probability at least $1 - 1/n$. 
\end{proof}

It remains to prove \cref{thm:chieetal}. In what follows, we first present a proof sketch describing how one can distill this result from~\cite{ChiDXZ22_stocs}. Afterwards, in the remainder of this section we present the proof details.

\begin{proof}[Proof Sketch of \cref{thm:chieetal}]
The theorem statement restricted to monotone \emph{increasing} sequences bounded by $M = \Theta(n)$ is proven in \cite[Theorem 1.5]{ChiDXZ22_stocs}.
We observe that one can obtain the claimed theorem statement by adapting the proof in \cite[Section 4.2]{ChiDXZ22_stocs} as follows.

First, the same proof works verbatim after replacing ``monotone increasing'' by ``monotone non-decreasing'' or ``monotone non-increasing''. (Indeed, monotonicity is only used to argue that the sequence $(\lfloor A[i]/q \rfloor)_{i=1}^n$ can be partitioned into $O(M/q)$ constant intervals, and this property holds no matter whether the sequence $A$ is increasing, non-decreasing, or non-increasing; similarly for $B$.)

Second, in order to generalize from $M = \Theta(n)$ to arbitrary $M$, the same algorithm and correctness proof as in \cite[Section 4.2]{ChiDXZ22_stocs} works verbatim. The running time analysis works after performing the following replacements:
\begin{align*}
    [40n^\alpha, 80n^\alpha] && \to \quad\qquad & [40M^\alpha, 80M^\alpha] \\
    O(n^{1-\alpha}) && \to \quad\qquad & O(M^{1-\alpha}) \\
    \Ot(n^{2-2\alpha}) && \to \quad\qquad & \Ot(M^{2-2\alpha}) \\
    O(n^\alpha / 2^l) && \to \quad\qquad & O(M^\alpha / 2^l) \\
    O(n/2^\ell) && \to \quad\qquad & O(M / 2^\ell) \\
    O(n^{2-\alpha}) && \to \quad\qquad & O(n M^{1-\alpha}) \\
    \Ot(n^{2-\alpha}) && \to \quad\qquad & \Ot(n M^{1-\alpha}) \\
    \Ot(n^{1+\alpha}) && \to \quad\qquad & \Ot(n M^{\alpha}) \\
    \Ot(n^{1+\alpha} + n^{2-\alpha}) && \to \quad\qquad & \Ot(n M^{\alpha} + n M^{1-\alpha} + M^{2-2\alpha}) \\
    \Ot(n^{1.5}) && \to \quad\qquad & \Ot(n \sqrt{M} + M) \\
    O(2^\ell / n^\alpha) && \to \quad\qquad & O(2^\ell / M^\alpha) \\
    O(n^2/2^\ell) && \to \quad\qquad & O(nM / 2^\ell).
\end{align*}
Performing exactly these replacements in \cite[Section 4.2]{ChiDXZ22_stocs} yields an algorithm for min-plus convolution of two monotone sequences with entries in $\{0,\ldots,M\}$ that runs in time $\Ot(n \sqrt{M} + M)$. 
\end{proof}

In the remainder of this section, we give a complete proof of \cref{thm:chieetal} by performing the above described replacements in the proof of \cite[Theorem 1.5]{ChiDXZ22_stocs}. Naturally, our presentation has a large overlap with the proof as given in~\cite{ChiDXZ22_stocs}.

\subsection{Algorithm of \cref{thm:chieetal}}
Let $A[1 \dots n]$ and $B[1 \dots n]$ be two monotone non-decreasing sequences of length $n$ with integer entries in $\{0, \dots, M\}$. 
Let $\alpha \in (0, 1)$ be a constant parameter to be determined later. Sample a prime $p$ uniformly at random in $[M^\alpha, 2 M^\alpha]$.
Without loss of generality we can assume that $n$ is a power of 2. We make the following assumption.

\begin{assumption}\label{ass:minconv}
For every $i \in [n]$, either $(A[i] \bmod p) \leq p/3$
or $A[i] = +\infty$, and $A$ is monotone except for the infinity entries. The number of intervals of infinity in $A$ is at most $O(M^{1-\alpha})$. Similar for $B$.
\end{assumption}

\begin{lemma}\label{lem:minconv_can_assume}
    Let $A$ and $B$ be two monotone sequences of length $n$ with positive integer values bounded by $M$. The computation of $\textup{\textsc{MinConv}}(A,B)$ can be reduced to a constant number of computations of $\textup{\textsc{MinConv}}(A^i,B^i)$ where $A^i$ and $B^i$ satisfy \cref{ass:minconv}.
\end{lemma}

\begin{proof}
We define the sequences $A^0$, $A^1$ and $A^2$ such that for $i \in [n]$:
\begin{itemize}
    \item if $(A[i] \bmod p) \in [0, p/3)$, let $A^0[i] := A[i]$ ; $A^1[i] = A^2[i] = +\infty$ 
    \item if $(A[i] \bmod p) \in [p/3, 2p/3)$, let $A^1[i] = A[i]$ ; $A^0[i] = A^2[i] = +\infty$ 
    \item if $(A[i] \bmod p) \in [2p/3, p)$, let $A^2[i] = A[i]$ ; $A^0[i] = A^1[i] = +\infty$ 
\end{itemize}
and similarly for $B^0$, $B^1$ and $B^2$. Clearly, the number of intervals of infinity in each sequence is bounded by $O(M^{1-\alpha})$.

Furthermore, each pair $(A', B')$ where $A' = A^x - \lceil x p / 3\rceil$ and $B'= B^y - \lceil y p / 3 \rceil$ for $x, y \in \{0, 1, 2\}$ satisfies \cref{ass:minconv}. We compute the element-wise minimum $C' = \min_{x, y \in \{0, 1, 2\}} \minconv{A^x - \lceil x p/ 3 \rceil}{ B^y - \lceil y p / 3 \rceil} + \lceil x p/3 \rceil + \lceil y p/3 \rceil$.
Since elements in $A^x - \lceil x p/ 3 \rceil$ are not smaller than  elements in $A$, and similar for $B^y - \lceil y p/ 3 \rceil$, we have $C'[k] \geq C[k]$. On the other hand, for $i$ such that $C[k] = A[i] + B[{k-i}]$, the elements $A[i]$ and $B[k - i]$ are contained in one of the 9 pairs. So $C'[k] = C[k]$ for every $k$.
\end{proof}

We decompose the sequence $C = \minconv{A}{B}$ into two sequences $\Tilde C$ and $C^*$ such that for any $k \in [2, 2n]$ if $C[k]$ is finite then $\Tilde C[k] :=  \left \lfloor \frac{C[k]}{p} \right \rfloor$ and $C^* := (C[k] \bmod p)$. If $C[k]$ is not finite, then we set $\Tilde C[k] := \infty$ and we don't require anything on $C^*[k]$. By computing $\Tilde C$ and $C^*$, we can retrieve the sequence $C$ by setting $C[k] = p \cdot \Tilde C[k] + C^*[k]$.

\subsubsection{Computing \texorpdfstring{$\boldsymbol{\Tilde C}$}{C~}}
Define the sequence $\Tilde A$ such that for any $i \in [n]$ if $A[i]$ is finite then $\Tilde A[i] :=  \left \lfloor \frac{A[i]}{p} \right \rfloor$, otherwise $\Tilde A[i] := \infty$. Similarly, let $\Tilde B$ be such that for any $i \in [n]$ if $B[i]$ is finite  then $\Tilde B[i] :=  \left \lfloor \frac{B[i]}{p} \right \rfloor$, otherwise $\Tilde B[i] = \infty$.
Let $A[i] + B[j] = C[k]$ be finite. Then by \cref{ass:minconv}, $\lfloor A[i] / p \rfloor + \lfloor B[j] / p \rfloor = \lfloor C[k] / p \rfloor$, i.e. $\Tilde C = \minconv{\Tilde A}{\Tilde B}$.

The finite entries of $\Tilde A$ and $\Tilde B$ are monotone non-decreasing sequences with values bounded by $O(M^{1-\alpha})$. Since there are at most $O(M^{1-\alpha})$ intervals of infinity, we can divide $\Tilde A$ into at most $O(M^{1-\alpha})$ intervals $[i_1, i_2] \subset [n]$ such that for any $i \in [i_1, i_2]$ we have $\Tilde A[i_1] = \Tilde A[i]$. Similarly, divide $\Tilde B$ in at most $O(M^{1-\alpha})$ intervals $[j_1, j_2] \subset [n]$ on which $\Tilde B$ is constant. Then for each pair of intervals $([i_1, i_2], [j_1, j_2])$ of $\Tilde A$ and $\Tilde B$, the sum $\Tilde A[i] + \Tilde B[j]$ is constant for all $i \in [i_1, i_2]$ and $j \in [j_1, j_2]$. Hence, we can compute the sequence $\Tilde C = \minconv{\Tilde A}{\Tilde B}$ as follows.
Initialize $\Tilde C[k] = 0$ for every $k \in [2, 2n]$.
Iterate over every pair of intervals $[i_1, i_2] \subset [n]$ and $[j_1, j_2] \subset [n]$ and let $\Delta = A[i_1] + B[j_1]$. Update the entry $\Tilde C[k]$ for the value $\min \{\Tilde C[k], \Delta\}$ for every $k \in [i_1 + j_1, i_2 + j_2]$. 
The update step can be done in $O(\log n)$ time using a segment tree data structure. In total, there are $O(M^{2-2\alpha})$ pairs of intervals $[i_1, i_2]$ and $[j_1, j_2]$, so the computation of $\Tilde C$ takes $\Tilde O(M^{2 - 2\alpha})$ time. 

\subsubsection{Computing \texorpdfstring{$\boldsymbol{C^*}$}{C*}}

To compute $C^*$, let $h$ be the integer such that $2^{h-1} \leq p < 2^h$. 
Construct for all $\ell \in \{0, 1, \dots, h\}$ the sequences $\Al$ and $\Bl$ such that for $i \in [n]$:
$$
\Al[i] :=  \left \lfloor \frac{A[i] \bmod p}{2^\ell} \right \rfloor \text{\quad  and \quad} \Bl[i] :=  \left \lfloor \frac{B[i] \bmod p}{2^\ell} \right \rfloor.$$
Our goal is to compute a sequence $\Cl$ such that for all $k \in [2, 2n]$ if $C[k]$ is finite then:
\begin{equation}\label{eq:Cl_pty}
\left \lfloor \frac{C[k] \bmod p - 2(2^\ell -1)}{2^\ell} \right \rfloor \leq  \Cl[k] \leq \left \lfloor \frac{C[k] \bmod p + 2(2^\ell -1)}{2^\ell} \right \rfloor.
\end{equation}
Note that if $C[k]$ is infinite then we don't care about the value of $\Cl[k]$.
We compute $\Cl$ by starting with $\ell = h$. Since $p < 2^h$, both $A^{(h)}$ and $B^{(h)}$ are zero sequences, and thus we can set $C^{(h)}$ to a zero sequence, which satisfies the above property. We then recursively use $\CUPL$ to compute $\Cl$. In the end, we set $C^* = C^{(0)}$. Indeed, if $C[k]$ is finite then $C^{(0)}[k] = (C[k] \bmod p)$, as desired. 
Note that, although $\Cl$ is not necessarily the min-plus convolution of $\Al$ and $\Bl$, the following property holds.

\begin{lemma}\label{lem:minconv_filter}
    Fix $\ell \in \{0, 1, \dots, h-1\}$ and suppose that $C^{(\ell)}$ satisfies Property (\ref{eq:Cl_pty}).
    Let $i, k \in [n]$ be such that $A[i]$ and $B[k - i]$ are finite and $A[i] + B[k - i] = C[k]$. 
    Then for every $\ell \in \{0, 1, \dots, h\}$ there exists $b \in [-10, 10]$ such that $\Al[i] + \Bl[k-i] = \Cl[k] + b.$
\end{lemma}
\begin{proof}
    We have
    \begin{align*}
        \Al[i] + \Bl[k-i] - \Cl[k] 
        & = \left\lfloor \frac{A[i] \bmod p}{2^\ell} \right\rfloor + \left\lfloor \frac{B[k - i] \bmod p}{2^\ell} \right\rfloor - \Cl[k] \\
        & \leq \frac{A[i] \bmod p}{2^\ell} + \frac{B[k - i] \bmod p}{2^\ell} - \frac{C[k] \bmod p}{2^\ell} + 3 \\
        & = \frac{(A[i] + B[k - i] - C[k]) \bmod p}{2^\ell} + 3 \\
        & = 3
    \end{align*}
    where the inequality comes from Property (\ref{eq:Cl_pty}) and the second equality follows from \cref{ass:minconv}.
    Similarly,
    \begin{align*}
        \Al[i] + \Bl[k-i] - \Cl[k] 
        & = \left\lfloor \frac{A[i] \bmod p}{2^\ell} \right\rfloor + \left\lfloor \frac{B[k - i] \bmod p}{2^\ell} \right\rfloor - \Cl[k] \\
        & \geq \frac{A[i] \bmod p}{2^\ell} -1 + \frac{B[k - i] \bmod p}{2^\ell} - 1 - \frac{C[k] \bmod p}{2^\ell} - 2 \\
        & = \frac{(A[i] + B[k - i] - C[k]) \bmod p}{2^\ell} - 4 \\
        & = - 4.
    \end{align*}
\end{proof}

We explain the general idea behind the computation of $\Cl$, before presenting the procedure in detail. Suppose that $\CUPL$ was computed in the previous round. For simplicity, let us assume that every entry of $A$ and $B$ is finite. For every $k \in [2, 2n]$, we are interested in the index $q \in [n]$ such that $A[q] + B[k-q] = C[k]$. Naturally, $q$ satisfies $\Tilde A[q] + \Tilde B[k-q] = \Tilde C[k]$, and, by the above \cref{lem:minconv_filter}, it also satisfies $\AUPL[i] + \BUPL[k-i] = \CUPL[k] + b$ for some $b \in [-10, 10]$. Since we computed $\Tilde C$ and $\CUPL$, we can identify the set of indices, among which $q$, that satisfy the above two properties. To speed up the search, instead of considering every index $i \in [n]$, we divide $\Tilde A$, $\Tilde B$, $\AUPL$ and $\BUPL$ into intervals on which the sequences are respectively constant. By taking the cross combinations of those intervals, we obtain $[i_1, i_2] \subset [n]$ and $[j_1, j_2] \subset [n]$ such that for every $i \in [i_1, i_2]$ and $j \in [j_1, j_2]$, the quantities $\Tilde A[i] + \Tilde B[j]$ and $\AUPL[i] + \BUPL[j]$ are respectively constant.  By comparing those terms to $\Tilde C[i + j]$ and $\CUPL[i+j]$ respectively, we can select corresponding values $\Al[i] + \Bl[j]$ that are candidate for the entry $\Cl[i + j]$. We will show that by selecting the minimum value $\Al[i] + \Bl[j]$ among the selected candidates, we obtain $\Cl[k]$ that satisfies Property (\ref{eq:Cl_pty}).
To this end, let us define the notion of \emph{segment}, which reflects the idea of cross combining intervals from $\Tilde A$, $\Tilde B$, $\AUPL$ and $\BUPL$.

\begin{definition}
    Let $\ell \in \{0, 1, \dots, h\}$, $[i_1, i_2] \subset [n]$ and $k \in [2, 2n]$. We call the pair  $([i_1, i_2], k)_\ell$ a \emph{segment with respect to $\ell$} if $[i_1, i_2]$ is an interval of maximum length such that     
    for every $i \in [i_1, i_2]$ the entries $A[i]$, $B[k - i]$ and $C[k]$ are finite with $\Al[i] = \Al[{i_1}]$, $\Bl[{k-i}] = \Bl[{k - i{_1}}]$, $\Tilde A[i] = \Tilde A[i_1]$ and $\Tilde B[k - i] = \Tilde B[k - i_1]$.
\end{definition}

For every $k$, the segment $([i_1, i_2], k)_{\ell}$ with respect to $\ell$ is defined so that, for finite values, the quantities
$\Tilde A[i] + \Tilde B[k-i]$ and $\Al[i] + \Bl[k-i]$ are respectively constant for $i \in [i_1, i_2]$. By comparing  $\Al[i] + \Bl[k-i]$ to $\Cl[k]$, we can use the condition of \cref{lem:minconv_filter} to select candidate segments for $C^{(\ell -1)}$. However, we also clearly have $\Tilde A[i] + \Tilde B[k - i] = \Tilde C[k]$ for pairs $i, k$ such that $A[i] + B[k-i] = C[k]$. Thus, we define the set of \emph{false positive} for a fixed $b \in [-10, 10]$ and $\ell \in \{1, \dots, h \}$ as 
\begin{align*}
    \Tl_b := \Big\{ \text{ segment } ([i_1, i_2], k)_{\ell} \Big.  \Big\vert\ &  \Al[i_1] + \Bl[k-i_1] = \Cl[k] + b  \\
    & \Big. \text{ and } \Tilde A[i_1] + \Tilde B[k - i_1] \neq \Tilde C[k] \Big\}
\end{align*}
which corresponds to the set of segments with respect to $\ell$ that satisfy the condition of \cref{lem:minconv_filter} but are clearly not good candidates for $C^{(\ell - 1)}$.

\paragraph*{Recursive computation of $\boldsymbol{C^*}$.}
We can now describe the computation of $C^*$ in detail.
As already mentioned, since $p < 2^h$, we set $C^{(h)}$ to the zero sequence. Then, the sets $T_{b}^{(h)}$ are empty for all $b \neq 0$ and 
$T_{0}^{(h)}$ is the set of all segments $([i_1, i_2], k)_h$ such that $\Tilde A[i_1] + \Tilde B[k - i_1] \neq \Tilde C[k]$. The latter can be computed in $O(M^{2-2\alpha})$ time by considering intervals on which $\Tilde A$ and $\Tilde B$ are constant, similarly to how we constructed $\Tilde C$.
Next, we iterate over $\ell = h-1, \dots, 0$, and construct both $\Cl$ and $\Tl_b$ for every $b \in [-10, 10]$ using $\CUPL$ and $\TUPL_b$ computed in the previous step.
In the end, we set $C^* = C^{(0)}$. 

\subparagraph*{Computing $\boldsymbol{\Cl}$ from  $\boldsymbol{\CUPL}$ and $\boldsymbol{\TUPL_b}$.}
Construct the following polynomials on variables $x, y, z$:
\[
A^p(x, y, z) = \sum_{i =1}^n x^{\Al[i] - 2\AUPL[i]} \cdot y^{\AUPL[i]} \cdot z^i
\]
\[
B^p(x, y, z) = \sum_{j =1}^n x^{\Bl[j] - 2\BUPL[j]} \cdot y^{\BUPL[j]} \cdot z^j
\]
and compute the polynomial multiplication $C^p(x, y, z) = A^p(x, y, z) \cdot B^p(x, y, z)$ using standard Fast Fourier Transform. Observe that in $A^p(x, y, z)$ and $B^p(x, y, z)$ every $x$-degree is 0 or 1, every $y$-degree is at most $O(M^\alpha / 2^\ell) \leq O(M^\alpha)$ and every $z$-degree is at most $n$. Hence, the computation of $C^p(x, y, z)$ takes time $\Tilde O(nM^{\alpha})$.

Now we use the condition given by \cref{lem:minconv_filter} to filter out segments. For every  offset $b \in [-10, 10]$ and index $k \in [2, 2n]$, enumerate all terms $\lambda x^c y^d z^e$ of $C^p(x, y, z)$ such that $e = k$ and $d = \CUPL[k] + b$ and let $C^p_{k, b}(x)$ be the sum of all such terms $\lambda x^c$.
Construct also the polynomial that contains monomials with false positive coefficients:
\[
R^p_{k, b}(x) = \sum_{\substack{([i_1, i_2], k)_{\ell + 1} \in \TUPL_b \\ i \in [i_1, i_2] }} x^{\Al[i] + \Bl[k-i] - 2 \left(\AUPL[i] + \BUPL[k-i]\right)}.
\]
This way, we can extract all the candidate values for $\Cl[k]$ corresponding to the offset $b$:
$$\Gamma_{k, b} =\left\{ c + 2 \left(\CUPL[k] + b \right) \ \vert\ \lambda x^c \text{ is a term in } (C^p_{k, b}(x) -  R^p_{k, b}(x)) \right\}.$$
Finally, we construct $\Cl[k]$ by ranging over all offsets $b \in [-10, 10]$ and taking the minimum of all candidate values, i.e., we let $\Cl[k] = \min_{b \in [-10, 10]} \min (\Gamma_{k, b} )$. We show in Section \ref{sec:minconv_correctness} that $\Cl$ satisfies the desired Property (\ref{eq:Cl_pty}). 

Note that we can construct $R^p_{k, b}(x)$ in $O(|\TUPL_b|)$ time. Indeed, if $([i_1, i_2], k)_{\ell+1}$ is a segment with respect to $\ell + 1$, then $\AUPL[i]$ is constant on $[i_1, i_2]$ and so $[i_1, i_2]$ can be split into two intervals $[i_1, i_{12}]$ and $[i_{12}, i_2]$ such that on one interval $\Al$ is constant with value $\Al[i] = 2 \AUPL[i_1]$ and on the other $\Al$ is constant with value $\Al[i] = 2 \AUPL[i_1] + 1$. The same can be done for $\BUPL$ and $\Bl$. By taking the 4 cross-combinations of values for $\Al$ and $\Bl$, we obtain 4 segments $([s, f], k)_{\ell}$ with respect to $\ell$ such that $[s, f] \subset [i_1, i_2]$. 
We say that the segment $([s, f], k)_{\ell}$ is \emph{contained} in the segment $([i_1, i_2], k)_{\ell + 1}$. Each such segment $([s, f], k)_{\ell}$ contributes to one term $\lambda x^a$  in $R^p_{k, b}(x)$ where $a \in \{0, 1\}$ and $\lambda = f - s + 1$. Since the breaking points of $[i_1, i_2]$ into at most 4 intervals can be found via binary search, we can compute the terms associated to each segment $([i_1, i_2], k)_{\ell + 1} \in \TUPL_b$ in $\Ot(1)$ time and thus compute $R^p_{k, b}(x)$ in $O(|\TUPL_b|)$ time.
In total, the computation of $\Cl$ from $\CUPL$ and $\TUPL$ takes $\Ot(n M^\alpha + |\TUPL_b|)$.

\subparagraph*{Computing $\boldsymbol{\Tl_b}$ from $\boldsymbol{\TUPL_b}$.}
By the above observation, each segment with respect to $\ell + 1$ can be split into at most 4 segments with respect to $\ell$. We show in \cref{lem:minconv_inclusion_Tbl} that every segment with respect to $\ell$ that is in $\bigcup_{b = -10}^{10} \Tl_b$ is contained in a segment with respect to $\ell +1$ that is in $\bigcup_{b = -10}^{10} \TUPL_b$. This way, to construct $\Tl_b$ we only need to consider the segments of $\bigcup_{b = -10}^{10} \TUPL_b$, split them into segments of $\bigcup_{b = -10}^{10} \Tl_b$ using binary search (see previous paragraph) and assign them to the corresponding $\Tl_b$. Therefore, this phase runs in  $O(|\TUPL_b|)$ time. 

\begin{lemma}\label{lem:minconv_inclusion_Tbl}
    Fix $\ell \in \{0, 1, \dots, h-1\}$. Suppose that Property (\ref{eq:Cl_pty}) holds for $\Cl[k]$ and $\CUPL[k]$ for every $k \in [2, 2n]$. For any segment $([s, f], k)_{\ell} \in \bigcup_{b =-10}^{10} \Tl_b$ there exists a segment $([i_1, i_2], k)_{\ell + 1} \in \bigcup_{b =-10}^{10} \TUPL_b$ with $[s, f] \subset [i_1, i_2]$.  
\end{lemma}
\begin{proof}
    Let us first bound the value of $\Cl[i] - 2 \CUPL[i]$ as follows. Property (\ref{eq:Cl_pty}) gives
    $
    \frac{(C[i] \bmod p)}{2^\ell} - 3 \leq \Cl[i] \leq \frac{(C[i] \bmod p)}{2^\ell} + 2
    $. 
    So we have
    $$
    2 \CUPL[i] - 7 \leq 2 \frac{(C[i] \bmod p)}{2^{\ell+1}} - 3 \leq \Cl[i] \leq 2 \frac{(C[i] \bmod p)}{2^{\ell + 1}} + 2 \leq 2 \CUPL[i] + 8
    $$
    and thus $-7 \leq \Cl[i] - 2 \CUPL[i] \leq 8$. 
    
    Let $([s, f], k)_\ell$ be a segment in $\Tl_b$ for some $b \in [-10, 10]$. We show that $|\AUPL[s] + \BUPL[k - s] - \CUPL[k] | \leq 10$. Indeed,
    \begin{align*}
        \AUPL[s] + \BUPL[k - s] - \CUPL[k] 
        & \geq \frac{\Al[s]}{2} - \frac{1}{2} +  \frac{\Bl[k - s]}{2} - \frac{1}{2} - \frac{\Cl[k]}{2} - \frac{7}{2} \\
        & = \frac{1}{2} \left( \Al[s] + \Bl[k-a] - \Cl[k] \right) - 9/2 \\
        & \geq -10/2 - 9/2 > -10
    \end{align*}
    where the first inequality comes from the bounds on $\Al[i] - 2 \AUPL[i]$, $\Bl[i] - 2 \BUPL[i]$ and $\Cl[i] - 2 \CUPL[i]$; and the last inequality comes from $-10 \leq \Al[s] + \Bl[k-s] - \Cl[k] \leq 10$. Similarly, we derive the following bound
    \begin{align*}
        \AUPL[s] + \BUPL[k - s] - \CUPL[k] 
        & \leq \frac{\Al[s]}{2} +  \frac{\Bl[k - s]}{2} - \frac{\Cl[k]}{2} + \frac{8}{2} \\
        & = \frac{1}{2} \left( \Al[s] + \Bl[k-a] - \Cl[k] \right) + 8/2 \\
        & \leq 10/2 + 8/2 < 10.
    \end{align*}

    Hence, if $([s, f], k)_\ell$ is a segment in $\bigcup_{b =-10}^{10} \Tl_b$, then $|\AUPL[s] + \BUPL[k - s] - \CUPL[k] | \leq 10$. Additionally, we also have $\Tilde A[s] = \Tilde A[i]$ and $\Tilde B[k - s] = \Tilde B [k - i]$ for all $i \in [s, f]$. So $[s, f]$ has to be contained in an interval $[i_1, i_2]$ such that $([i_1, i_2], k)_{\ell +1}$ is a segment of $\TUPL_b$ for some $b \in [-10, 10]$.
\end{proof}

\subsection{Proof that \texorpdfstring{$\boldsymbol \Cl$}{C\^(l)} satisfies Property \texorpdfstring{(\ref{eq:Cl_pty})}{(3)}}\label{sec:minconv_correctness}
We prove by induction that each sequence $\Cl$ satisfies Property (\ref{eq:Cl_pty}). This trivially holds for the base case $\ell = h$ since $C^{(h)}$ is the zero sequence. Fix $\ell \in \{0, 1, \dots , h-1\}$. Let $q$ be such that $A[q] + B[k-q] = C[k]$. Then by induction hypothesis and \cref{lem:minconv_filter}, there exists an offset $b \in [-10, 10]$ such that $\AUPL[q] + \BUPL[k-q] = \CUPL[k] + b$. Then, by the construction of polynomials $A^p(x, y, z)$ and $B^p(x, y, z)$
\begin{align*}
    C^p_{k, b}(x) 
    &= \sum_{i \ : \  \AUPL[i] + \BUPL[k -i] = \CUPL[k] + b} x^{\Al[i] - 2\AUPL[i] + \Bl[k-i] - 2\BUPL[k - i]} \\
    &= \sum_{\substack{i \ : \  \AUPL[i] + \BUPL[k -i] = \CUPL[k] + b \\ \wedge \quad \Tilde A[i] + \Tilde B[k-i] = \Tilde C[k] }}  x^{\Al[i] - 2\AUPL[i] + \Bl[k-i] - 2\BUPL[k - i]} \\
    &\phantom{=} + \sum_{\substack{i \ : \  \AUPL[i] + \BUPL[k -i] = \CUPL[k] + b \\ \wedge \quad \Tilde A[i] + \Tilde B[k-i] \neq \Tilde C[k] }}  x^{\Al[i] - 2\AUPL[i] + \Bl[k-i] - 2\BUPL[k - i]} \\
    &= \sum_{\substack{i \ : \  \AUPL[i] + \BUPL[k -i] = \CUPL[k] + b \\ \wedge \quad \Tilde A[i] + \Tilde B[k-i] = \Tilde C[k] }}  x^{\Al[i] - 2\AUPL[i] + \Bl[k-i] - 2\BUPL[k-i]} \\
    &\phantom{=} + \sum_{\substack{ ([i_1, i_2], k)_{\ell + 1} \in \TUPL_b \\  i \in [i_1, i_2]}}  x^{\Al[i] - 2\AUPL[i] + \Bl[k-i] - 2\BUPL[k - i]} \\
    &= R^p_{k, b} + x^{- 2(\CUPL[k] + b)} \cdot \left( \sum_{\substack{i \ : \  \AUPL[i] + \BUPL[k -i] = \CUPL[k] + b \\ \wedge \quad \Tilde A[i] + \Tilde B[k-i] = \Tilde C[k] }}  x^{\Al[i] + \Bl[k-i]} \right).
\end{align*}
Since $\Tilde A[q] + \Tilde B[k-q] = \Tilde C[k]$ and $\AUPL[q] + \BUPL[k-q] = \CUPL[k] + b$, the term $\Al[q] + \Bl[k-q]$ is contained in the set of candidate values $\Gamma_{k, b}$, and satisfies
\begin{align*}
    \Al[q] + \Bl[k-q] 
    &= \left\lfloor \frac{A[q] \bmod p}{2^\ell} \right\rfloor + \left\lfloor \frac{B[k - q] \bmod p}{2^\ell} \right\rfloor  \\
    &\leq \frac{A[q] + B[k-q] \bmod p}{2^\ell} \\
    &= \frac{C[k] \bmod p}{2^\ell} \leq \left\lfloor \frac{C[k] \bmod p + 2^\ell - 1}{2^\ell}\right\rfloor
\end{align*}
and 
\begin{align*}
    \Al[q] + \Bl[k-q] 
    &= \left\lfloor \frac{A[q] \bmod p}{2^\ell} \right\rfloor + \left\lfloor \frac{B[k - q] \bmod p}{2^\ell} \right\rfloor \\
    &\geq \frac{A[q] + B[k-q] \bmod p - 2(2^\ell - 1)}{2^\ell} \\
    &\geq \left\lfloor \frac{(C[k] \bmod p) - 2(2^\ell - 1)}{2^\ell}\right\rfloor.
\end{align*}
Thus, the term $\Al[q] + \Bl[k-q]$ satisfies Property (\ref{eq:Cl_pty}) and thus gives a valid $\Cl[k]$. Finally, consider any term $\Al[i] + \Bl[k-i]$ satisfying $\Tilde A[i] + \Tilde B[k-i] = \Tilde C[k]$. Since $A[i] + B[k-i] \geq C[k]$, it also satisfies $(A[i] \mod p) + (B[k-i] \mod p) \geq (C[k] \mod p)$ and thus,
\begin{align*}
    A[i] + B[k-i]  
    &=\left\lfloor \frac{A[i] \bmod p}{2^\ell} \right\rfloor + \left\lfloor \frac{B[k - i] \bmod p}{2^\ell} \right\rfloor  \\
    &\geq \frac{(A[i] + B[k-i] \bmod p) - 2(2^\ell - 1)}{2^\ell} \\
    &\geq \left\lfloor \frac{(C[k] \bmod p) - 2(2^\ell - 1)}{2^\ell}\right\rfloor.
\end{align*}
So by choosing the minimum candidate term, we get a valid $\Cl[k]$. This concludes the proof that the construction of $\Cl$ is correct.

\subsection{Running time}
Let us bound the number of segments with respect to $\ell$, for any $\ell \in \{0, 1, \dots, h\}$. 
Observe that finite entries in $\Al$, $\Bl$ and $\Cl$ are bounded by $O(M^\alpha / 2^\ell)$. 
Furthermore, since the finite entries $A$ and $B$ are monotone non-decreasing and $A$ and $B$ contain at most $O(M^{1-\alpha})$ intervals of infinity (by \cref{ass:minconv}), $[n]$ can be divided into at most $O(M/2^\ell)$ intervals in which $\Al$ and $\Tilde A$ are constant. Similarly, $[n]$ can be divided into at most $O(M / 2^\ell)$ intervals in which $\Bl$ and $\Tilde B$ are constant. 
Hence, there are $O(n \cdot M / 2^\ell)$ segments with respect to $\ell$. 

We now analyze the running time of the described algorithm. Computing the sequence $\Tilde C$ takes time $\Ot(M^{2-2\alpha})$. 
Since $p < 2^h$,  there are at most  $O(n M^{1 - \alpha})$ segments with respect to $h$. So computing $C^{(h)}$ and $T^{(h)}_b$ for all $b \in [-10, 10]$ takes time $O(n M^{1 - \alpha})$. In each consecutive step, computing $C^p(x, y, z)$ takes time $\Ot(n M^\alpha)$; computing $\Cl$ from $C^p(x, y, z)$ and $\TUPL_b$ takes time $\Ot(n M^\alpha + |\TUPL_b|)$; and computing $\Tl_b$ from $\TUPL_b$ takes time $O(|\TUPL_b|)$. We show in \cref{lem:minconv_size_Tbl}, that the expected number of segments in $\TUPL_b$ is bounded by $\Ot(n M^{1 - \alpha})$, for any $\ell \in \{0, 1, \dots, h\}$. Therefore, the total expected running time is $\Ot(M^{2 - 2\alpha} + n M^{1-\alpha} + n M^\alpha)$. By setting $\alpha = 1/2$ we obtain the claimed expected running time $\Ot(n \sqrt{M} + M)$. 

\begin{lemma}\label{lem:minconv_size_Tbl}
    For every $\ell \in \{0,1, \dots, h\}$ and $b \in [-10, 10]$, the expected size of $\Tl_b$ is $\Tilde O(nM^{1 - \alpha})$.
 \end{lemma}
 \begin{proof}
    When $2^\ell \geq p /100$, the number of segments with respect to $\ell$ is  $O(|T_b^{(h)}|) \leq O(n M^{1 - \alpha})$. Assume that $2^\ell < p/100$. 
    Consider a segment $([i_1, i_2], k)_\ell$ such that $\Tilde A[i_1] + \Tilde B[k - i_1] \neq \Tilde C[k]$. To bound the expected size of $\Tl_b$, we want to bound the probability that $([i_1, i_2], k)_\ell$ is in $\Tl_b$. That is the case if
    \[
    \left \lfloor \frac{A[i_1] \bmod p}{2^\ell} \right \rfloor +  \left \lfloor \frac{B[k - i_1] \bmod p}{2^\ell} \right \rfloor = \Cl[k] + b
    \]
    which implies by Property (\ref{eq:Cl_pty}) that
    \[
    -4 \leq  \frac{A[i_1] \bmod p}{2^\ell}  + \frac{B[k - i_1] \bmod p}{2^\ell} - \frac{C[k] \bmod p}{2^\ell} - b \leq 4.
    \]
    Let $q \in [n]$ be such that $A[q] + B[k - q] = C[k]$. Then $([i_1, i_2], k)_\ell \in \Tl_b$ if $(A[i_1] + B[k - i_1] - A[q] - B[k - q] \bmod p)$ is contained in $\in [2^\ell(b-4), 2^\ell(b+4)]$. For each possible remainder $r \in [2^\ell(b-4), 2^\ell(b+4)]$, we have $|r| \leq 14 \cdot 2^\ell < 14 \cdot p/100 < p/6$.
    Since $(C[k] \bmod p) < 2p/3$ by \cref{ass:minconv}, and $\Tilde A[i_1] + \Tilde B[k - i_1] \neq \Tilde C[k]$, we also have $\left|A[i_1] + B[k - i_1] - C[k]\right| \geq p/3$. It follows that $\left| A[i_1] + B[k - i_1] -A[q] - B[k - q] - r \right|$ is a positive number bounded by $O(M)$. So the probability that $p \in [M^\alpha, 2M^\alpha]$ divides this number is the ratio between the number of prime divisors of $\left| A[i_1] + B[k - i_1] -A[{q}] - B[k - q] - r \right|$ and the number of primes in $[M^\alpha, 2M^\alpha]$. Any positive number bounded by $O(M)$ has at most $O(\log M)$ prime divisors and by the Prime Number Theorem in \cite{jameson2003prime} there are at least $\Omega(M^\alpha /\log M)$ primes in the interval $[M^\alpha, 2M^\alpha]$. Finally, since there are at most $O(2^\ell)$ remainders $r$, the probability that a given segment with respect to $\ell$ is in $\Tl_b$ is $\Ot(2^\ell / M^{\alpha})$. There are $O(n \cdot M / 2^\ell)$ segments with respect to $\ell$, so in expectation $\Tl_b$ contains $\Tilde O(n M^{1 - \alpha})$ segments.
 \end{proof}

\section{Using rectangular convolution in Bringmann--Cassis algorithm}
\label{sec:adapt_subsetsum_knapsack}
In this section, we change the Knapsack algorithm in \cite{BringmannC22} so that it uses rectangular monotone max-plus convolution instead of square monotone max-plus convolution. This allows us to get a more precise running time of the algorithm as stated in \cref{thm:Knapsack_SubsetSum_rectangular}, and proves \cref{thm:Knapsack_SubsetSum,thm:Knapsack_SubsetSum_sym} used in \cref{sec:knapsack_algos,sec:bmbm_sym} respectively.
The weight sequence $\mc W_{\mc I}$ is defined in \cref{sec:sym_knapsack}.

\begin{theorem}[{Adapted from \cite[Theorem 19]{BringmannC22}}]\label{thm:Knapsack_SubsetSum_rectangular}
    Suppose that the bounded monotone max-plus convolution of two sequences of length $n$ with non-negative values bounded by $M$ can be computed in time $T(n, M)$.
    Assume the following niceness properties of $T(n, M)$: 
    \[T(n, M) \cdot k \leq O(T(n \cdot k, M \cdot k)) \quad \text{and} \quad T(\Ot(n), \Ot(M)) \leq \Ot(T(n, M)).\] 
    Let $(\mc I, t)$ be a Knapsack instance, $v \in \mathbb N$ and let $n = |\mc I|$. Then we can compute the sequence $\mc P_{\mc I}[0 \dots t ; 0 \dots v]$ in time $\Ot(n + T(t, v))$. We can also compute the sequence $\mc W_{\mc I}[0 \dots v ; 0 \dots t]$ in time $\Ot(n + T(v, t))$.
\end{theorem}

Note that we state the running time in \cref{thm:Knapsack_SubsetSum_rectangular} as a function of the running time of rectangular monotone max-plus convolution $T(n, M)$. In \cref{thm:MPConv} we showed that $T(n, M) = \Ot(n\sqrt{M})$. Hence, \cref{thm:MPConv,thm:Knapsack_SubsetSum_rectangular} imply that the sequences $\mc P_{\mc I}[0 \dots t ; 0 \dots v]$ and $\mc W_{\mc I}[0 \dots v ; 0 \dots t]$ can be computed in time $\Ot(n + t \sqrt{v})$ and $\Ot(n + v \sqrt{t})$, respectively. This proves \cref{thm:Knapsack_SubsetSum,thm:Knapsack_SubsetSum_sym}.

\begin{proof}[Proof of \cref{thm:Knapsack_SubsetSum,thm:Knapsack_SubsetSum_sym}.
]
    Combine \cref{thm:MPConv} and \cref{thm:Knapsack_SubsetSum_rectangular}.
\end{proof}

To prove \cref{thm:Knapsack_SubsetSum_rectangular}, we first focus on the computation of $\mc P_{\mc I}[0 \dots t ; 0 \dots v]$, before discussing the analogous $\mc W_{\mc I}[0 \dots v ; 0 \dots t]$.

\subsection{Computing \texorpdfstring{$\boldsymbol{\mc P_{\mc I}[0 \dots t ; 0 \dots v]}$}{PI[0 ... t;0 ... v]}}

The algorithm of \cref{thm:Knapsack_SubsetSum_rectangular} is nearly identical to the algorithm of \cite[Theorem 19]{BringmannC22}. The main difference is the use of \emph{rectangular} bounded monotone max-plus convolution instead of \emph{square} bounded monotone max-plus convolution. In the following description of this algorithm, we omit details of the correctness argument and focus on the analysis of the running time. 

The idea of the algorithm is to partition the item set $\mc I$ into groups $G_{(a, b)} \subset \mc I$ such that all items $i \in \mc I$ with $w_i \in [2^{a - 1}, 2^a)$ and $p_i \in [2^{b - 1}, 2^b)$ are in group $G_{(a, b)}$. This partitioning step takes linear time.
We explain below how to compute the sequences $\mc P_{G_{(a,b)}}[0 \dots t ; 0 \dots v]$ for each group $G_{(a, b)}$ in time $\Ot(T(t, v))$. 
Those sequences are then combined using max-plus convolutions to obtain $\mc P_{\mc I}[0 \dots t ; 0 \dots v]$. 
Note that we can assume without loss of generality that every item has profit at most $v$. So the number of groups is $O(\log t \log v)$, and since each computed sequence has length $t$ with entries bounded by $v$, the combination step takes time $\Ot(T(t, v))$.

\subparagraph*{Computing $\boldsymbol{\mc P_{G_{(a, b)}}[0 \dots t ; 0 \dots v]}$.}
Fix a group $G := G_{(a, b)}$ consisting of items with weight $w_i \in [2^{a - 1}, 2^a)$ and profit $p_i \in [2^{b - 1}, 2^b)$. Note that
any $x \in \{0, 1\}^n$ with $w_{G}(x) \leq t$ and $p_{G}(x) \leq v$ selects at most $z := \lceil \min \{t / 2^{a - 1}, v / 2^{b - 1}\} \rceil$ items from $G$.
We randomly split the items in $G$ into $z$ subgroups $G_1, G_2, \dots, G_{z}$. Then any fixed $x \in \{0, 1\}^n$ selects at most $u := O(\log z)$ items in each subgroup $G_i$ with probability at least $1 - 1/\textup{poly}(z)$. 
In particular, since items in $G_i$ have weight at most $2^a$ and profit at most $2^b$, this implies $w_{G_i}(x) \leq 2^a \cdot u$ and $p_{G_i}(x) \leq 2^b \cdot u$ with probability at least $1 - 1/\textup{poly}(z)$. 
Define the sequence $\mc P_{\mc J}^u[\cdot]$ for any subset $\mc J \subset \mc I$ such that for any $k \in \mathbb N$
$$
\mc P_{\mc J}^u[k] = \max \{p_{\mc J}(x) \ | \ x \in \{0, 1\}^n, w_{\mc J}(x) \leq k, \sum_{i \in \mc J}x_i \leq u\}.
$$
Then to compute $\mc P_{G}[0 \dots t ; 0 \dots v]$, we first compute the sequences $\mc P_{G_i}^u[2^a  u ; 2^b  u]$ for every $i \in [z]$ and then combine them via max-plus convolution. We show that both steps take time $\Ot(T(t, v))$. This will imply that we can compute $\mc P_{G_{(a, b)}}[0 \dots t ; 0 \dots v]$ in time $\Ot(T(t, v))$ for every group $G_{(a, b)}$.

\subparagraph*{Computing $\boldsymbol{\mc P_{G_i}^u[0 \dots 2^a u ; 0 \dots 2^b u]}$.}
We use a color coding technique. Specifically, randomly split the items of $G_i$ into $u^2$ buckets $A_1, A_2, \dots, A_{u^2}$. We construct the sequence $\mc P_{A_i}^1[2^a ; 2^b]$ for every $i \in [u^2]$ in $\Ot(n)$ time and then compute their max-plus convolution. Since the sequences are monotone non-decreasing, of length at most $2^a$ and with entries at most $2^b$, computing the max-plus convolution of the $u^2$ sequences takes time $O(T(2^a \cdot u^2, 2^b \cdot u^2) \cdot u^2)$.

By the birthday paradox, any fixed $x \in \{0, 1\}^n$ with constant probability is scattered among the buckets, i.e., each bucket $A_i$ contains at most 1 item selected by $x$. In particular, for every entry $\mc P_{G_i}^u[j]$ the corresponding optimal solution is scattered with constant probability. Hence, every entry of the computed array has the correct value $\mc P_{G_i}^u[j]$ with constant probability. To boost the success probability to $1 - 1/\textup{poly}(z)$, we can repeat this process $O(\log z)$ times and take the entry-wise maximum over all repetitions. 

Finally, by a union bound over all subgroups $G_1, G_2, \dots, G_z$, for a fixed vector $x$ the entries $\mc P_{G_1}^u[w_{a_1}(x)], \mc P_{G_2}^u[w_{a_2}(x)], \dots, \mc P_{G_z}^u[w_{a_1}(x)]$ corresponding to the partition of $x$ among the $z$ subgroups are correctly computed with probability at least $1 - 1/\textup{poly}(z)$.
In total, computing the sequences $\mc P_{G_i}^u[0 \dots 2^a  u ; 0 \dots 2^b u]$ for all $i \in [z]$ takes time 
$$O(T(2^a \cdot u^2, 2^b \cdot u^2) \cdot u^2 \cdot z) \leq \Ot(T(2^a \cdot z, 2^b \cdot z)) \leq \Ot(T(t, v)) 
$$
where we use the fact that $u = O(\log z) = \Ot(1)$ and the niceness assumptions on $T(n, M)$.

\subparagraph*{Merging the sequences $\boldsymbol{\mc P_{G_i}^u[0 \dots 2^a u ; 0 \dots 2^b u]}$.}
Once all sequences $\mc P_{G_i}^u[0 \dots 2^a u ; 0 \dots 2^b u]$
for $i \in [z]$ are computed, we merge them in a tree-like fashion using max-plus convolution, similarly to the merging step of \cref{alg:knapsack_bmbm}. Since we merge $z$ sequences, there are $\lceil \log z \rceil$ levels of computations. At level $\ell$ we compute the max-plus convolution of $2^\ell$ monotone non-decreasing sequences of length $O(2^a \cdot u \cdot z / 2^\ell)$ and with entries bounded by $O(2^b \cdot u \cdot z / 2^\ell)$. Hence, to total merging step takes time
$$
\sum_{\ell = 0}^{\lceil \log z \rceil} T\left(2^a \cdot u \cdot \frac{z}{2^\ell}, 2^b \cdot u \cdot \frac{z}{2^\ell}\right) \cdot 2^\ell 
\leq O \left( \sum_{\ell = 0}^{\lceil \log z \rceil} T\left(2^a \cdot u \cdot z, 2^b \cdot u \cdot z\right)  \right)
\leq \Ot\left(T(2^a \cdot z, 2^b \cdot z)\right)
$$
where we use again $u = \Ot(1)$ and the niceness assumption on $T(n, M)$.
Finally, by the definition of $z$, the above running time is $\Ot(T(t, v))$. 

We merge the arrays $\mc P_{G_{(a, b)}}[0 \dots t, 0 \dots, v]$ by taking their max-plus convolution and obtain $\mc P_{\mc I}[0 \dots t ; 0 \dots v]$. Since the number of groups is $O(\log t \log v)$, and since each computed sequence has length $t$ with entries bounded by $v$, the combination step takes time $\Ot(T(t, v))$.

\subsection{Computing \texorpdfstring{$\boldsymbol{\mc W_{\mc I}[0 \dots v ; 0 \dots t]}$}{WI[0 ... v;0 ... t]}}
Note that since min-plus convolution is equivalent to max-plus convolution, the min-plus convolution on two sequences of length $n$ with non-negative integer entries bounded by $M$ can be computed in time $T(n, M)$.
The construction of $\mc W_{\mc I}[0 \dots v ; 0 \dots t]$ is nearly identical to the above describe method for $\mc P_{\mc I}[0 \dots t ; 0 \dots v]$, except for two notable modifications. Naturally, we use min-plus convolutions instead of max-plus convolutions and in the base case we compute subarrays of $\mc W_{\mc J}^u[\cdot]$ instead of $\mc P_{\mc J}^u[\cdot]$, which is defined for any subset $\mc J \subset \mc I$ and $u \in \mathbb N$ such that for any $k \in \mathbb N$
\[
\mc W_{\mc J}^u[k] = \min \{w_{\mc J}(x) \ | \ x \in \{0, 1\}^n, p_{\mc J}(x) \geq k, \sum_{i \in J}x_i \leq u \}
\]
In a nutshell, the algorithm does the following steps:
\begin{enumerate}
    \item Partition the item set $\mc I$ into $O(\log t \log v)$ groups $G_{(a, b)} \subset \mc I$ as defined above. This takes time $O(n)$.
    \item For each group $G := G_{(a, b)}$:
    \begin{enumerate}
        \item Randomly split $G$ into $z := \lceil \min \{t/ 2^a, (v + \pmax)/2^b \}\rceil$ subgroups $G_i \subset G$ for $i \in [z]$. Let $u = O(\log z)$.
        \item For each $G_i$, compute the array $\mc W_{G_i}^u[0 \dots 2^b u ; 0 \dots 2^a u]$ by running the color coding method $O(\log z)$ times and taking the entry-wise minimum. A single color coding method takes time $T(2^b \cdot u^2 ; 2^a \cdot u^2 ) \cdot u^2$. Hence, the computation of $\mc W_{G_i}^u[0 \dots 2^b  u ; 0 \dots 2^a u]$ for every $i \in [z]$ takes time 
        $$
        O(T(2^b \cdot u^2, 2^a \cdot u^2) \cdot u^2 \cdot z) = \Ot(T(2^b \cdot z, 2^a \cdot z)) = \Ot(T(v, t))
        $$
        by definitions of $k$ and $z$, and the niceness assumptions on $T(n, M)$.
        \item Compute the min-plus convolutions of all arrays $\mc W_{G_i}^u[0 \dots 2^a  u ; 0 \dots 2^b u]$ for $i \in [z]$ to obtain $\mc W_{G}[0 \dots v ; 0 \dots t]$. This takes time $\Ot(T(v ; t))$.
    \end{enumerate}
    Compute the min-plus convolution of all arrays $\mc W_{G_{(a, b)}}[0 \dots v ; 0 \dots t]$ to obtain $\mc W_{\mc I}[0 \dots v ; 0 \dots t]$. Since there are $O(\log t \log v)$ many $G_{(a, b)}$ groups, this takes time $\Ot(T(v, t))$. 
\end{enumerate}


\end{document}